\documentclass{lmcs}
\pdfoutput=1

\usepackage{lastpage}
\lmcsdoi{22}{1}{21}
\lmcsheading{}{\pageref{LastPage}}{}{}%
{Aug.~09,~2024}{Mar.~10,~2026}{}

\pdfoutput=1 

\usepackage[utf8]{inputenc}
\usepackage[T1]{fontenc}
\usepackage{amsthm}

\usepackage{ifthen}
\usepackage{stmaryrd}
\usepackage{dsfont}
\usepackage{todonotes}
\usepackage{amsmath,amssymb,amsfonts,amsthm}

\usepackage{hyperref}
\usepackage{cleveref}
\crefname{thm}{Theorem}{Theorems}
\crefname{cor}{Corollary}{Corollaries}
\crefname{lem}{Lemma}{Lemmas}
\crefname{prop}{Proposition}{Propositions}
\crefname{rem}{Remark}{Remarks}
\crefname{exa}{Example}{Examples}
\crefname{defi}{Definition}{Definitions}
\crefname{conj}{Conjecture}{Conjectures}
\crefname{step}{Requirement}{Requirements}
\crefname{propC}{Proposition}{Propositions}

\usepackage{mathtools}
\usepackage{graphicx}
\usepackage{textcomp}
\usepackage{xcolor}
\usepackage{thm-restate}
\usepackage{enumitem}

\usepackage{tikz}
\usetikzlibrary{decorations.pathreplacing}
\usetikzlibrary{calc}
\usetikzlibrary{cd}

\usepackage{algorithm}
\usepackage[noend]{algpseudocode}

\algnewcommand\algorithmicswitch{\textbf{switch}}
\algnewcommand\algorithmiccase{\textbf{case}}
\algnewcommand\algorithmicassert{\texttt{assert}}
\algnewcommand\Assert[1]{\State \algorithmicassert(#1)}%

\algdef{SE}[SWITCH]{Switch}{EndSwitch}[1]{\algorithmicswitch\ #1\ \algorithmicdo}{\algorithmicend\ \algorithmicswitch}%
\algdef{SE}[CASE]{Case}{EndCase}[1]{\algorithmiccase\ #1 \textbf{:}}{\algorithmicend\ \algorithmiccase}%
\algtext*{EndSwitch}%
\algtext*{EndCase}%

\def\BibTeX{{\rm B\kern-.05em{\sc i\kern-.025em b}\kern-.08em
    T\kern-.1667em\lower.7ex\hbox{E}\kern-.125emX}}

\newcommand{\Coloneqq}{\mathrel{\vcenter{\hbox{$:$}}{\coloneqq}}}

\newcommand{\smalldots}{{...}}

\newcommand{\NN}{\mathbb{N}}
\newcommand{\ZZ}{\mathbb{Z}}
\newcommand{\QQ}{\mathbb{Q}}
\newcommand{\RR}{\mathbb{R}}
\newcommand{\BB}{\mathbb{B}}
\newcommand{\VV}{\mathbb{V}}

\newcommand{\colonsub}{\mathbin{:\subseteq}}
\newcommand{\ar}{\textit{ar}}

\newcommand{\ccor}{\mathbin{\curlyvee}}
\newcommand{\cand}{\mathbin{\curlywedge}}
\newcommand{\cnot}{\mathbin{\smallfrown}}
\DeclareMathOperator{\Span}{span}
\DeclareMathOperator{\id}{id}
\DeclareMathOperator{\lcm}{lcm}
\DeclareMathOperator{\indic}{\mathds{1}}
\DeclareMathOperator{\powerset}{\mathcal{P}}

\DeclareMathOperator{\mcA}{\mathcal{A}}
\DeclareMathOperator{\mcB}{\mathcal{B}}
\DeclareMathOperator{\mcD}{\mathcal{D}}
\DeclareMathOperator{\mcF}{\mathcal{F}}
\DeclareMathOperator{\mcG}{\mathcal{G}}

\DeclareMathOperator{\mcM}{\mathcal{M}}

\DeclareMathOperator{\mcR}{\mathcal{R}}

\DeclareMathOperator{\mcT}{\mathcal{T}}
\DeclareMathOperator{\dom}{dom}
\DeclareMathOperator{\proj}{\pi}
\DeclareMathOperator{\unproj}{\proj^{\forall}}
\DeclareMathOperator{\dotproj}{{\dot\proj}}
\DeclareMathOperator{\dotunproj}{{\dotproj}^{\forall}}
\newcommand{\dotunprojrel}[1]{{\dotproj}^{\forall}_{#1}}
\newcommand{\unprojrel}[1]{{\proj}^{\forall}_{#1}}
\DeclareMathOperator{\len}{len}

\DeclareMathOperator{\seq}{seq}

\DeclareMathOperator{\un}{un}
\DeclareMathOperator{\dnf}{dfnf}
\DeclareMathOperator{\depth}{dep}

\DeclareMathOperator{\AS}{AS}

\DeclareMathOperator{\SL}{SL}

\DeclareMathOperator{\FO}{FO}
\DeclareMathOperator{\AC}{CQ}

\DeclarePairedDelimiter{\set}{\{}{\}}
\DeclarePairedDelimiter{\bracket}{[}{]}

\DeclareMathOperator{\Lin}{Lin} 
\DeclarePairedDelimiter{\Angle}{\langle}{\rangle}

\newcommand{\parone}{\mathbf{1}}

\newcommand{\stdrepr}[1]{\nu_{#1}}

\renewcommand{\vec}{\mathbf}

\renewcommand{\phi}{\varphi}

\makeatletter
\newcommand\incircbin
{%
  \mathpalette\@incircbin
}
\newcommand\@incircbin[2]
{%
  \mathbin%
  {%
    \ooalign{\hidewidth$#1#2$\hidewidth\crcr$#1\bigcirc$}%
  }%
}
\newcommand{\oland}{\incircbin{\land}}
\makeatother

\DeclareMathOperator{\domain}{D}

\newcommand{\compl}[1]{#1^c}

\newcommand{\poly}{{\rm poly}}

\newcommand{\slice}[2]{#1{\upharpoonright}_{#2}}

\newcommand{\univ}[1]{\pi^{\forall}_{#1}}

\newcommand{\indicator}[2]{\indic_{#1}\bracket*{#2}}

\newcommand{\defeq}{\coloneqq}

\newcommand{\abs}[1]{{|#1|}}

\newcommand{\SDF}{{\rm SDF}}
\newcommand{\NFDNF}{{\rm DNF}_{+}}
\newcommand{\emp}{()}

\newcommand{\upscale}[1]{#1\raisebox{1.8pt}{$\uparrow$}}

\newcommand{\sem}[1]{{[\![#1]\!]}}

\newcommand\numberthis{\addtocounter{equation}{1}\tag{\theequation}}

\begin{document}

\title[On PTime Decidability of $k$-Negations Fragments of FO Theories]{On Polynomial-Time Decidability of $k$-Negations Fragments of First-Order Theories}

\author[C.~Haase]{Christoph Haase\lmcsorcid{0000-0002-5452-936X}}[a]
\author[A.~Mansutti]{Alessio Mansutti\lmcsorcid{0000-0002-1104-7299}}[b]
\author[A.~Pouly]{Amaury Pouly\lmcsorcid{0000-0002-2549-951X}}[c]

\address{Department of Computer Science, University of Oxford, UK}
\address{IMDEA Software Institute, Spain}
\address{CNRS, Université Paris Cité, IRIF, France}

\begin{abstract}
  This paper introduces a generic framework that provides sufficient conditions
  for guaranteeing polynomial-time decidability of fixed-negation fragments of
  first-order theories that adhere to certain fixed-parameter tractability
  requirements. It enables deciding sentences of such theories with arbitrary
  existential quantification, conjunction and a fixed number of negation symbols
  in polynomial time. 
  
  It was recently shown by Nguyen and Pak \emph{[SIAM J.\ Comput.  51(2): 1--31
  (2022)]} that an even more restricted such fragment of Presburger arithmetic,
  the first-order theory of the structure $(\ZZ,0,1,+,\leq)$, is NP-hard. In
  contrast, by application of our framework, we show that the fixed negation
  fragment of weak Presburger arithmetic, which drops the order relation~$\leq$
  from Presburger arithmetic in favor of the equality relation~$=$, is
  decidable in polynomial time.
  We give two further examples of instantiations of our framework, showing
  polynomial-time decidability of the fixed negation fragments of weak
  linear real arithmetic (the first-order theory of the structure
  $(\RR,0,1,+,=)$) and of the restriction of Presburger arithmetic in which each
  inequality contains at most one variable.
\end{abstract}

\maketitle

\section{Introduction}
\label{section:introduction}

It is well-known that even the simplest first-order theories are computationally
difficult to decide~\cite{Gradel91}. In particular, it follows from a result of
Stockmeyer that every theory with a non-trivial predicate such as equality is
PSPACE-hard to decide~\cite{Stockm76}. Even when restricting to existential
fragments or fragments with a fixed number of quantifier alternations, deciding
such fragments is NP-hard at best. There are two further kinds of restrictions
that may lead to tractability. First, restricting the Boolean structure of the
matrix of formulae in prenex form yields tractable fragments of, e.g., the
Boolean satisfiability problem. For instance, the Horn and XOR-fragments of
propositional logic are decidable in polynomial time, and this even applies to
quantified Boolean Horn formulae, see e.g.~\cite{Chen09}. Second, restricting
the number of variables can also lead to tractable fragments of a first-order
theory, especially for structures over infinite domains such as Presburger
arithmetic, the first-order theory of the structure $(\ZZ,0,1,+,\leq)$. While
the existential fragment of Presburger arithmetic is NP-complete in
general~\cite{BT76,VonZur78}, it becomes polynomial-time decidable when
additionally fixing the number of variables~\cite{Sca84}; this is a consequence
of polynomial-time decidability of integer programming in fixed
dimension~\cite{Lenstra83}. Already when moving to an $\exists \forall$
quantifier prefix, Presburger arithmetic becomes NP-hard~\cite{Schoning97}. On
the first sight, this result seems to preclude any possibility of further
restrictions that may lead to tractable fragments of Presburger arithmetic.
However, another tractable fragment was identified in the context of
investigating the complexity of the classical \emph{Frobenius problem}. Given
$a_1, \ldots, a_n\in \NN$, this problem asks to determine the largest integer
that cannot be obtained as a non-negative linear combination of the $a_i$, which
is called the \emph{Frobenius number}. For $n>0$ fixed, deciding whether the
Frobenius number exceeds a given threshold can be reduced to the so-called
\emph{short fragment} of Presburger arithmetic, a highly restricted fragment in
which everything, the number of atomic formulae and the number of variables (and
\emph{a fortiori} the number of quantifier alternations), is fixed --- except
for the coefficients of variables appearing in linear terms of atomic
inequalities. Kannan~\cite{Kannan90} showed that the
$\forall^k\exists^\ell$-fragment of short Presburger arithmetic is decidable in
polynomial time for all fixed $k,\ell$, which implies that the decision version
of the Frobenius problem is in polynomial time for fixed $n$. However, in a
recent breakthrough, Nguyen and Pak showed that there are fixed $k,\ell,m$ such
that the $\exists^k \forall^\ell \exists^m$-fragment of short Presburger
arithmetic is NP-hard, and by adding further (fixed) quantifier alternations the
logic climbs the polynomial hierarchy~\cite{NguyenP22}.

The main contribution of this paper is to develop an algorithmic framework that
enables us to show that \emph{fixed negation fragments} of certain first-order
theories are decidable in polynomial time. Formulae in this fragment are
generated by the following grammar, where $\Psi$ are atomic formulae of the
underlying first-order theory, and an arbitrary but a priori fixed number of
negation symbols is allowed to occur:
\[
\Phi \Coloneqq \exists x\, \Phi \mid \neg \Phi \mid \Phi \wedge \Phi \mid
\Psi \,.
\]
We give sufficient conditions for the fixed negation fragment of a first-order
theory to be decidable in polynomial time. We highlight that this fragment is more
permissive than the ``short fragment'' of Kannan, as it allows for an unbounded
number of quantified variables and an unbounded number of conjunctions. However,
it also implicitly fixes the number of quantifier alternations as well as the number
of disjunctions.

Our algorithmic framework is parametric on a concrete representation of the sets
definable within the first-order theory~$\mcT$ under consideration
and only requires a sensible representation of solution sets for conjunctions of
atomic formulae.  From this representation, the framework guides us to the
definition of a companion structure~$\mcR$ for the theory~$\mcT$
in which function symbols and relations in~$\mcR$ are interpreted as reductions
from parametrized complexity theory, such as~UXP reductions, see
e.g.~\cite[Chapter~15]{DowneyF99}. By requiring mild conditions on the types of
reductions and parameters that the functions and relations in~$\mcR$ must obey,
we are able to give a general theorem for the tractability of the fixed negation
satisfiability and entailment problems for $\mcT$. One technical issue we show
how to overcome in a general way is how to treat negation, which is especially
challenging when the initial representation provided to the framework is not
closed under complementation. Our main source of inspiration here is the notion 
of the so-called
\emph{difference normal form} of propositional logic, a rather unorthodox normal
form introduced by Hausdorff~\cite[Ch.~1\S5]{Hausdorff1914}.

As one of the main application of our framework, we show that the fixed negation 
fragment of weak Presburger arithmetic (\emph{weak PA})
is polynomial-time decidable. Weak PA is the 
first-order theory of the structure
$(\ZZ,0,1,+,=)$, which is strictly less expressive than standard Presburger
arithmetic.
It was recently shown that unrestricted weak PA has the
same complexity as standard Presburger arithmetic~\cite{ChistikovHHM22}. 
In contrast, Bodirsky et al.\
 showed that the weak PA fragment of existential linear
Horn equations $\bigwedge_{i \in I} (A_i \cdot \vec x = b_i) \rightarrow (C_i
\cdot \vec x = d_i)$ over $\ZZ$ with $|I|$ unbounded can be decided in PTIME~\cite{BMMM18}. It
follows from the generic results in this paper that the quantified versions of
those formulae with the number of quantifier alternations and~$\abs I$ fixed is
also polynomial-time decidable. In fact, our framework not only allows for deciding
satisfiability and validity of fixed negation formulae of Weak PA in PTIME, but
also to compute a representation of the set of solutions of a given formula.
This is the best possible such result, since we can show that, for~$I$
unbounded, the $\exists\forall$ fragment of linear Horn equations in two
variables is NP-hard. 
\begin{prop}
  Deciding two-variables $\exists\forall$ weak PA Horn sentences is NP-hard.
\end{prop}

\begin{proof}
The $\exists\forall$ Horn sentences of weak PA are of the form $\exists x
\forall y \bigwedge_{i=1}^k \psi_i(x,y)$, where each $\psi_i$ is a Horn clause,
i.e., a disjunction of literals in which at most one literal occurs positively.
NP-hardness for deciding these sentences follows 
by a straightforward reduction from the problem of deciding a
univariate system of non-congruences $\bigwedge_{i=1}^k x \not\equiv r_i \pmod
{m_i}$, where $m_i \geq 2$ and $r_i \in [0,m_i-1]$ for every $i \in [1,k]$. This
problem is NP-hard~\cite[Theorem~5.5.7]{BachShallit96}. 
For the reduction, simply apply the following equivalence: for every $x \in \ZZ$,
\[ 
  \bigwedge\nolimits_{i=1}^k x \not\equiv r_i
  \pmod {m_i} \ \ \iff \ \ \forall y : 
  \bigwedge\nolimits_{i=1}^k \lnot \big(x - r_i = m_i \cdot y \big)\,.
  \qedhere
\]
\end{proof}
An extended abstract of this paper appeared in the proceedings of the 48th 
International Symposium on Mathematical Foundations of Computer Science (MFCS 2023)~\cite{HMP23}.

\subsection{Structure of this paper}
The goals of this paper are twofold, and consequently the paper consists 
of two parts. In the first part, after recalling and introducing some basic 
definitions concerning first-order logic in \Cref{section:basic-def} 
as well as concepts underlying representations of objects and concepts from 
parametrized complexity in \Cref{sec:represent-and-complex}, we present our general algorithmic
framework in \Cref{section:fo-framework}. This framework is parametrized by a 
first-order theory $\mcT$, and we develop sufficient conditions on $\mcT$ 
for its $k$-negations fragment to be decidable in polynomial time, for every fixed $k \geq 0$.
\Cref{section:fo-framework} is intended to enable a reader to easily apply the
framework to determine whether
a given first-order theory has a polynomial-time decidable $k$-negations fragment. For that reason,
all proofs are relegated to \Cref{section:diffnf,appendix:sec4-new}, and the section's concepts are illustrated
using a simple fragment of Presburger arithmetic with one variable per inequality 
as a running example. 

In the second part of the paper, 
\Cref{sec:affine_subspaces,sec:shifted_lattices}, 
we instantiate
the framework to weak linear real arithmetic and weak Presburger arithmetic,
proving that their $k$-negations fragments are decidable in polynomial-time.

\subsection{Related work}
There is a long history of research on identifying syntactic fragments of
first-order logic with the goal of obtaining decidable fragments, possibly of
low complexity. Specifically, restrictions on the use of negation have been
widely studied. For instance, Kozen~\cite{Kozen81} showed that deciding positive
sentences in which no negation symbol occurs is NP-complete.
Voronkov~\cite{Voronkov99} generalized this fragment and showed that the
ground-negative fragment of first-order logic, where any negated atomic formula
is required to be a ground term, is $\Pi_2^\text{P}$-complete. B\'ar\'any, Ten
Cate and Segoufin~\cite{BCS15} identified \emph{guarded-negation first-order
logic} which requires all occurrences of negation to be of the form $\alpha
\wedge \neg \phi$ such that $\alpha$ is an atomic formula containing all free
variables of $\phi$. Guarded-negation first-order logic is 2EXPTIME-complete and
unified various previously identified decidable fragments of first-order logic
such as the \emph{unary negation}~\cite{TS13} and \emph{guarded
fragment}~\cite{And98} of first-order logic. More recent work explores
restrictions on features other than negation. Jonsson, Lagerkvist and
Osipov~\cite{JonssonLO24}, investigated constraint satisfaction problems (CSPs)
over two signatures, $\mathcal{A}$ and $\mathcal{B}$, where at most $k$
constraints from $\mathcal{B}$ are allowed. Under further assumptions on
the CSPs over $\mathcal{A}$, and on the definability of $\mathcal{B}$ within $\mathcal{A}$, 
they showed that these CSPs can be solved in polynomial time for fixed $k$.






\section{Preliminaries}
\label{section:basic-def}

This section focuses on notation and simple definitions that might be
non-standard to some readers. We assume familiarity with basic concepts from
logic and abstract algebra. 

\subsection{Sets and functions}
We write $\seq(A)$ for the set of all finite tuples over a set $A$, and denote
by $\emp$ the empty tuple. This definition corresponds to the standard notion of
Kleene star $A^*$ of a set $A$. The discrepancy in notation is introduced to
avoid writing $(\Sigma^*)^*$ for the domain of all tuples of finite words over
an alphabet~$\Sigma$, as in formal language theory the Kleene star comes
equipped with the axiom $(\Sigma^*)^* = \Sigma^*$. We denote this domain by
$\seq(\Sigma^*)$. 

We write $f \colonsub X\to Y$ (resp.~$f \colon X \to Y$) to denote a
\emph{partial (resp.~total) function} from $X$ to~$Y$. The domain of $f$ is
denoted by $\dom(f)$. We write $\id_A \colon A \to A$ for the identity function
on~$A$. Given $f \colonsub A \to B$, $g \colonsub B \to C$ and $h \colonsub D
\to E$, we denote by $(g \circ f) \colonsub A \to C$ and $(f \times h) \colonsub
A \times D \to B \times E$ the \emph{composition} and the \emph{Cartesian
product} of functions.

\label{section:structures-and-comp}
\label{sec:structures}
\subsection{Structures with indexed families of functions}
We consider a generalization of the traditional definition of structure from
universal algebra that accommodates for a potentially infinite number of
functions. As usual, a \emph{structure} $\mcA=(A,\sigma,I)$ consists of a
\emph{domain} $A$ (a~set), a \emph{signature} $\sigma$, and an
\emph{interpretation function} $I$. In this paper, the signature is a quadruple
$\sigma=(\mcF,\mcG,\mcR,\ar)$ containing not only a set of \emph{function
symbols} $\mcF$, a set of \emph{relation symbols} $\mcR$, and the \emph{arity
function} $\ar\colon \mcF\uplus \mcR\uplus\mcG\to\NN$, but also a set of
\emph{(indexed) families of function symbols} $\mcG$. Each element in the finite
set $\mcG$ is a pair $(g,X)$ where $g$ is a function symbol and $X$ is a
countable set of indices. The interpretation function~$I$ associates to every
$f\in \mcF$ a map $f^{\mcA}\colon A^{\ar(f)}\to A$, to every $(g,X)\in\mcG$ a
map $g^{\mcA}\colon X\times A^{\ar((g,X))}\to A$, and to every $R\in\mcR$ a
relation $R^{\mcA} \subseteq A^{\ar(R)}$ which we often view as a function
$R^{\mcA} \colon A^{\ar(R)} \to \{\top,\bot\}$.

\begin{exa}\label{example:family-of-functions}
  Consider the structure $\mcA=(\ZZ,\sigma,I)$ in which the signature~$\sigma$
  contains a single family of functions $(\mathrm{mul},\NN)$ of arity one, and
  the interpretation function~$I$ associates to $\mathrm{mul}$ the map
  $\mathrm{mul}^{\mcA}(n,x) = n \cdot x$ for all $n \in \NN$ and $x \in \ZZ$.
  Here, the family of functions $(\mathrm{mul},\NN)$ uniformly defines
  multiplication by a non-negative integer constant $n$. \hfill$\diamond$
\end{exa}

The standard notions  (see, e.g.,~\cite{BurrisSankappanavar1981}) of
\emph{homomorphism, embedding and isomorphism of structures}, as well as the
notions of \emph{congruence for a structure} and \emph{quotient structure}
extend naturally to structures having families of functions. For instance, a
\emph{homomorphism} from~$\mcA = (A,\sigma,I)$ into ${\mcB = (B,\sigma,J)}$ is a
map ${h \colon A \to B}$ that \emph{preserves} all functions, families of
functions and relations; so in particular given $(g,X) \in \mcG$, the map $h$
satisfies $g^{\mcB}(x,h(a_1),\dots,h(a_{\ar(g)})) =
h(g^{\mcA}(x,a_1,\dots,a_{\ar(g)}))$ for every $x \in X$ and
$a_1,\dots,a_{\ar(g)} \in A$.

We denote structures in calligraphic letters~$\mcA,\mcB,\ldots$ and their
domains in capital letters $A,B,\ldots$\,. When the arity function $\ar$ and the
interpretation~$I$ are clear from the context, we
write~$(A,f_1^{\mcA},\smalldots,f_j^{\mcA},(g_1^{\mcA},X_1),\smalldots,(g_\ell^{\mcA},X_\ell),R_1^{\mcA},\smalldots,R_k^{\mcA})$
for $\mcA = (A,\sigma,I)$ with $\sigma =
(\{f_1,\smalldots,f_j\},\{(g_1,X_1),\smalldots,(g_\ell,X_\ell)\},\{R_1,\smalldots,R_k\},\ar)$,
and often drop the superscript $\mcA$. For instance, the structure
from~\Cref{example:family-of-functions} can be denoted as
$(\ZZ,(\mathrm{mul},\NN))$.



\subsection{First-order theories (finite tuples semantics)}
\label{sec:fo_structures}
The first-order (FO) language of the signature $\sigma = (\mcF,\mcG,\mcR,\ar)$
is the set of all formulae $\Phi,\Psi,\dots$ built from the grammar 
\begin{equation}
  \label{eq:first-order-language}
  \begin{aligned}
    \Phi,\Psi &\Coloneqq\ r(t_1,\dots,t_{\ar(r)}) \mid \lnot
    \Phi \mid \Phi \land \Psi \mid \exists x . \Phi\\[3pt]
    t &\Coloneqq\ x \mid f(t_1,\dots,t_{\ar(f)}) \mid
    g(i,t_1,\dots,t_{\ar(g)}),
  \end{aligned}
\end{equation}
where $x \in \VV$ is a first-order variable, $r \in \mcR$, $f \in \mcF$, $(g,X)
\in \mcG$ and $i \in X$ (more precisely, $i$ belongs to a representation of $X$,
more details are given in~\Cref{sec:represent-and-complex} below). Lexemes of the form
$r(t_1,\dots,t_{\ar(r)})$ are the \emph{atomic formulae} of the language.
Throughout this paper, we implicitly assume an order on the variables in~$\VV$,
and write $x_j$ for the $j$-th variable (indexed from~$1$). We write
$\AC(\sigma)$ for the set of all~\emph{conjunctive queries} 
of the first-order language of~$\sigma$, 
that is the set of all (quantifier-free) conjunctions of atomic formulae in the language.

Consider a structure~$\mcA = (A,\sigma,I)$. Given an atomic formula
$r(t_1,\dots,t_{\ar(r)})$ from the grammar in~\Cref{eq:first-order-language}
having $x_n$ as the largest appearing variable, we write
$\sem{r(t_1,\smalldots,t_{\ar(r)})}_{\mcA} \subseteq A^n$ for the set of
$n$-tuples, corresponding to values of the first $n$ variables, that makes the
formula~$r(t_1,\smalldots,t_{\ar(r)})$ true under the given interpretation $I$.
Furthermore, let us define $\vec{I} \coloneqq \{ (i_1,\smalldots,i_k) \in
\seq(\NN) : i_1,\smalldots,i_k \text{ all distinct}\}$. We denote by~$\FO(\mcA)$
the structure all \emph{first-order sets definable in}~$\mcA$, which is the
structure  
\[
  \FO(\mcA)
  \coloneqq
  (\sem{\mcA}_{\FO},\bot,\top,\lor,\land,
  -,(\proj,\vec{I}),(\unproj,\vec{I}),\leq),
  \ \text{where}
\] 
\begin{enumerate}
  \item $\sem{\mcA}_{\FO}$ is the least set containing
  $\sem{r(t_1,\smalldots,t_{\ar(r)})}_{\mcA}$, for each atomic formula
  $r(t_1,\smalldots,t_{\ar(r)})$, and that is closed under the functions $\bot,
  \top, \lor, \land,$ $-, (\proj,\vec{I})$ and $(\unproj,\vec{I})$, defined
  below.
  \item The functions $\bot$ and $\top$ are interpreted as $\emptyset$ and
  $\{\emp\}$, respectively.
  \item Given $X \subseteq A^n$, $Y \subseteq A^m$ and $\vec i = (i_1,\dots,i_k)
  \in \vec I$, and defining $M \coloneqq \max(n,m)$,
    \begin{align*}
      X \lor Y &\coloneqq \{(a_1,\smalldots,a_M) : (a_1,\smalldots,a_n) \in X \text{ or } (a_1,\smalldots,a_m) \in Y \},\\[3pt]
      X \land Y &\coloneqq \{(a_1,\smalldots,a_M) : (a_1,\smalldots,a_n) \in X \text{ and } (a_1,\dots,a_m) \in Y \},\\[3pt]
      X - Y &\coloneqq \{(a_1,\smalldots,a_M) : (a_1,\smalldots,a_n) \in X \text{ and } (a_1,\smalldots,a_m) \not\in Y \},\\[3pt] 
      \proj (\vec i, X) &\coloneqq \{ \gamma \in A^{n} : \text{there is } \vec a \in A^k \text{ s.t. } \gamma[\vec i \gets \vec a] \in X \},\\[3pt]
      \unproj (\vec i, X) &\coloneqq \{ \gamma \in A^{n} : \text{for every } \vec a \in A^k,  \gamma[\vec i \gets \vec a] \in X \},\\[3pt]
      X \leq Y &\text{ if and only if } X \times A^m \subseteq Y \times A^n,
    \end{align*}
    where $\gamma[\vec i \gets \vec a]$ is the tuple obtained from $\gamma \in
    A^n$ by replacing its $i_j$-th component with the $j$-th component of $\vec a$, for every $j \in [1,\min(k,n)]$.
\end{enumerate}
The semantics $\sem{.}_{\mcA}$ of the FO language of $\sigma$ is extended to
non-atomic formulae via~$\FO(\mcA)$. As usual, $\sem{\lnot \Phi}_{\mcA}
\coloneqq \top - \sem{\Phi}_{\mcA}$, $\sem{\Phi \land \Psi}_{\mcA} \coloneqq
\sem{\Phi}_{\mcA} \land \sem{\Phi}_{\mcA}$, and $\sem{\exists x_i . \Phi}_{\mcA}
\coloneqq \proj((i),\sem{\Phi}_{\mcA})$. 
We omit the subscript $\mcA$ from $\sem{\cdot}_{\mcA}$ when it is clear from context, writing simply $\sem{\cdot}$.
We remark that $\FO(\mcA)$ contains
operators whose syntactic counterpart is absent from the FO language of
$\sigma$, such as the universal projection $\unproj$. This is done for
algorithmic purposes, as the framework we introduce in
\Cref{section:fo-framework} treats these operators as first-class citizens.

\subsection{Fixed negations fragments}
Let $k \in \NN$ be fixed. The \emph{$k$-negations fragment of the FO language}
of a signature $\sigma$ is the set of all formulae having at most $k$
negations~$\lnot$. Note that, following the grammar provided
in~\Cref{eq:first-order-language}, this restriction also bounds the number of
disjunctions and alternations between existential and universal quantifiers that
formulae can have. Given a structure $\mcA = (A,\sigma,I)$, we study the
following problem:
\vspace{5pt}
\begin{center}
\begin{tabular}{rl}
$k$ \textit{negations satisfiability problem}: \hspace{-5pt} & Given a formula
  $\Phi$ with at most $k$ negations,\\ 
  & decide whether $\sem{\Phi} \neq \emptyset$.
\end{tabular}
\end{center}

\section{Representations and parametrised complexity of signatures}
\label{sec:represent-and-complex}

Per se, a structure $\mcA$ cannot be analysed algorithmically, in particular
because the elements of $A$ do not have a notion of size. A standard way to
resolve this issue is defining computability via the notion of representations
(as it is done for instance in computable analysis~\cite{Wei00}). 

\subsection{Representations}
\label{sec:representations}
A \emph{representation} for a set $A$ is a surjective partial map $\rho
\colonsub \Sigma^* \to A$, where $\Sigma$ is a finite alphabet. Words $w \in
\Sigma^*$ are naturally equipped with a notion of \emph{size}, that is their
length, denoted by $\abs{w}$. Observe that not all words are valid
representations for elements of~$A$ ($\rho$ is partial) and each element from
$A$ may be represented in several ways ($\rho$ is not assumed to be injective).
We write $(\approx_{\rho}) \subseteq \Sigma^* \times \Sigma^*$ for the
equivalence relation $\{(w_1,w_2) : w_1,w_2 \in \dom(\rho) \text{ and }
\rho(w_1) = \rho(w_2)\}$ and define ${h_{\rho} \colon \dom(\rho)/_{\approx_\rho}
\to A}$ to be the bijection satisfying $h_{\rho}([w]_{\approx_\rho}) = \rho(w)$,
for every $w \in \dom(\rho)$. Here, $\dom(\rho)/_{\approx_\rho}$ is the set of
all equivalence classes $[w]_{\approx_\rho}$ of words $w \in \dom(\rho)$.

\begin{exa}\label{ex:twoscomplement}
  The two's complement least significant digit first representation
  of~$\ZZ$ is given by the map $\rho \colonsub \{0,1\}^* \to \ZZ$
  mapping every non-empty word of binary digits~${{d_0}\dots{d_m} \in
  \{0,1\}^{m+1}}$ to the integer $-d_m \cdot 2^m + \sum_{i=0}^{m-1} d_i \cdot
  2^i$. This representation is not defined on the empty word. A property of this
  representation is that padding each word to the right by repeating its most
  significant digit does not change the encoded number, that is,
  $({d_0}\dots{d_m}) \approx_\rho ({d_0}\dots{d_m}{d_m} \dots {d_m})$, where
  $d_m$ is repeated an arbitrary number of times.\hfill$\diamond$
\end{exa}


It is often more practical to represent elements of $A$ by objects that are more
sophisticated than words in~$\Sigma^*$, such as tuples, automata, graphs, etc.
Taking these representations does not change the notion of computability or
complexity, because they can be easily encoded as words (over a bigger alphabet,
if necessary). In our setting, of particular interest are representations as
tuples of words. The notion of size for words trivially extends to tuples:
$\abs{(w_1,\dots,w_n)} \coloneqq n + \sum_{i=1}^n |w_i|$. Given representations
$\rho \colonsub \Sigma^* \to A$ and $\rho' \colonsub \Pi^* \to A'$, we rely on
the following \emph{operations on representations}:
\begin{itemize}
    \item The Cartesian product $\rho \times \rho'$ of representations, defined
    as in~\Cref{section:basic-def}.
    \item The representation $\seq(\rho) \colonsub \seq(\Sigma^*) \to \seq(A) $
    that, for every $n \in \NN$, given a tuple $(w_1,\dots,w_n) \in
    \dom(\rho)^n$ returns $(\rho(w_1),\dots,\rho(w_n))$. 
\end{itemize}
We also require representations for basic objects such as~$\NN$, $\ZZ$ and so
on. Specifically, we assume to have \emph{canonical
representations}~$\stdrepr{X}$ for the following countable domains $X$:
\begin{itemize}
  \item $X = \NN$ or $X = \ZZ$, so that~$\stdrepr{X}$ is a representation of
  $\NN$ or $\ZZ$, respectively. We assume this representation to be any standard
  binary encoding of natural numbers or integers that allows arithmetic operations
  such as addition, multiplication and integer division to be implemented in
  polynomial time, as for instance the representation in~\Cref{ex:twoscomplement}.

  \item $X$ is any finite set, e.g., we assume to have a representation
  $\stdrepr{\BB}$ for the Booleans $\BB = \{\top,\bot\}$. Note that, since $X$
  is finite, operations on this set are constant time.

  \item $X = \Sigma^*$ where $\Sigma$ is any finite alphabet. In this case,
    $\stdrepr{\Sigma^*}\coloneqq\id_{\Sigma^*}$.
\end{itemize}
For a canonical representation $\stdrepr{X}$ and $n \in \NN$, we
write $\stdrepr{X^n}$ for the Cartesian product $(\stdrepr{X})^n$.%

\subsection{Implementations and computability}
\label{sec:represented_structures}
Let $\rho \colonsub \Pi^* \to A$ and $\rho_1,\dots,\rho_n$ be representations,
with $\rho_i \colonsub \Sigma_i^* \to A_i$. A~function $f \colon A_1 \times
\dots \times A_n \to A$ is said to be $(\rho_1 \times \dots \times \rho_n,
\rho)$-\emph{computable} if there is a function ${F \colon \dom(\rho_1) \times
\dots \times \dom(\rho_n) \to \dom(\rho)}$ that is computable (by a Turing
machine) and satisfies $\rho(F(w_1,\dots,w_n)) =
f(\rho_1(w_1),\dots,\rho_n(w_n))$ for all $w_i \in \dom(\rho_i)$, $i \in [1,n]$.
The function $F$ is said to be a $(\rho_1 \times {\dots} \times \rho_n,
\rho)$-\emph{implementation} of $f$. For simplicity, we do not mention the
representations of a computable function when it operates on canonical types:
for sets $A,A_1,\dots,A_n$ admitting canonical representations, a function ${f
\colon A_1 \times \dots \times A_n \to A}$ is said to be \emph{computable}
whenever it is $(\stdrepr{A_1} \times {\dots} \times \stdrepr{A_n},
\stdrepr{A})$-\emph{computable} (the $\stdrepr{A_i}$ and $\stdrepr{A}$ are the
canonical representations of $A_i$ and $A$).

\begin{exa}\label{example:ZZ-canonical}
  The addition function $+\colon \ZZ \times \ZZ \to\ZZ$ is $(\stdrepr{\ZZ}
  \times \stdrepr{\ZZ},\stdrepr{\ZZ})$-computable. Since $\stdrepr{\ZZ}$ is a
  canonical representation, we simply say that $+$ is \emph{computable}. This is
  the standard notion of computability over~$\ZZ$, with respect to a binary encoding of
  integers.\hfill$\diamond$
\end{exa}

Let $\mcA = (A,\sigma,I)$ be a structure and $\rho \colonsub \Sigma^* {\to} A$
be a representation. Let $\mcM \coloneqq (\dom(\rho),\sigma,J)$ be a structure
where the interpretation $J$ associates computable functions to each function,
family of functions and relations in~$\sigma$, and makes $\approx_{\rho}$ a
congruence for $\mcM$. Then, $\mcM$ is said to be a $\rho$\emph{-implementation}
of $\mcA$ whenever $\rho$ is a homomorphism between $\mcM$ and $\mcA$. We
highlight the fact that, compared to a standard homomorphism between structures,
an implementation is always surjective (since $\rho$ is surjective) and
forces~$J$ to give an interpretation to functions and relations in~$\sigma$ in
terms of computable functions.

\begin{exa}\label{example:ZZ-implementation}
  The structure $(\dom(\stdrepr{\ZZ}),+)$ is a $\stdrepr{\ZZ}$-implementation of
  $(\ZZ,+)$. For a further example, consider the structure
  $\mcA = (L,\cup,\cap,(\cdot)^c)$ where $L$ is the set of all regular languages over a
  fixed finite alphabet $\Sigma$, and $\cup$, $\cap$, and $(\cdot)^c$ are the
  canonical operations of union, intersection and complementation of languages,
  respectively. As a representation, one can consider the map $\rho$ 
  taking as input a deterministic finite automaton (DFA) over~$\Sigma$,
  and returning the language the automaton accepts.
  We obtain the structure $\mcM = (\dom(\rho),\cup,\cap,(\cdot)^c)$ 
  in which the functions~$\cup$, $\cap$, and $(\cdot)^c$ 
  can be implemented by Turing machines manipulating DFAs; 
  and $\mcM$ is a $\rho$-implementation of $\mcA$.\hfill$\diamond$
\end{exa}

\subsection{Parametrised complexity of signatures}\label{sec:para_complexity}
The framework we define in the next section requires the introduction of a
notion of parametrised complexity for the signature of a structure (which we
call a \emph{UXP signature}) which we now formulate. First, let us recall the
standard notion of UXP reduction from parametrised complexity
theory~\cite[Chapter~15]{DowneyF99}. Let $\Gamma$ and $\Pi$ be two finite
alphabets, and $D \subseteq \Gamma^*$. A \emph{parameter function} is a map
$\eta \colon \Gamma^* \to \NN$ such that $\eta(w) \geq 1$ for every $w \in
\Gamma^*$. 
%
%
A computable function $F \colon D \to \Pi^*$ is said to be a \emph{uniform
slicewise polynomial reduction} for two parameter functions $\eta$ and $\theta$,
or $(\eta,\theta)$-UXP reduction for short, whenever there is an increasing map
$G \colon \NN \to \NN$ such that for every $w \in D$,
$F(w)$~runs~in~time~$|w|^{G(\eta(w))}$ (w.l.o.g. assume~${\abs{w} \geq 2}$), and
$\theta(F(w)) \leq G(\eta(w))$.

As usual in computability theory, functions~$F$ with multiple arguments are
handled by introducing a special symbol to the alphabet $\Gamma$, say $\#$, to
separate the arguments, thus viewing $F$ as a function in one input. For
instance, an operator \mbox{$\oplus \colon \Sigma_1^* \times \Sigma_2^* \to
\Sigma^*$} can be interpreted by a computable function taking as inputs words
$w_1 \# w_2$ with ${(w_1,w_2) \in \Sigma_1^* \times \Sigma_2^*}$. The product
${(\eta_1 \cdot \eta_2)(w_1 \# w_2) \coloneqq \eta_1(w_1) \cdot \eta_2(w_2)}$ of
parameter functions $\eta_1 \colon \Sigma_1^* \to \NN$ and $\eta_2 \colon
\Sigma_2^* \to \NN$ can be used to refine the complexity analysis of $\oplus$ to
each of its two arguments. We write $\parone$ for the trivial parameter function
defined as $\parone(w) \coloneqq 1$ for all $w \in \Sigma^*$.

Let $\mcA = (A,\sigma,I)$ be a structure, $\sigma = (\mcF,\mcG,\mcR,\ar)$, $\rho
\colonsub \Sigma^* \to A$ be a representation, and $\eta \colon \Sigma^* \to
\NN$ be a parameter function. We say that $\mcA$ has a $(\rho,\eta)$\emph{-UXP
signature} whenever there is an interpretation function $J$ such that \emph{(i)}
$(\dom(\rho),\sigma,J)$ is a $\rho$-implementation of~$\mcA$ and \emph{(ii)} $J$
associates a $(\eta^{\ar(f)},\eta)$-UXP reduction to every $f \in \mcF$, a
$(\parone\cdot\eta^{\ar(g)},\eta)$-UXP reduction to every $(g,X) \in \mcG$, and
a $(\eta^{\ar(R)},\parone)$-UXP reduction to every $R \in \mcR$. Note that for
$\eta = \parone$, all those reductions become polynomial time  functions. In
this case we say that $\mcA$ has a ($\rho$-)tractable signature.



\begin{exa} 
  Consider the structures of regular languages~$\mcA$ and of deterministic finite 
  automata~$\mcM$ in~\Cref{example:ZZ-implementation}. 
  The functions $\cup$, $\cap$, and $(\cdot)^c$ 
  can be implemented in polynomial time on DFAs, 
  therefore $\mcA$ has a $\rho$-tractable signature. 
  However, $\mcA$ does not have a tractable signature for the representation of regular
  languages as non-deterministic finite automata (NFAs), because computing
  $(\cdot)^c$ on NFAs requires first to determinise the automaton. \hfill$\diamond$
\end{exa}

As in the case of representations, it is often more practical to have parameter
functions from objects other than words. Given a parameter function $\theta
\colon \Sigma^* \to \NN$, we consider the operations
$\len(\theta) \colon \seq(\Sigma^*) \to \NN$, ${\max(\theta) \colon
\seq(\Sigma^*) \to \NN}$ and $\depth(\theta) \colon \seq(\seq(\Sigma^*)) \to
\NN$ on parameter functions. For~$\vec w = (w_1,\dots,w_n)$, they are defined as
  \begin{center}
    $\len(\theta)(\vec w) \coloneqq \textstyle\sum_{i=1}^n \theta(w_i);\quad
    \max(\theta)(\vec w) \coloneqq \textstyle\max_{i=1}^n \theta(w_i);\quad
    \depth(\theta) \coloneqq \len(\len(\theta)).$
  \end{center}


\section{A framework for the fixed negation fragment of first-order theories}
\label{section:fo-framework}

Fix a structure $\mcA = (A,\sigma,I)$ and consider the structure~$\FO(\mcA)$
from~\Cref{sec:fo_structures}:
\[
    \FO(\mcA) \coloneqq
        (\sem{\mcA}_{\FO},\bot,\top,\lor,\land,
        -,(\proj,\vec{I}),(\unproj,\vec{I}),\leq).
\]
In this section, we describe a framework that can be employed to show that the
$k$ negation satisfiability problem for $\FO(\mcA)$ is in PTIME. Part of our
framework is generic, i.e., it~applies to any first-order theory, while other
parts are specific to the theory under consideration. To keep the presentation of the
framework concise, we postpone detailed 
proofs of the formal
statements to the subsequent \Cref{section:diffnf,appendix:sec4-new}. 
The current section is
best regarded as a blueprint intended to guide the instantiation of the
framework.

\subsection{Ideas underlying the framework}
\label{subsection:framework-ideas}

To understand the framework, we first discuss how we can
exploit the fact that our formulae only have a fixed number of negations. For
simplicity, let us focus for the time being on \emph{quantified Boolean
formulae} (QBF) in prenex form. These are formulae of the form $\exists q_1
\forall q_2 \dots \exists q_n \Phi$, where $\Phi$ is a formula from
propositional logic, and $q_1,\dots,q_n$ are some Boolean variables occurring in
it. A first key question is whether bringing the quantifier-free part~$\Phi$ of
a QBF formula into a particular normal form can be computationally beneficial. Of
course, due to $\Phi$ having a fixed number of negations, $\Phi$ could be translated
into DNF in PTIME.
However, because of quantifier alternation together with the unbounded number of
conjunctions, choosing this normal form comes with several intricacies. Another
option we might try is to put $\Phi$ into a form where all but a fixed amount of
constraints are in Horn form, and then try to rely on the algorithm to solve
quantified Horn Boolean satisfiability in PTIME~\cite{KarpinskiBS87}. This works
for the Boolean case, but not for an arbitrary theory. For instance, as shown in~\Cref{section:introduction}, the quantified Horn satisfiability
problem for the FO theory of~$\mathcal{Z} = (\ZZ,0,1,+,=)$, i.e.~weak PA, is
already NP-hard for the alternation prefix $\exists\forall$ and $2$ variables
(and NEXPTIME-hard in general~\cite{ChistikovHHM22}). It turns out that a
suitable normal form for $\Phi$ is given by formulae of the form $\Phi_1 -
(\Phi_2 - ( \dots - (\Phi_{k-1} - \Phi_k)))$, where each $\Phi_i$ is a
negation-free formula in disjunctive normal form (DNF), and $\Psi_1 - \Psi_2$ is
the relative complementation $\Psi_1 \land \lnot \Psi_2$.  
As we will see in~\Cref{subsection:steps-framework}, this atypical normal form
(introduced by Hausdorff in \cite{Hausdorff1914} and called \emph{difference
normal form} in~\cite{Junker00}) not only fully makes use of our restriction on
the number of negations, but also exhibits nice properties in relation to
quantification.

\begin{exa}\label{example:qbf-0}
    Consider the propositional formula $\Phi(a,b,c) \coloneqq (a \lor b) \land
    (\lnot a \lor c) \land (\lnot b \lor \lnot c)$. This formula is satisfied by the
    assignments $(a = \top, b = \bot, c = \top)$ and $(a = \bot, b = \top, c =
    \bot)$, so in particular the QBF formula $\forall a \exists b \exists c \, \Phi$
    is valid. The difference normal form of $\Phi$ is:
    \[
        \Psi \coloneqq (a \lor b) - \big((a \lor (b \land c)) - ((a \land c) - (a \land b \land c))\big).
    \]
    All propositional formulae can be converted into difference normal form, as we will
    see in~\Cref{section:diffnf}. In~\Cref{example:qbf}, we will discuss how to eliminate the quantifiers from $\forall a \exists b \exists c \, \Psi$.~\hfill$\diamond$
\end{exa}

A second key question is what representation of the domain~$\sem{\mcA}_{\FO}$
works best for our purposes, as formulae might not be the right ``data
structures''. Though the difference normal form already sets how to treat
disjunctions and negations, we have the flexibility to vary the representation
of conjunctions of atomic formulae. Let us be a bit more precise. Consider a
domain $\domain \subseteq \sem{\mcA}_{\FO}$ containing \emph{at least} all the
sets $\sem{\Psi}$, for every $\Psi \in \AC(\sigma)$. We define $\un(\domain)$ to
be the smallest set contaning $\domain$ and being closed under the
disjunction~$\lor$, and $\dnf(\domain)$ to be the smallest set
containing~$\un(\domain)$ and being closed under relative complements~$X-Y$,
with $X \in \un(\domain)$ and $Y \in \dnf(\domain)$. From the fact that all
propositional formulae can be converted into difference normal form, we conclude
that ${\{ \sem{\Phi} : \Phi \text{ quantifier free} \} \subseteq
\dnf(\domain)}$. Then, for a representation~$\rho \colonsub \Sigma^* \to
\domain$, the difference normal form gives a straightforward way of
representing~$\dnf(\domain)$. First, we define the representation ${\un(\rho)
\colonsub \seq(\Sigma^*) \to \un(\domain)}$, given by ${\un(\rho)(c_1,\dots,c_n)
\coloneqq (\rho(c_1) \lor \dots \lor \rho(c_n))}$, where each $c_i$ belongs to
$\Sigma^*$. A representation for~$\dnf(\domain)$ is then given by the map
${\dnf(\rho) \colonsub \seq(\seq(\Sigma^*)) \to \dnf(\domain)}$ defined as
\[\dnf(\rho)(u_1,\dots,u_m) \coloneqq \un(\rho)(u_1) - \big(\un(\rho)(u_2) -
\big(\dots - \big(\un(\rho)(u_{m-1}) - \un(\rho)(u_m)\big)\big)\big),
\] 
where each $u_i$ belongs to~$\seq(\Sigma^*)$. The key point is that the representation
$\rho$ can be selected so that the elements in~$\domain$ are encoded as
something other than formulae. For instance, for linear arithmetic theories,
alternative representations are given by finite automata~\cite{Buchi60} or geometric
objects~\cite{0001HM22}. In the instantiations of the framework provided
in~\Cref{sec:affine_subspaces,sec:shifted_lattices} we will use the geometric
objects. Of course, relying on representations other than formulae requires an
efficient way of changing representation. This is stressed in the
forthcoming~\Cref{theorem:PuttingAllTogether}.
One last observation: above, we defined the
domain~$\domain$ to be a superset of $\{\sem{\Psi}_{\mcA} : \Psi \in
\AC(\sigma)\}$. This is because more general sets might be required to make
$\dnf(\domain)$ closed under (universal) projection. For instance, in weak
integer arithmetic, the formula $\exists y : x = 2 \cdot y$, stating that $x$ is
even, cannot be expressed with a quantifier-free formula, hence $\sem{\exists y
: x = 2 \cdot y}_{\mathcal{Z}}$ must be added~to~$\domain$.

\subsection{What the framework achieves}
The following proposition formalises the observations
in~\Cref{subsection:framework-ideas}. Recall that an algorithm is 
in $\chi$-UXP, for a parameter $\chi \colon \Sigma^* \to \NN$, if it runs in
time $\abs{w}^{G(\chi(w))}$ for every $w \in \Sigma^*$, for some function $G
\colon \NN \to \NN$ not depending on~$w$. A decision problem is in
\mbox{$\chi$-UXP} if there is a $\chi$-UXP algorithm solving that problem.

\begin{restatable}{prop}{PuttingAllTogether}\label{theorem:PuttingAllTogether}
    Fix $k \in \NN$. Assume the following objects to be defined:
    \begin{enumerate}
        \item\label{item:rep-rho} A representation $\rho$ of
        $\domain:=\bigcup_{n\in\NN}\domain_n$, where, for all $n\in\NN$,
        $\domain_n\subseteq\powerset(A^n)$ is s.t.~$\sem{\Psi}_{\mcA} \in
        \domain_n$ for every $\Psi \in \AC(\sigma)$ having maximum variable
        $x_n$.
        \item\label{item:map-F} A $(\xi,\theta)$-UXP reduction $F \colon
        \AC(\sigma) \to \dom(\rho)$ s.t.~${(\rho \circ F)(\Psi) =
        \sem{\Psi}_{\mcA}}$ for all~$\Psi \in \AC(\sigma)$.
    \end{enumerate}
    If $\mcD \coloneqq {(\dnf(\domain),\bot,\top,\lor,\land,-,(\proj,\vec I),
    (\unproj, \vec I), \leq)}$ has a \mbox{$(\dnf(\rho),\depth(\theta))$-UXP}
    signature,
    \begin{itemize}
        \item the $k$ negations satisfiability problem for $\FO(\mcA)$ is in
        $\xi$-UXP (in PTIME, if $\xi = \parone$), and 
        \item there is a $\xi$-UXP (polynomial time, if $\xi = \parone$)
        algorithm that, given a formula $\Phi$ of\, $\FO(\mcA)$ having at most
        $k$ negations, returns $X$ in $\dom(\dnf(\rho))$ such
        that~${\dnf(\rho)(X) = \sem{\Phi}_{\mcA}}$.
    \end{itemize}
\end{restatable}

\noindent
By virtue of the discussion in~\Cref{subsection:framework-ideas}, establishing~\Cref{theorem:PuttingAllTogether}
is straightforward: the reduction~$F$ enables an efficient conversion from
$\AC(\sigma)$ to elements in $\dom(\rho)$, and, since $\mcD$ is a structure,
$\dnf(\domain)$ is closed under all the operations in the signature and thus it
is equal to $\sem{\mcA}_{\FO}$. Consequently, $\FO(\mcA)$ has a
$(\dnf(\rho),\depth(\theta))$-UXP signature, and one can efficiently use
$\dnf(\rho)$ as a data structure to carry out the algorithm to decide
satisfiability: it suffices invoking the various UXP reductions implementing the
functions and relations in~$\mcD$. We remark that the sole purpose of the
parameter $\theta$ is to factor in the parameter~$\xi$, and one can set $\theta \coloneqq \parone$ in the case of $\xi \coloneqq \parone$ (i.e., the
case yielding PTIME algorithms).

\subsection{The framework}
\label{subsection:steps-framework}
To apply~\Cref{theorem:PuttingAllTogether}, one has to provide the required
representation~$\rho$ and the reduction~$F$ from~\Cref{item:rep-rho,item:map-F}, 
and show that the structure $\mcD$ has the desired UXP signature. 
Whereas the choice of~$\rho$ and $F$ depends on the FO theory at
hand, we show that a significant portion of the work required to prove that
$\mcD$ has an UXP signature can be treated in a general way, thanks to
the notion of difference normal form. This ``automation'' is the
core of our framework, which provides a minimal set of subproblems that are
sufficient to conclude that $\mcD$ has a $(\dnf(\rho),\depth(\theta))$-UXP
signature. Below, we divide those subproblems into two requirements, one for Boolean
connectives and one for quantification. One significant result in this context
is that negation can be treated in a general way.
As a running example, we sketch the instantiation of our framework on 
the fragment of Presburger arithmetic having one variable per inequality.
\begin{restatable}[Boolean connectives]{step}{STEPONE}\label{step:basic_framework_ptime}
    Establish the following properties:
    \begin{enumerate}[label=(F\arabic*)]
        \item\label{step:basic_framework_ptime:Item1} The
        structure $(\domain,\land,\leq)$ has a $(\rho,\theta)$-UXP signature. 
        \item\label{step:basic_framework_ptime:Item2} The
        structure $(\un(\domain),\leq)$ has a
        $(\un(\rho),\len(\theta))$-UXP~signature.
    \end{enumerate}  
\end{restatable}
\noindent 
\Cref{step:basic_framework_ptime} asks to provide algorithms for solving typical
computational problems that are highly domain-specific:
\Cref{step:basic_framework_ptime:Item1} considers the intersection and inclusion
problems for elements of $\domain$, with respect to the representation~$\rho$,
whereas \Cref{step:basic_framework_ptime:Item2} deals with the inclusion problem
for unions of elements in~$\domain$, with respect to the
representation~$\un(\rho)$. In the case of unions, we highlight the parameter
$\len(\theta)$ which fixes the length of the union.

\begin{exa}[PA with one variable per inequality]\label{example-unary-PA}
    Consider the fragment of Presburger arithmetic in which formulae follow the grammar in~\Cref{eq:first-order-language}, and atomic formulae are of the form $x \leq k$ or $x \geq k$, where $x$ is a variable ranging over the integers, and $k
    \in \ZZ$. We start by defining the objects in~\Cref{item:rep-rho,item:map-F} 
    of~\Cref{theorem:PuttingAllTogether}: 
    \begin{itemize}
        \item Define $\domain \coloneqq \{\sem{\Psi} : \Psi \in \AC(\sigma)\}$,  and $\theta \coloneqq \parone$.
        \item We set $\rho$ to be a map taking as inputs~$\top$, $\bot$, or systems of constraints~$\Psi$ of the form ${\bigwedge_{i=1}^m x_i \in [\ell_i,u_i]}$,
        where $\ell_i \in \ZZ \cup \{-\infty\}$ and $u_i \in \ZZ \cup \{+\infty\}$, for every $i \in [1,m]$. 
        The definition of $\rho$ is given by~$\rho(\bot) \coloneqq \sem{\bot}$, $\rho(\top) \coloneqq \sem{\top}$ and $\rho(\Psi) \coloneqq  \sem{\bigwedge_{j=1}^m (\ell_j \leq x_j \land x_j \leq u_j)}$, where $-\infty \leq x_j$ and $x_j \leq +\infty$ are interpreted as~$\top$.
        
        \item For a given $\Phi \in \AC(\sigma)$ the function $F$ returns the element of $\dom(\rho)$ computed (in polynomial time) as follows: Let $x_n$ be the maximum variable occurring in $\Phi$. For every $i \in [1,n]$, find the largest $\ell_i \in \ZZ$ and the smallest $u_i \in \ZZ$ such that $\ell_i \leq x_i$ and $x_i \leq u_i$ occur in $\Phi$. If one of these inequalities does not appear, take the corresponding bound to be $\pm \infty$ instead. Return $\bigwedge_{i=1}^n x_i \in [\ell_i,u_i]$.
    \end{itemize}
    Let us now consider~\Cref{step:basic_framework_ptime:Item1}
    of~\Cref{step:basic_framework_ptime}. 
    To establish this item, we must provide UXP reductions for computing, given $\Phi,\Psi \in \dom(\rho)$, (1) an element in $\dom(\rho)$ representing $\rho(\Phi) \land \rho(\Psi)$, and (2) whether $\rho(\Phi) \leq \rho(\Psi)$. 
    Both problems can in fact be solved in polynomial time. 
    Let $\Phi = \bigwedge_{j=1}^m x_i \in [\ell_i,u_i]$ and $\Psi = \bigwedge_{j=1}^n x_i \in [\ell_i',u_i']$.  (Similar arguments hold when at least one among $\Phi$ or $\Psi$ is $\top$ or $\bot$.) 
    For the first problem, the algorithm simply returns $\bigwedge_{i=1}^{\max(n,m)} x_j \in [\max(\ell_j,\ell_j'),\min(u_j,u_j')]$, 
    where the values $\ell_j$ or $\ell_j'$ (resp.~$u_j$ or $u_j'$) not occurring in $\Phi$ or $\Psi$ are defined as $-\infty$ (resp.~$+\infty$).
    For the second problem, the algorithm returns true whenever 
    either $u_i < \ell_i$ for some $i \in [1,m]$ (in this case, $\Phi$ is unsatisfiable), or $\ell_i' \leq \ell_i \leq u_i \leq u_i'$ for every $i \in [1,\max(n,m)]$. 
    
    To establish~\Cref{step:basic_framework_ptime:Item2}, 
    we must provide
    a $(\len(\parone)^{2},\parone)$-UXP reduction for testing
    the inclusion~$\un(\rho)(\Phi_1,\dots,\Phi_j) \leq \un(\rho)(\Psi_1,\dots,\Psi_k)$,
    where each $\Phi_i$ and $\Psi_i$ belongs to~$\dom(\rho)$. By definition of
    $\un(\rho)$, an algorithm for deciding $\rho(\Phi) \leq
    \un(\rho)(\Psi_1,\dots,\Psi_k)$, with $\Phi$ in $\dom(\rho)$, suffices. 
    Viewing elements of $\dom(\rho)$ as formulae in $\AC(\sigma)$, 
    this is equivalent to testing the unsatisfiability of $\Phi \land
    (\lnot \Psi_1) \land \dots \land (\lnot \Psi_k)$. The parameter of the UXP 
    reduction is $\len(\parone)$, and so it suffices to provide a~$n^{\poly(k)}$ procedure. 
    One such procedure consists of converting $\Phi \land
    (\lnot \Psi_1) \land \dots \land (\lnot \Psi_k)$ into DNF, which can be done in~$n^{\poly(k)}$ time, to then test the unsatisfiability of each disjunct. 
    Since $\lnot (x \leq k) \iff x \geq k+1$, every disjunct belongs to $\AC(\sigma)$, 
    and its unsatisfiability can be decided in polynomial time. 
    (Note: taking $k = 1$ provides another argument for solving 
    the inclusion $\rho(\Phi) \land \rho(\Psi)$ required in~\Cref{step:basic_framework_ptime:Item1}.)
    \hfill$\diamond$
\end{exa}

Establishing~\Cref{step:basic_framework_ptime} implies that
the full Boolean algebra (including relative complementation) of $\dnf(\domain)$
has the UXP signature that is required in~\Cref{theorem:PuttingAllTogether}.

\begin{restatable}[Outcome of~\Cref{step:basic_framework_ptime}]{lem}{StepOneImpliesBooleanAlgebra}\label{lemma:Step1-implies-BooleanAlgebra}
    Suppose that: 
    \begin{enumerate}[label=(F\arabic*)]
        \item The
        structure $(\domain,\land,\leq)$ has a $(\rho,\theta)$-UXP signature. 
        \item The structure $(\un(\domain),\leq)$ has a
        $(\un(\rho),\len(\theta))$-UXP~signature.
    \end{enumerate}  
    Then, the
    structure ${(\dnf(\domain),\bot,\top,\lor,\land,-,\leq)}$ has a
    $(\dnf(\rho),\depth(\theta))$-UXP signature.
\end{restatable}

\noindent
The proof of this lemma boils down to the definition of suitable UXP reductions
implementing the binary operations~$\land$, $\lor$ and $-$. We postpone this
proof to~\Cref{section:diffnf}, where we study in depth algorithmic aspects of
the difference normal form.

\vskip3.25ex plus1ex minus.2ex

Moving forward, we now consider projections and universal projections. Again, the
goal is to minimise the efforts needed to add support for these operations. In
this sense, the decision to adopt the difference normal form now becomes
crucial. First, we need to introduce a variant of universal projection which
we call relative universal projection. Given $Z \subseteq A^{m}$, $X \subseteq
A^{n}$ and $\vec i = (i_1,\dots,i_k) \in \vec I$, the \emph{relative universal
projection} $\proj_Z^\forall(\vec i, X)$ of $X$ with respect to $Z$ is defined
as follows, where $M \coloneqq \max(m,n)$:
\begin{align*}
    \proj_Z^\forall(\vec i, X) \coloneqq \{ (a_1,\smalldots,a_{M}) \in A^{M} : & \ \vec a \coloneqq (a_1,\smalldots,a_m) \in \proj(\vec i, Z) \text{ and for all } 
    \vec b \in A^k\\
    &\text{ if }
    \vec a[\vec i \gets \vec b] \in Z \text{ then } 
    (a_1,\smalldots,a_n)[\vec i \gets \vec b] \in X \}. \,
\end{align*}
Informally speaking, and illustrated in \Cref{fig:rel-projection}, the set $\proj_Z^\forall(\vec i, X)$ acts as a universal
projection for the part of $X$ that lies inside~$Z$. Note that $\unproj(\vec i,
X) = \unprojrel{\top}(\vec i, X)$.

\begin{figure}
    \begin{tikzpicture}
        \node (bg) at (0,0) {\includegraphics[scale=0.7]{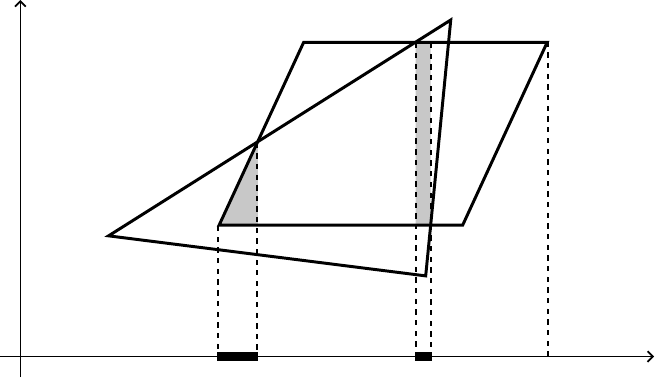}};
        \node (x) at (3.6,-2.2) {$x$};
        \node (y) at (-3.85,2.0) {$y$};
        \node (Z) at (2.2,1.45) {$Z$};
        \node (X) at (-1.8,-0.4) {$X$};

        \node (l) at (-1.3,-2.2) {}; 
        \node (v1) at (-1.05,-2.1) {}; 
        \node (v2) at (1.13,-2.1) {}; 
        \node (r) at (2.62,-2.2) {}; 

        \node (t) at (0.1,-2.4) {$\pi_Z^{\forall}(y,X)$};

        \draw [decorate,decoration={brace,amplitude=5pt,mirror,raise=3.4ex}]
            (l.center) -- (r.center) node[midway,yshift=-2.8em]{$\pi(y,Z)$};

        \draw[->] (t.west) -| (v1.center); 
        \draw[->] (t.east) -| (v2.center); 
    \end{tikzpicture}

    \caption{Illustration of the relative universal projection $\pi_Z^\forall(y,X)$.}
    \label{fig:rel-projection}
\end{figure}

The lemma below outlines a key ``mutual distribution'' property of projection
and relative universal projection over relative complement. In the context of
the difference normal form, this property allows us to disregard complementation
when adding support for quantification.

\begin{restatable}{lem}{LemProjUnivDual}\label{lem:proj_univ_proj_X_setminus_Y}
    Consider $X \subseteq A^n$, $Y \subseteq A^m$, $Z \subseteq A^r$, and 
    $\vec i \in \vec I$. Then,
    \begin{center}
    $\proj(\vec i, X - Y) = \proj(\vec i,X) - \unprojrel{X}(\vec i, Y)$ \ and \
    $\unprojrel{Z}(\vec i, X - Y) = \unprojrel{Z}(\vec i, X) - \proj(\vec i,Y)$.
    \end{center}
\end{restatable} 


\begin{exa}[QBF]\label{example:qbf}
    Consider the formula~$\Psi$ from~\Cref{example:qbf-0}.
    Let $\vec i = (b,c)$.
    Following the equivalences in~\Cref{lem:proj_univ_proj_X_setminus_Y}, we have
    \[ 
        \proj(\vec i, \sem{\Psi}) \ = \ \proj(\vec i,\sem{a \lor b}) - \left(\unprojrel{\sem{a \lor b}}(\vec i,\sem{a \lor (b \land c)}) - \big(\proj(\vec i,\sem{a \land c}) - \unprojrel{\sem{a \land c}}(\vec i\sem{a \land b \land c})\big)\right).
    \]
    It is easy to see that: 
    \begin{itemize}
        \item $\proj(\vec i,\sem{a \lor b})$ corresponds to $\sem{\top}$. In other words, $\exists b,c : a \lor b$ is a valid formula.
        \item $\unprojrel{\sem{a \lor b}}(\vec i,\sem{a \lor (b \land c)})$ corresponds to $\sem{a}$; i.e., $a \iff {(\forall b,c : a \lor b \Rightarrow a \lor (b \land c))}$~is~valid.
        \item $\proj(\vec i,\sem{a \land c})$ corresponds to $\sem{a}$, whereas $\unprojrel{\sem{a \land c}}(\vec i, \sem{a \land b \land c})$ corresponds to $\sem{\bot}$.
    \end{itemize}
    Therefore $\proj(\vec i,\sem{\Psi}) = \sem{\top} - (\sem{a} - (\sem{a} - \sem{\bot})) = \sem{\top}$. Moreover, since~${\proj^\forall(a,\sem{\top}) = \sem{\top}}$, 
    it follows that the formula~$\forall a \exists b \exists c \, \Phi$
    from~\Cref{example:qbf-0} is valid.
    \hfill$\diamond$
\end{exa}

We postpone
the proof of~\Cref{lem:proj_univ_proj_X_setminus_Y} and of the
forthcoming~\Cref{lemma:step-one-two-imply-FO-tract}
to~\cref{appendix:sec4-new}. Below, let us write $\dotproj$ for the restriction
of the projection operator~$\proj$ on inputs $(\vec i, X)$ where $X \in
\domain$,  
and write $\dotunprojrel{}$ for the restriction of the relativised universal
projection $\proj^\forall$ on inputs $(Z, \vec i, X)$ where $Z \in \domain$
and $X \in \un(\domain)$. Thanks to~\Cref{lem:proj_univ_proj_X_setminus_Y},
adding to~\Cref{step:basic_framework_ptime} the following requirement is sufficient to
conclude that $\mcD$ has the UXP signature required
by~\Cref{theorem:PuttingAllTogether}.

\noindent
\begin{minipage}{\linewidth}
\begin{restatable}[Projection and universal projection]{step}{STEPTHREE}\label{step:beforeReducPiUnivToSimpleCasesUxp}
    Establish the following properties:
    \begin{enumerate}[label=(F\arabic*)]
        \setcounter{enumi}{2}
        \item\label{step:beforeReducPiUnivToSimpleCasesUxp:Item1} For every $X \in \domain$
            and $\vec i \in \vec I$, $\dotproj(\vec i, X)$ belongs to $\dnf(\domain)$.
        \item\label{step:beforeReducPiUnivToSimpleCasesUxp:Item2} For every $Z \in
            \domain$, $\vec i \in \vec I$ and $X \in \un(\domain)$,
            $\dotunprojrel{Z}(\vec i, X)$ belongs to $\dnf(\domain)$.
        \item\label{step:beforeReducPiUnivToSimpleCasesUxp:Item3} There is a
            $(\parone \cdot \theta,\depth(\theta))$-UXP reduction
            $(\stdrepr{\vec I} \times \rho,\dnf(\rho))$-implementing~$\dotproj$.
        \item\label{step:beforeReducPiUnivToSimpleCasesUxp:Item4} There is a
            $(\theta \cdot \parone \cdot \len(\theta),\depth(\theta))$-UXP
            reduction $(\rho \times \stdrepr{\vec I} \times
            \un(\rho),\dnf(\rho))$-implementing~$\dotunprojrel{}\!$.
    \end{enumerate}
\end{restatable}
\end{minipage}

\begin{restatable}[Outcome of \Cref{step:basic_framework_ptime} and~\ref{step:beforeReducPiUnivToSimpleCasesUxp}]{lem}{StepOneTwoImplyFOTract}
    \label{lemma:step-one-two-imply-FO-tract}
    Assume~\ref{step:basic_framework_ptime:Item1}--\ref{step:beforeReducPiUnivToSimpleCasesUxp:Item4} to hold.
    Then, the structure~$\mcD$ from~\Cref{theorem:PuttingAllTogether} has a
    $(\dnf(\rho),\depth(\theta))$-UXP signature.
\end{restatable}

\begin{exa}[PA with one variable per inequality, cont.]\label{example:PA-step2}
    We show how to implement \Cref{step:beforeReducPiUnivToSimpleCasesUxp} on the fragment
    of Presburger arithmetic  
    introduced in~\Cref{example-unary-PA}. Recall that $\domain \coloneqq
    \{\sem{\Psi} : \Psi \in \AC(\sigma)\}$, and $\dom(\rho)$
    is the set of systems of constraints of the form $\bigwedge_{i=1}^m x_i \in [\ell_i,u_i]$,  where $\ell_i \in \ZZ \cup \{-\infty\}$ and $u_i \in \ZZ \cup \{+\infty\}$, for every $i \in [1,m]$, 
    plus two elements $\top$ and $\bot$. 
    
    Establishing~\Cref{step:beforeReducPiUnivToSimpleCasesUxp:Item1,step:beforeReducPiUnivToSimpleCasesUxp:Item3}
    is straightforward. Let $\Phi$ be an element of $\dom(\rho)$, and consider a
    vector of variables $\vec x$. Item~(F3) asks to show that the set of
    solutions to $\exists \vec x \, \Phi$ is an element $S \in \dnf(\domain)$.
    The next item asks for an algorithm to compute an element of
    $\dom(\dnf(\rho))$ that represent $S$. In fact, we can compute in polynomial
    time an element of $\dom(\rho)$ that represents~$S$. If $\Phi$ is $\top$ or
    $\bot$, we simply return $\Phi$. Let $\Phi = \bigwedge_{i=1}^m x_i \in
    [\ell_i,u_i]$. If $u_i < \ell_i$ for some $x_i$ in $\vec x$, then $\exists
    \vec x \, \Phi$ is unsatisfiable, and we can simply return $\bot$.
    Otherwise, remove $x_i \in [\ell_i,u_i]$ from the constraints in $\Phi$.
    After repeating this step for all the variables in $\vec x$, the algorithm
    returns the resulting system.

    Let us move to~\Cref{step:beforeReducPiUnivToSimpleCasesUxp:Item2,step:beforeReducPiUnivToSimpleCasesUxp:Item4}. 
    These items ask for an algorithm that given in input $\Phi \in \dom(\rho)$, a vector of variables $\vec x$, and $(\Psi_1,\dots,\Psi_k) \in \dom(\un(\rho))$, returns an element of $\dom(\dnf(\rho))$ 
    representing the set of solutions~$S$ to ${(\exists \vec x \Phi) \land \forall \vec x (\Phi \Rightarrow \Psi_1 \lor \dots \lor \Psi_k)}$. 
    For simplicity, let us suppose $\vec x = (x_{n+1},\dots,x_{m})$, 
    ${\Phi = \bigwedge_{i=1}^m x_i \in [\ell_i, u_i]}$ 
    and ${\Psi_j = \bigwedge_{i=1}^m x_i \in [\ell_{ji},u_{ji}]}$.
    We first present an algorithm that does not yield the correct UXP reduction, to illustrate potential pitfalls when instantiating the framework, and then provide the correct algorithm.
    The ``flawed'' algorithm tries to perform quantifier elimination 
    on $\forall \vec x$ in a na\"ive way:
    
    \noindent 
    \begin{minipage}{\linewidth}~

    {\rm\begin{algorithmic}[1]
        \State rewrite  $\forall \vec x (\Phi \Rightarrow \Psi_1 \lor \dots \lor \Psi_k)$ as $\lnot \exists \vec x (\Phi \land \lnot (\Psi_1) \land \dots \lnot (\Psi_k))$
        \State bring $(\Phi \land \lnot (\Psi_1) \land \dots \lnot (\Psi_k))$ in DNF, obtaining a formula $\bigvee_{j = 1}^M \Phi_j$
        \State for every $j \in [1,M]$ compute a system $\Phi_j'$ in $\dom(\rho)$, equivalent to $\exists \vec x\, \Phi_j$
        \State compute a system $\Phi'$ in $\dom(\rho)$, equivalent to $\exists \vec x\, \Phi$
        \State \textbf{return} $\Phi' - (\bigvee_{j=1}^M \Phi_j')$
    \end{algorithmic}}
    \smallskip
    \end{minipage}
    The problem with this algorithm is that, in the worst case, the number of disjuncts $M$ is $\Omega(n^k)$; 
    and so the parameter $\depth(\parone)$ of the output is not bounded by any function in the parameter $\parone \cdot \parone \cdot \len(\parone)$ of the input. In particular, the number of disjuncts must not depend on $n$.
    Note, however, that the running time $m^{\poly(k)}$ of the algorithm is unproblematic. Our main mistake in designing this algorithm is not using any property of the theory at hand.

    We now provide a correct algorithm. Let $\vec y \coloneqq (x_1,\dots,x_{n})$ denote the variables we are not projecting away.
    A vector $\vec v = (v_1,\dots,v_{n}) \in \ZZ^{n}$ is in the set of solutions $S$ 
    whenever it satisfies $\exists \vec x \Phi$ 
    and the sentence obtained from ${\forall \vec x (\Phi \Rightarrow \Psi_1 \lor \dots \lor \Psi_k)}$ by replacing every variable in $\vec y$ with the corresponding value in $\vec v$ is a tautology. 
    If~$\vec v$ satisfies $\exists \vec x \Phi$, then the latter sentence is 
    equivalent to $\forall \vec x (\Phi' \Rightarrow \Psi_1' \lor \dots \lor \Psi_k')$ where ${\Phi' = \bigwedge_{i=n+1}^m x_i \in [\ell_i, u_i]}$
    and each $\Psi_j'$ is either $\bot$ (if $\vec v$ is not a solution to $\exists \vec x \Psi_j$) or it is $\bigwedge_{i={n+1}}^m x_i \in [\ell_{ji},u_{ji}]$. Crucially, there are only $2^k$ such sentences. This observation yields the following algorithm:

    \noindent 
    \begin{minipage}{\linewidth}
    {\rm\begin{algorithmic}[1]
        \State $T \gets {}$ the empty list 
        \For{$J \subseteq [1,k]$}
            \If{$\forall \vec x ((\bigwedge_{i=n+1}^m x_i \in [\ell_i, u_i]) \implies \bigvee_{j \in J} \bigwedge_{i={n+1}}^m x_i \in [\ell_{ji},u_{ji}])$ is a tautology}\label{algo:pa1ineq:line3}
                \State\label{algo:pa1ineq:line4} let $\ell_i' \coloneqq \max\{\ell_i,\ell_{ji} : j \in J\}$ and $u_i' \coloneqq \min\{u_i,u_{ji} : j \in J\}$, for all $i \in [1,n]$
                \State\label{algo:pa1ineq:line5} append $\bigwedge_{i=1}^n x_i \in [\ell_i',u_i']$ to the list~$T$
            \EndIf
        \EndFor
        \State \textbf{return} $T$ 
        \Comment{$T$ is an element of $\dom(\un(\rho))$}
    \end{algorithmic}}
    \smallskip
    \end{minipage}
    Intuitively, Lines~\ref{algo:pa1ineq:line4} and~\ref{algo:pa1ineq:line5} 
    construct a formula whose solutions satisfy $\exists \vec x \Phi$ and ensure 
    the sentence in Line~\ref{algo:pa1ineq:line3} is a tautology. The \emph{\textbf{for}} loop enumerates all such tautologies. 
    The output $T$ therefore contains $2^k$ systems of constraints; 
    in other words, $\depth(\parone)(T) \leq 2^{(\parone \cdot \parone \cdot \len(\parone))(\Phi, \vec x, (\Psi_1,\dots,\Psi_k))}$.
    Checking that the sentences in Line~\ref{algo:pa1ineq:line3} are tautologies can be done using our ``flawed'' algorithm, which runs in $(m-n)^{\poly(k)}$ time. Hence, the overall running time is bounded by $m^{\poly(k)}$ (assuming $m \geq 2$). We conclude that this algorithm achieves the UXP reduction required in~\Cref{step:beforeReducPiUnivToSimpleCasesUxp:Item1,step:beforeReducPiUnivToSimpleCasesUxp:Item3}. 
    \hfill$\diamond$
\end{exa}



\section{The difference normal form of propositional logic}\label{section:diffnf}
In this section, we explore computational aspects of the difference normal form. 
As a result of this exploration, we will establish~\Cref{lemma:Step1-implies-BooleanAlgebra} 
from~\Cref{section:fo-framework}.
To simplify the presentation, we first focus our study on propositional logic,
to then extend it to first-order theories. Recall that formulae in propositional
logic are from the following grammar: 
\[ 
  \Phi, \Psi \coloneqq \bot \mid \top \mid p \mid \lnot \Phi \mid \Phi \land \Psi \mid \Phi \lor \Psi
\]
where $p$ is an atomic proposition from a countable set $\mathcal{AP}$. We write
$\sem{\Phi}$ for the set of valuations $v \colon \mathcal{AP} \to \{\bot,\top\}$
making $\Phi$ true. Recall that $(\Phi - \Psi) \coloneqq \Phi \land \lnot \Psi$.

We study the problem of bringing formulae in propositional logic into
\emph{strict} difference normal form (SDF). A formula $\Psi \coloneqq \Phi_1 -
(\Phi_2 - ( \dots - (\Phi_{k-1} - \Phi_k)))$ in difference normal form is said
to be strict whenever $\sem{\Phi_k} = \emptyset$ and $\sem{\Phi_{i+1}}
\subsetneq \sem{\Phi_i}$, for every $i \in [1,k-1]$. Note that $\sem{\Psi} =
\emptyset$ if and only if $k = 1$. Every $\Phi_i$ is a negation-free DNF
formula, and so checking the strict entailment $\sem{\Phi_{j}} \subsetneq
\sem{\Phi_i}$ can be done in polynomial time in the size of $\Phi_i$ and
$\Phi_j$: it suffices to iterate through the disjuncts of $\Phi_j$, and check
that each of them is syntactically contained in a disjunct of $\Phi_i$, modulo
associativity and commutativity of $\land$ (e.g., $a \land b$ is syntactically
contained in $b \land c \land a$). We write $\SDF(\Sigma)$ for the set of all
formulae in SDF having atomic propositions from a finite set $\Sigma \subseteq
\mathcal{AP}$. We write $\NFDNF(\Sigma)$ for the set of all formulae in
negation-free DNF having atomic propositions from~$\Sigma$.

\subsection{Operations on formulae in strict difference normal form}
\label{subsection:operations-SDF}

We first study the \emph{cons} operation $(:) \colon \NFDNF(\Sigma) \times
\SDF(\Sigma) \to \SDF(\Sigma)$ that on an input $(\Phi,\Psi)$ returns an SDF
representing $\Phi - \Psi$. Given $\Psi = \Phi_1 - (\dots - \Phi_k)$, this
operator is recursively defined as follows:
\begin{center}
  \scalebox{0.95}{ $\Phi: \Psi \defeq 
    \begin{cases}
      \bot             
      &\text{if } \sem{\Phi} = \emptyset\\
      \Phi - \bot           
      &\text{if } \sem{\Phi_1} = \emptyset\\
      \Phi - (\Phi_1 - (\dots -\Phi_k)) &\text{else if } \sem{\Phi_1} \subsetneq
      \sem{\Phi}\\ 
      \Phi - \big((\Phi \mathbin{\oland} \Phi_1):(\Phi_2 - (\dots -\Phi_k))\big)
      &\text{else if } \sem{\Phi \mathbin{\oland} \Phi_1} \subsetneq
      \sem{\Phi}\\ 
      (\Phi \mathbin{\oland} \Phi_2) : (\Phi_3 - (\dots - \Phi_k)) &\text{else
      if } k \geq 3\\ 
      \bot &\text{otherwise}\end{cases}$ }
\end{center}
where, given $\Phi' = \bigvee_{i=1}^r \phi_i'$ and $\Phi'' = \bigvee_{j=1}^s
\phi_j''$ in $\NFDNF(\Sigma)$ and having $r$ and $s$ disjuncts respectively,
$\Phi' \mathbin{\oland} \Phi''$ is short for the negation-free DNF formula
$\bigvee_{i = 1}^r \bigvee_{j = 1}^s \phi_i' \land \phi_j''$. A simple induction
on $k$ shows that~$(:)$ is well-defined.

\begin{lem}\label{lemma:cons-operator-correct}
  The operator~$(:)$ respects the type ${\NFDNF(\Sigma) \times \SDF(\Sigma) \to
  \SDF(\Sigma)}$. Consider a formula~$\Phi$ in $\NFDNF(\Sigma)$, and $\Psi =
  \Phi_1 - (\dots - \Phi_k)$ in $\SDF(\Sigma)$. Then, ${\Phi : \Psi = \Phi_1' -
  (\dots - \Phi_\ell')}$ with $\ell \leq k+1$, $\sem{\Phi_1'} \subseteq
  \sem{\Phi}$ and, for each $j \in [1,\ell]$, $\Phi_j' = \oland Y_j$ for some
  $Y_j \subseteq \{\bot,\Phi,\Phi_1,\dots,\Phi_k\}$. Furthermore, $\sem{\Phi : \Psi}
  = \sem{\Phi} \setminus \sem{\Psi}$.
\end{lem}

\begin{proof}
  The proof is by induction on the length~$k$ of $\Psi$.
  \begin{description}
    \item[base case: $k = 1$] We have $\Psi = \Phi_1$ and $\sem{\Psi} =
      \emptyset$. By definition of~$(:)$, either $\Phi : \Psi = \bot$ or
      $\Phi : \Psi = \Phi - \bot$, depending on whether $\sem{\Phi} =
      \emptyset$. In both cases, the lemma is satisfied.
    
    \item[induction step: $k \neq 1$] In this case, $\sem{\Phi_1} \neq
    \emptyset$. If $\sem{\Phi} = \emptyset$ then $\Phi : \Psi = \emptyset$, and
    again the lemma is trivially satisfied. Below, assume $\sem{\Phi} \neq
    \emptyset$. We divide the proof following the last three cases in the
    definition of~$(:)$\,:
    \begin{description}
        \item[case: $\sem{\Phi_1} \subsetneq \sem{\Phi}$]
          By definition, $\Phi : \Psi = \Phi - (\Phi_1 - (\dots - \Phi_k))$,
          which belongs to $\SDF(\Sigma)$ and respects all properties of the
          lemma. 

        \item[case: $\sem{\Phi \oland \Phi_1} \subsetneq \sem{\Phi}$ and not $\sem{\Phi_1} \subsetneq \sem{\Phi}$]
          Let $\Psi' \coloneqq \Phi_2 - (\dots - \Phi_k)$. By definition, $\Phi
          : \Psi = \Phi - ((\Phi \oland \Phi_1) : \Psi')$. Since $\Psi'$ has
          length $k-1$, by the induction hypothesis, ${(\Phi \oland \Phi_1)} :
          \Psi' = \Phi_1' -(\dots - \Phi_\ell')$ where~${\ell \leq k}$,
          ${\sem{\Phi_1'} \subseteq \sem{\Phi \oland \Phi_1}}$, and for every $j
          \in [1,\ell]$, $\Phi_j' = \oland Y_j$ for some $Y_j \subseteq \{\bot,
          {(\Phi \oland \Phi_1)},\Phi_2,\dots,\Phi_k\}$. Moreover, $\sem{(\Phi
          \oland \Phi_1) : \Psi'} = \sem{\Phi \oland \Phi_1} \setminus
          \sem{\Psi'}$. Since $\sem{\Phi_1'} \subseteq \sem{\Phi \oland \Phi_1}
          \subsetneq \sem{\Phi}$, the formula $\Phi - (\Phi_1' - (\dots -
          \Phi_\ell'))$ is in $\SDF(\Sigma)$, and all properties required by the
          lemma are satisfied.
        \item[case: not $\sem{\Phi \oland \Phi_1} \subsetneq \sem{\Phi}$ and $k \geq 3$]
          Let $\Psi' \coloneqq \Phi_3 - (\dots - \Phi_k)$. By definition, $\Phi
          : \Psi = (\Phi \oland \Phi_2) : \Psi'$. Since $\Psi'$ has length $k-2$,
          by the induction hypothesis,  $(\Phi \oland \Phi_2) : \Psi' = \Phi_1'
          - (\dots - \Phi_\ell')$ with $\ell \leq k-1$, $\sem{\Phi_1'} \subseteq
          \sem{\Phi \oland \Phi_2}$, $\sem{(\Phi \oland \Phi_2) : \Psi'} =
          {\sem{\Phi \oland \Phi_2} \setminus \sem{\Psi'}}$, and for every $j
          \in [1,\ell]$, $\Phi_j' = \oland Y_j$ for some ${Y_j \subseteq
          {\{\bot,(\Phi \oland \Phi_2), \Phi_3,\dots,\Phi_k\}}}$. Since
          $\sem{\Phi : \Psi} = \sem{(\Phi \oland \Phi_2) : \Psi'}$, to conclude
          the proof of this case it suffices to show that $\sem{\Phi \oland
          \Phi_2} \setminus \sem{\Psi'} = \sem{\Phi} \setminus \sem{\Psi}$. 
          
          By definition of~$\oland$, $\sem{\Phi \oland \Phi_1} \subseteq
          \sem{\Phi}$, and since $\sem{\Phi \oland \Phi_1}$ is not strictly
          included in~$\sem{\Phi}$, we have $\sem{\Phi \oland \Phi_1} =
          \sem{\Phi}$, which implies $\sem{\Phi} \subseteq \sem{\Phi_1}$. From
          the validity ``${A \subseteq B}$ implies ${A \setminus (B \setminus C)
          = A \cap C}$'' we get ${\sem{\Phi} \setminus \sem{\Psi} = \sem{\Phi}
          \cap (\sem{\Phi_2} \setminus \sem{\Psi'})}$. Then, from the validity
          ${A \cap (B \setminus C) = (A \cap B) \setminus C}$ together with the
          definition of $\oland$, we conclude that~$\sem{\Phi} \setminus
          \sem{\Psi} = \sem{\Phi \oland \Phi_2} \setminus \sem{\Psi'}$.
        
        \item[case: not $\sem{\Phi \oland \Phi_1} \subsetneq \sem{\Phi}$ and $k
        = 2$] In this case $\Psi = \Phi_1 - \bot$. By definition $\Phi : \Psi =
        \bot$, and to conclude the proof it suffices to show that $\sem{\Phi}
        \setminus \sem{\Psi} = \emptyset$. This follows from the fact that, as
        in the previous case, we have $\sem{\Phi} \subseteq \sem{\Phi_1}$. 
        \qedhere
      \end{description}
  \end{description}
\end{proof}

We introduce the operations of \emph{union}~$\ccor$, \emph{intersection}~$\cand$
and \emph{difference}~$\cnot$, with type $\SDF(\Sigma) \times \SDF(\Sigma) \to
\SDF(\Sigma)$, that given a pair of SDFs $(\Phi,\Psi)$ compute an SDF
representing $\Phi \lor \Psi$, $\Phi \land \Psi$ and $\Phi - \Psi$,
respectively. These three operators have mutually recursive definitions. Below,
whenever their length is not one (i.e., $\sem{\Phi} \neq \emptyset$ or
$\sem{\Psi} \neq \emptyset$), we assume $\Phi$ and $\Psi$ to be respectively
$\Phi = \Phi_1 - \Phi'$ and $\Psi = \Psi_1 - \Psi'$, with $\Phi_1,\Psi_1 \in
\NFDNF(\Sigma)$ and $\Phi',\Psi' \in \SDF(\Sigma)$. {\allowdisplaybreaks%
\begin{align*}
  \Phi \cand \Psi & := 
    \begin{cases}
      \bot\\
      (\Phi_1 \oland \Psi_1) : (\Phi' \ccor \Psi') 
    \end{cases}
    &\begin{aligned}
      \text{if } \sem{\Phi} = \emptyset \text{ or } \sem{\Psi} = \emptyset,\\
      \hfill\text{otherwise.} 
    \end{aligned}
    \\[3pt]
  \Phi \cnot \Psi & := 
    \begin{cases}
      \Phi\\
      \Phi_1 : (\Phi' \ccor \Psi)
    \end{cases}
    &\begin{aligned}
      \text{if } \sem{\Phi} = \emptyset \text{ or } \sem{\Psi} = \emptyset,\\
      \hfill\text{otherwise.} 
    \end{aligned}
    \\[3pt]
  \Phi \ccor \Psi & := 
    \begin{cases}
      \Psi \\
      \Phi \\
      \Phi_1 : (\Phi' \cnot \Psi) \\
      (\Phi_1 \lor \Psi_1) : \big((\Phi' \ccor \Psi') \cnot ((\Phi_1 \oland \Psi_1) : (\Phi' \cand \Psi'))\big) 
    \end{cases}
    &\begin{aligned}
      \hfill \text{if } \sem{\Phi} = \emptyset,\\
      \hfill \text{else if } \sem{\Psi} = \emptyset,\\ 
      \hfill \text{else if } \sem{\Psi_1} \subseteq \sem{\Phi_1},\\ 
      \hfill \text{otherwise}.
    \end{aligned}
\end{align*}%
}

We show that the operators $\cand$, $-$ and $\ccor$ always yield an element of
$\SDF(\Sigma)$. In particular, the computation of such an element always
terminates. 

\begin{lem}\label{lemma:operation-computation}
  The operators~$\ccor$,~$\cand$,~$\cnot$ are of type ${\SDF(\Sigma) \times
  \SDF(\Sigma) \to \SDF(\Sigma)}$. Given two formulae~$\Phi$ and $\Psi$ in
  $\SDF(\Sigma)$, we have $\sem{\Phi \ccor \Psi} = \sem{\Phi} \cup \sem{\Psi}$,
  $\sem{\Phi \cand \Psi} = \sem{\Phi} \cap \sem{\Psi}$, and $\sem{\Phi \cnot
  \Psi} = \sem{\Phi} \setminus \sem{\Psi}$. Consider $\oplus \in
  \{\ccor,\cand,\cnot\}$, $\Phi = \Phi_1 - (\dots - \Phi_i)$ and $\Psi = \Psi_1
  - (\dots - \Psi_j)$ and $\Phi \oplus \Psi = \Phi_1' - (\dots - \Phi_\ell')$,
  and let ${k \in [1,\ell]}$. We have $\sem{\Phi_k'} = \sem{ \bigvee G_k}$ where
  $G_k$ is a subset of formulae from~$\{\phi,\,\psi,\, \phi \land \psi : {\phi
  \in \{\Phi_1,\dots,\Phi_i\}},\ \psi \in \{\Psi_1,\dots,\Psi_j\}\}$.
\end{lem}

\begin{proof}
  We show the statement by well-founded induction. The domain of the induction
  is $(\SDF(\Sigma) \times \SDF(\Sigma), <)$, whose elements correspond to
  inputs of the operators $\ccor$, $\cand$, $\cnot$, where the order $<$ is
  defined as follows: for every $(X,Y)$ and $(W,Z)$ in $\SDF(\Sigma) \times
  \SDF(\Sigma)$,  
  with $X = X_1 - (\dots - X_{\ell_X})$, ${Y = Y_1 - (\dots - Y_{\ell_Y})}$, $W
  = W_1 - (\dots - W_{\ell_W})$ and $Z = Z_1 - (\dots - Z_{\ell_Z})$,

  \begin{center}
    $(X,Y) < (W,Z)$ if and only if 
    \begin{tabular}[t]{ll}
      (i) & $\sem{Y_1} \subsetneq \sem{Z_1}$, or\\ 
      (ii) & $\ell_X < \ell_W$ and $\sem{Y_1} \subseteq \sem{Z_1}$.  
    \end{tabular}
  \end{center}
  According to~$<$, the only (required) base case corresponds to $(X,Y)$ such
  that $\ell_X = 1$ and $\sem{Y} = \emptyset$. For simplicity, we treat as base
  cases all pairs $(X,Y)$ where $\ell_X = 1$ or $\sem{Y} = \emptyset$.

  Below, we refer to the objects in the statement of the lemma.

  \begin{description}
    \item[base case: $i = 1$] In this case, $\sem{\Phi} = \emptyset$. We have
    ${\Phi \ccor \Psi = \Psi}$, $\Phi \cand \Psi = \bot$ and $\Phi \cnot \Psi =
    \Phi$. In all cases, the statement of the lemma is satisfied. 
    \item[base case: $\sem{\Psi} = \emptyset$]
    We have ${\Phi \ccor \Psi = \Phi}$, $\Phi \cand \Psi = \bot$ and $\Phi \cnot
    \Psi = \Phi$. Again, the statement of the lemma is trivially satisfied.
  \end{description}
  For the induction step, assume $i \geq 2$ and $\sem{\Psi} \neq \emptyset$.
  Note that then $\sem{\Phi} \neq \emptyset$ and ${j \geq 2}$, by definition of
  SDF. For brevity, let $\Phi' \coloneqq \Phi_2 - (\dots - \Phi_i)$ and ${\Psi'
  \coloneqq \Psi_2 - (\dots - \Psi_j)}$. We analyse all operators separately.
  \begin{description} 
    \item[induction step: case $\Phi \cand \Psi$]
      By definition, $\Phi \cand \Psi = (\Phi_1 \oland \Psi_1) : (\Phi' \ccor
      \Psi')$. From the fact that~$\sem{\Psi_2} \subsetneq \sem{\Psi_1}$, and by
      the induction hypothesis, we conclude that the statement of the lemma
      holds for $\Phi' \ccor \Psi'$. In particular, this means that
      \begin{enumerate}
        \item\label{cand:i1} $\sem{\Phi' \ccor \Psi'} = \sem{\Phi'} \cup
        \sem{\Psi'}$, 
        \item\label{cand:i2} $\Phi' \ccor \Psi' = (\Phi_1' - (\dots -
        \Phi_\ell'))$ where, for every $k \in [1,\ell]$, $\sem{\Phi_i'} =
        \sem{\bigvee G_k}$ with $G_k\subseteq \{\phi,\, \psi,\, \phi \land \psi
        : {\phi \in \{\Phi_2,\dots,\Phi_i\}},\ \psi \in
        \{\Psi_2,\dots,\Psi_j\}\}$. 
      \end{enumerate}
      By~\Cref{lemma:cons-operator-correct}, $(\Phi_1 \oland \Psi_1) : (\Phi_1'
      - (\dots - \Phi_\ell')) = (\Psi_1' - (\dots - \Psi_{r}'))$ where 
      \begin{enumerate}
        \setcounter{enumi}{2}
        \item\label{cand:i3} $\sem{(\Phi_1 \oland \Psi_1) : (\Phi_1' - (\dots -
        \Phi_\ell'))} = \sem{\Phi_1 \oland \Psi_1} \setminus \sem{(\Phi_1' -
        (\dots - \Phi_\ell'))}$, 
        \item\label{cand:i4} for all $k \in [1,r]$, $\Psi_k' = \oland Y_k$ with
        $Y_k \subseteq \{\bot,(\Phi_1 \oland
        \Psi_1),\Phi_1',\dots,\Phi_\ell'\}$.
      \end{enumerate}
      From~\Cref{cand:i1,cand:i3}, together with identities of set theory, we
      get $\sem{\Phi \cand \Psi} = \sem{\Phi} \cap \sem{\Psi}$: 
      \begin{align*}
        \sem{\Phi \cand \Psi} 
          &= \sem{\Phi_1 \oland \Psi_1} \setminus (\sem{\Phi'} \cup \sem{\Psi'}) 
            &\text{by~\Cref{cand:i1,cand:i3}}\\ 
          &= (\sem{\Phi_1} \cap \sem{\Psi_1}) \setminus (\sem{\Phi'} \cup \sem{\Psi'}) 
            &\text{definition~of~$\oland$}\\
          &= (\sem{\Phi_1} \setminus \sem{\Phi'}) \cap (\sem{\Psi_1} \setminus \sem{\Psi'}) 
            &\text{set-theoretical validity}\\
            &&\hspace{-20pt} \scalebox{0.9}{$(A \setminus B) \cap (C \setminus D) = (A \cap C) \setminus (B \cup D)$}\\ 
          &= \sem{\Phi} \cap \sem{\Psi}
            &\text{definition~of~$\Phi$ and $\Psi$.}
      \end{align*}
      From~\Cref{cand:i2,cand:i4} and by definition of $\oland$, for every $k
      \in [1,r]$, $\sem{\Psi_k'} = \sem{\bigvee D_k}$ where $D_k$ is a set of
      formulae of the form $\phi_1 \land \dots \land \phi_n$ with $\phi_{p} \in
      \{\Phi_1,\dots,\Phi_i,\Psi_1,\dots,\Psi_j\}$ for every $p \in [1,n]$
      (recall here that $\sem{\Psi_j} = \sem{\Phi_i} = \sem{\bot} = \emptyset$,
      by definition of SDF). Since $\Phi$ and $\Psi$ are in SDF, it is not
      necessary for $\phi_1 \land \dots \land \phi_n$ to contain distinct
      formulae among $\Phi_1,\dots,\Phi_j$, as one of them will imply the
      others. The same is true for formulae among $\Psi_1,\dots,\Psi_j$.
      Therefore, $\phi_1 \land \dots \land \phi_n$ can be simplified into a
      formula of the form $\psi_1$, $\psi_2$ or $\psi_1 \land \psi_2$, with
      $\psi_1 \in \{\Phi_1,\dots,\Phi_i\}$ and $\psi_2 \in
      \{\Psi_1,\dots,\Psi_j\}$. This completes the proof for the case~$\cand$.

      \item[induction step: case $\Phi \cnot \Psi$]
        By definition, $\Phi \cnot \Psi = \Phi_1 : (\Phi' \ccor \Psi)$. Since
        $\Phi'$ has length $i-1$, by the induction hypothesis the statement of
        the lemma holds for $\Phi' \ccor \Psi$. The rest of the proof follows
        similarly to the previous case. In particular, $\sem{\Phi \cnot \Psi} =
        \sem{\Phi} \setminus \sem{\Psi}$ is established thanks to the
        set-theoretical validity $(A \setminus B) \setminus C = A \setminus (B
        \cup C)$ together with $\sem{\Phi_1 : (\Phi' \ccor \Psi)} = \sem{\Phi_1}
        \setminus (\sem{\Phi'} \cup \sem{\Psi})$, which follows by the induction
        hypothesis and~\Cref{lemma:cons-operator-correct}. 

      \item[induction step: case $\Phi \ccor \Psi$, and $\sem{\Psi_1} \subseteq \sem{\Phi_1}$]
        By definition, $\Phi \ccor \Psi = \Phi_1 : (\Phi' \cnot \Psi)$. Once
        more, since $\Phi'$ has length $i-1$, by the induction hypothesis the
        statement of the lemma holds for $\Phi' \cnot \Psi$. The rest of the
        proof follows similarly to the previous cases. In particular, $\sem{\Phi
        \ccor \Psi} = \sem{\Phi} \cup \sem{\Psi}$ follows from the validity ``$C
        \subseteq A$ implies $(A \setminus B) \cup C = (A \setminus (B \setminus
        C))$''. 
      \item[induction step: case $\Phi \ccor \Psi$, and $\sem{\Psi_1}
      \not\subseteq \sem{\Phi_1}$] This case is a bit more involved, so we give
      full details. By definition, 
      \begin{center}
        $\Phi \ccor \Psi = (\Phi_1 \lor \Psi_1) : \big((\Phi' \ccor \Psi') \cnot ((\Phi_1 \oland \Psi_1) : (\Phi' \cand \Psi'))\big).$
      \end{center}
      We proceed with a series of applications
      of~\Cref{lemma:cons-operator-correct} and the induction hypothesis. Since
      $\sem{\Psi_2} \subsetneq \sem{\Psi_1}$, by the induction hypothesis,
      \begin{enumerate}
        \item\label{cor:i1} $\sem{\Phi' \cand \Psi'} = \sem{\Phi'} \cap
        \sem{\Psi'}$,
        \item\label{cor:i2} $\Phi' \cand \Psi' = \Phi_1' - (\dots -
        \Phi_{\ell}')$ where, for all $k \in [1,\ell]$, $\sem{\Phi_k'} =
        \sem{\bigvee G_k}$\\ with $G_k\subseteq \{\phi,\,\psi,\, \phi \land \psi :
        {\phi \in \{\Phi_2,\dots,\Phi_i\}},\ \psi \in \{\Psi_2,\dots,\Psi_j\}
        \}$, 
        \item\label{cor:i3} $\sem{\Phi' \ccor \Psi'} = \sem{\Phi'} \cup
        \sem{\Psi'}$,
        \item\label{cor:i4} $\Phi' \ccor \Psi' = \Phi_1'' - (\dots -
        \Phi_{s}'')$ where, for all $k \in [1,s]$, $\sem{\Phi_k''} =
        \sem{\bigvee D_k}$\\ with $D_k\subseteq \{\phi,\,\psi,\, \phi \land \psi :
        {\phi \in \{\Phi_2,\dots,\Phi_i\}},\ \psi \in \{\Psi_2,\dots,\Psi_j\}
        \}$. 
      \end{enumerate}
      By~\Cref{lemma:cons-operator-correct}, $(\Phi_1 \oland \Psi_1) : (\Phi'
      \cand \Psi') = \Psi_1' - (\dots - \Psi_r')$ where 
      \begin{enumerate}
        \setcounter{enumi}{4}
        \item\label{cor:i5} $\sem{(\Phi_1 \oland \Psi_1) : (\Phi' \cand \Psi')}
        = \sem{\Phi_1 \oland \Psi_1} \setminus \sem{\Phi' \cand \Psi'}$,
        \item\label{cor:i6} for all $k \in [1,r]$, $\Psi_k' = \oland Y_k$ where
        $Y_k \subseteq \{\bot,(\Phi_1 \oland
        \Psi_1),\Phi_1',\dots,\Phi_\ell'\}$, 
        \item\label{cor:i7} $\sem{\Psi_1'} \subseteq \sem{\Phi_1 \oland
        \Psi_1}$.
      \end{enumerate}
      Since~$\sem{\Psi_1} \not\subseteq \sem{\Phi_1}$ we have ${\sem{\Phi_1
      \oland \Psi_1} \subsetneq \sem{\Psi_1}}$. Because of~\Cref{cor:i7}, by
      induction hypothesis we get~$(\Phi_1'' - (\dots - \Phi_{s}'')) \cnot
      (\Psi_1' - (\dots - \Psi_r')) = \Psi_1'' -(\dots - \Psi_{t}'')$, where 
      \begin{enumerate}
        \setcounter{enumi}{7}
        \item\label{cor:i8} $\sem{(\Phi_1'' - (\dots - \Phi_{s}'')) \cnot
        (\Psi_1' - (\dots - \Psi_r'))} = \big(\sem{\Phi_1'' - (\dots -
        \Phi_{s}'')} \setminus \sem{ \Psi_1' - (\dots - \Psi_r')}\big)$,
        \item\label{cor:i9} for every $k \in [1,\ell]$, $\sem{\Phi_k'} =
        \sem{\bigvee E_k}$\\ with $E_k\subseteq \{\phi,\, \psi,\, \phi \land
        \psi : {\phi \in \{\Phi_1'',\dots,\Phi_s''\}},\ \psi \in
        \{\Psi_1',\dots,\Psi_r'\} \}$.
      \end{enumerate}
      By~\Cref{lemma:cons-operator-correct}, $(\Phi_1 \lor \Psi_1) : (\Psi_1'' -
      (\dots - \Psi_t'')) = \widetilde{\Phi}_1 - (\dots - \widetilde{\Phi}_u)$
      where 
      \begin{enumerate}
        \setcounter{enumi}{9}
        \item\label{cor:i10} $\sem{(\Phi_1 \lor \Psi_1) : (\Psi_1'' - (\dots -
        \Psi_t''))} = \sem{\Phi_1 \lor \Psi_1} \setminus \sem{\Psi_1'' - (\dots
        - \Psi_t'')}$,
        \item\label{cor:i11} for all $k \in [1,u]$, $\widetilde{\Psi}_k = \oland
        Z_k$ where $Z_k \subseteq \{\bot,(\Phi_1 \lor
        \Psi_1),\Psi_1'',\dots,\Psi_t''\}$.
      \end{enumerate}
      Now, by~\Cref{cor:i1,cor:i3,cor:i5,cor:i8,cor:i10} and by definition of
      $\oland$, we have 
      \begin{align*}
        \sem{\Phi \ccor \Psi} 
        &= \sem{(\Phi_1 \lor \Psi_1) : \big((\Phi' \ccor \Psi') \cnot ((\Phi_1 \oland \Psi_1) : (\Phi' \cand \Psi'))\big)}\\
        &=(\sem{\Phi_1} \cup \sem{\Psi_1}) \,{\setminus}\,
        \big((\sem{\Phi'} \cup \sem{\Psi'}) 
        \setminus ((\sem{\Phi_1} \cap \sem{\Psi_1}) \setminus (\sem{\Phi'} \cap \sem{\Psi'}))\big)\\
        &=(\sem{\Phi_1} \setminus \sem{\Phi'}) \cup (\sem{\Psi_1} \setminus \sem{\Psi'})\\ 
        &= \sem{\Phi} \cup \sem{\Psi},
      \end{align*}
      where the second to last equivalence follows from the validity ``$B
      \subseteq A$ and $D \subseteq C$ imply $(A \setminus B) \cup (C \setminus
      D) = (A \cup C) \setminus ((B \cup D) \setminus ((A \cap C) \setminus (B
      \cap D)))$''.

      From~\Cref{cor:i2,cor:i4,cor:i6,cor:i9,cor:i11}, we conclude that, for
      every $k \in [1,u]$, $\sem{\widetilde{\Psi}_k} = \sem{\bigvee F_k}$ where
      $F_k$ is a set of formulae of the form $\phi_1 \land \dots \land \phi_n$
      with $\phi_{p} \in \{\Phi_1,\dots,\Phi_i,\Psi_1,\dots,\Psi_j\}$ for every
      $p \in [1,n]$. As explained in the first of the induction steps we
      considered, $\phi_1 \land \dots \land \phi_n$ can be simplified into a
      formula of the form $\psi_1$, $\psi_2$ or $\psi_1 \land \psi_2$, with
      $\psi_1 \in \{\Phi_1,\dots,\Phi_i\}$ and $\psi_2 \in
      \{\Psi_1,\dots,\Psi_j\}$; concluding the proof.
      \qedhere
  \end{description}
\end{proof}

\subsection{Complexity of union, intersection and difference}

We now study the complexity of performing the various operations we have introduced.

\begin{lem}\label{lemma:SDF-longest-chain}
  Let $\Phi_1 - (\dots - \Phi_i)$, $\Psi_1 - (\dots - \Psi_j)$ and  $\Phi_1' -
  (\dots - \Phi_\ell')$ be formulae in $\SDF(\Sigma)$. Suppose that, for every
  $k \in [1,\ell]$, $\sem{\Phi_k'} = \sem{ \bigvee G_k }$, where $G_k \subseteq
  \{\phi,\, \psi,\, {\phi \land \psi} : {\phi \in \{\Phi_1,\dots,\Phi_i\}},\
  \psi \in \{\Psi_1,\dots,\Psi_j\} \}$. Then, $\ell \leq i \cdot j$.
\end{lem}

\begin{proof}
  Define $\Phi_0 \coloneqq \Psi_0 \coloneqq \top$. For $u \in [1,i]$ and $v \in
  [1,j]$, set $\alpha_{u,v} \coloneqq (\Phi_{u-i} - \Phi_{u}) \land (\Psi_{v-1}
  - \Psi_v)$. Every element of $G_k$, and thus every $\Phi_k'$, 
  can be represented as a union of elements from $\{\alpha_{u,v} : (u,v) \neq (1,1) \}$. 
  Indeed, let $\ell \in [0,i]$ and $r \in [0,j]$. 
  Then, $\Phi_\ell \land \Psi_r$
  is equivalent to $(\bigvee_{u=\ell}^{i-1} (\Phi_{u} - \Phi_{u+1})) \land
  (\bigvee_{v=r}^{j-1} (\Psi_v - \Psi_{v+1}))$. (Note that for $\ell = i$ the
  left disjunction simplifies as $\bot$, as expected since $\sem{\Phi_k} = \emptyset$.)
  Expanding this formula into DNF yields a union of elements of the form $\alpha_{u,v}$. 
  Since  $\Phi_1 - (\dots - \Phi_i)$,
  $\Psi_1 - (\dots - \Psi_j)$ belong to $\SDF(\Sigma)$, 
  the sets $\sem{\alpha_{u,v}}$ are pairwise disjoint:
  $\sem{\alpha_{u,v}} \cap
  \sem{\alpha_{u',v'}} = \emptyset$ whenever $(u,v) \neq (u',v')$.
  Finally, because $\Phi_1' - (\dots - \Phi_\ell')$ belongs to $\SDF(\Sigma)$, 
  and each $\Phi_k'$ is a union of
  elements from $\{\alpha_{u,v} : (u,v) \neq (1,1) \}$, 
  we conclude that $k \leq i \cdot j$.
\end{proof}

\begin{lem}\label{lemma:propcalc:BoolAlg}
  There is an algorithm that given $\oplus \in \{\ccor,\cand,\cnot\}$, and 
   ${\Phi = \Phi_1 - (\dots - \Phi_i)}$ and $\Psi = \Psi_1 - (\dots -
  \Psi_j)$ in $\SDF(\Sigma)$, computes $\Phi \oplus \Psi$ in time
  $2^{\poly(i,j)} \cdot \poly(n,\abs{\Sigma})$, where $n$ is the maximum number
  of disjuncts in some of the $\NFDNF(\Sigma)$ formulae
  $\Phi_1,\dots,\Phi_i,\Psi_1,\dots,\Psi_j$. 
\end{lem}

\begin{proof}
  The algorithm, whose pseudocode is given in~\Cref{algorithm:computeBooleanSDF}
  (function~\textsc{Apply}), is simple. It starts by computing the set
  $\mathbb{G}$ of all $G \subseteq \{\phi,\,\psi,\,\phi \oland \psi : \phi \in
  \{\Phi_1,\dots,\Phi_i\}, \psi \in \{\Psi_1,\dots,\Psi_j\} \}$ representing
  formulae that might be needed in order to compute $\Phi \oplus \Psi$,
  according to~\Cref{lemma:operation-computation}. Afterwards
  (function~\textsc{Apply'}), the algorithm simply follows the definitions of
  the operations~$\cand$,~$\cnot$ and~$\ccor$ given
  in~\Cref{subsection:operations-SDF}, making sure at each step to ``normalise''
  the computed formula $\Phi_1' - (\dots - \Phi_k')$ so that
  $\Phi_1',\dots,\Phi_k'$ are formulae given by elements in $\mathbb{G}$. This
  normalisation is done by the function~\textsc{Normalise}. In this function,
  the sentence ``Let $\Phi_1'$ be the representative of $\Phi_1$ in
  $\mathbb{G}$'' indicates that $\Phi_1' = \bigvee G$ for some $G \in
  \mathbb{G}$, and $\sem{\Phi_1'} = \sem{\Phi_1}$. This normalisation is
  required for the $\NFDNF(\Sigma)$ formulae to stay of polynomial size during
  the procedure. The pseudocode of the cons function $(:)$ follows exactly its
  mathematical definition given at the beginning
  of~\Cref{subsection:operations-SDF}, and is thus omitted
  from~\Cref{algorithm:computeBooleanSDF}.

  Since the algorithm simply follows the recursive definitions of the various
  operations, and by~\Cref{lemma:operation-computation} we can restrict the
  formulae of the output to be given by elements in $\mathbb{G}$, its
  correctness is immediate.

\begin{figure}
  \begin{algorithmic}[1]
    \Function{Apply}{$\oplus$, $\Phi$, $\Psi$} \State
    \hspace{-0.55cm}\textbf{input:}  $\oplus$: operation among $\cand$, $\cnot$
    and $\ccor$; \ $\Phi$ and $\Psi$: two formulae in SDF. \State
    \hspace{-0.55cm}\textbf{output:} A formula in SDF equivalent to $\Phi \oplus
    \Psi$.
    \medskip
    \State Let $\Phi = \Phi_1 - (\dots - \Phi_i)$ and $\Psi = (\Psi_1 - (\dots -
    \Psi_j))$. \State Compute the set $\mathbb{G}$ containing all sets $G$ such
    that
    \vspace{3pt}
    \State 
    \begin{minipage}{0.95\linewidth} 
      \begin{itemize}
        \setlength{\itemsep}{0pt}
        \item $G \subseteq \{\phi,\,\psi,\,\phi \oland \psi : \phi \in
        \{\Phi_1,\dots,\Phi_i\}, \psi \in \{\Psi_1,\dots,\Psi_j\} \}$, 
        \item for every $k \in [1,i]$, $\Phi_k$ occurs at most once in a formula
        of $G$\\ 
        \emph{(we consider $\phi$ and $\psi$ to be both occurring in $\phi
        \oland \psi$)}, and 
        \item for every $k \in [1,j]$, $\Psi_k$ occurs at most once in a formula
        of $G$.
      \end{itemize}
    \end{minipage}
    \vspace{2pt}
    \State \textbf{return} \Call{Apply'}{$\oplus$, $\Phi$, $\Psi$, $\mathbb{G}$}
    \EndFunction
    \medskip 
    \Function{Apply'}{$\oplus$, $\Phi$, $\Psi$, $\mathbb{G}$} \label{alg:apply}
    \If{$\sem{\Phi}= \emptyset$ or $\sem{\Psi} = \emptyset$} \Switch{$\oplus$}
    \Case{$\cand$} \textbf{return} a representative of $\bot$ in $\mathbb{G}$
    \EndCase \Case{$\cnot$} \textbf{return} $\Phi$ \EndCase \Case{$\ccor$}
    \textbf{return} ($\sem{\Phi} = \emptyset$ ? $\Psi$ : $\Phi$) \EndCase
    \EndSwitch \EndIf \State Let $\Phi = \Phi_1 - \Phi'$ and $\Psi = \Psi_1 -
    \Psi'$ \State \textbf{var} result \Switch{$\oplus$}
    \Case{$\cand$}\label{line:cand-recurse} result $\gets$ $(\Phi_1 \oland
    \Psi_1) : \text{\Call{Apply'}{$\ccor, \Phi', \Psi', \mathbb{G}$}}$ \EndCase
    \Case{$\cnot$}\label{line:cnot-recurse} result $\gets$ $\Phi_1 :
    \text{\Call{Apply'}{$\ccor,\Phi',\Psi,\mathbb{G}$}}$ \EndCase \Case{$\ccor$}
    \If{$\sem{\Psi_1} \subseteq \sem{\Phi_1}$}\label{line:cor-recurse-one}
    result $\gets$ $\Phi_1 : \text{\Call{Apply'}{$\cnot, \Phi', \Psi,
    \mathbb{G}$}}$ \Else \State result $\gets$ \textsc{Normalise}($(\Phi_1
    \oland \Psi_1) : \text{\Call{Apply'}{$\cand, \Phi', \Psi', \mathbb{G}$}}$,
    $\mathbb{G}$)\label{line:cor-recurse-two}  
    \State result $\gets$ \textsc{Apply'}($\cnot$, \Call{Apply'}{$\ccor$,
    $\Phi'$, $\Psi'$, $\mathbb{G}$}, result, $\mathbb{G}$) \State result $\gets$
    $(\Phi_1 \lor \Psi_1) : \text{result}$
          \label{line:cor-recurse-three} 
        \EndIf \EndCase \EndSwitch \State \textbf{return}
      \Call{Normalise}{result, $\mathbb{G}$} \EndFunction
    \medskip 
    \Function{Normalise}{$\Phi, \mathbb{G}$} \If{$\sem{\Phi} = \emptyset$}
    \textbf{return} a representative of $\bot$ in $\mathbb{G}$ \EndIf \State Let
    $\Phi = \Phi_1 - \Phi'$ \State Let $\Phi_1'$ be a representative of $\Phi_1$
    in $\mathbb{G}$ \State \textbf{return} $\Phi_1' -
    \text{\Call{Normalise}{$\Phi', \mathbb{G}$}}$ \EndFunction
  \end{algorithmic}
  \caption{Computing Boolean operations on SDFs.}
  \label{algorithm:computeBooleanSDF}
\end{figure}

  Let us discuss the runtime of~\textsc{Apply}. First, note that computing
  $\mathbb{G}$ requires $2^{\poly(i, j)}$ time: to construct the sets in
  $\mathbb{G}$, we first iterate through the subsets $A \subseteq
  \{\Phi_1,\dots,\Phi_i\}$ and $B \subseteq \{\Psi_1,\dots,\Psi_j\}$. Every such
  pair $(A,B)$ indicates what formulae appear in the set of $\mathbb{G}$ we are
  constructing. Afterwards, we iterate through all possible injections $f \colon
  A' \to B$ such that~$A' \subseteq A$ and $|A'| \leq |B|$, and for a choice of
  $f$ construct the set $\{ \phi : \phi \in A \setminus A' \} \cup \{ \phi \land
  \psi : \phi \in A' \text{ and } f(\phi) = \psi \} \cup \{ \psi : \psi \in B
  \setminus f(A') \}$. The total number of iterations (and hence of sets in
  $\mathbb{G}$) is bounded by {\allowdisplaybreaks
  \begin{align*} 
    &\underbrace{\sum_{l=0}^i \sum_{r=0}^j \Big({i \choose l} {j \choose r}}_{\text{subsets $A$ and $B$}} \underbrace{\sum_{d=0}^{\min(l,r)} \Big( {l \choose d} {r \choose d} d! \Big)}_{\text{injections }f} \Big)\\
    \leq{}& 3 \cdot \sum_{l=0}^i \sum_{r=0}^j
    \Big( \frac{i! \cdot j!}{(i-l)! \cdot (j-r)!} \Big)
    &\text{using } \sum_{d=0}^{\min(l,r)} \frac{1}{d!} < e\\
    \leq{}& 3 \cdot (i+1)! \cdot (j+1)! \leq 2^{\poly(i,j)}.
  \end{align*}
  }

  Let us now move to the function \textsc{Normalise}, that given $\mathbb{G}$
  computed as above, and a formula $\Phi'' = \Phi_1'' - (\dots - \Phi_k'')$ in
  $\SDF(\Sigma)$ where each $\Phi_\ell''$ ($\ell \in [1,k]$) is a can be
  represented with a set in $\mathbb{G}$, returns a formula $(\bigvee G_1) -
  (\dots - (\bigvee G_k))$ in $\SDF(\Sigma)$ such that $G_1,\dots,G_k \in
  \mathbb{G}$ and $\sem{\Phi_\ell''} = \sem{\bigvee G_\ell}$ for all $\ell \in
  [1,k]$. This function runs in time 
  \begin{align*}
    \underbrace{\vphantom{\poly(i,j,n,m,\abs{\Sigma})}k \cdot 2^{\poly(i,j)}}_{\text{iterate through $\Phi''$ and $\mathbb{G}$}}\
    \cdot\ \underbrace{\poly(i,j,n,m,\abs{\Sigma})}_{\text{comparison of $\NFDNF(\Sigma)$ formulae}}
  \end{align*}
  where $m$ is the maximum number of disjuncts in some $\NFDNF(\Sigma)$ formula
  $\Phi_\ell''$. Note that an upper bound to the number of disjuncts that
  formulae $\bigvee G$, with $G \in \mathbb{G}$, is given by $i \cdot n^2+j$.
  
  Lastly, consider the function \textsc{Apply'}. Given $\mathbb{G}$, an
  operation $\oplus \in \{\cand,\cnot,\ccor\}$, and two formulae $\Phi$ and
  $\Psi$ in $\SDF(\Sigma)$ made of $\NFDNF(\Sigma)$ formulae of the form
  $\bigvee G$ with $G \in \mathbb{G}$, this function returns a formula in
  $\SDF(\Sigma)$ that is equivalent to $\Phi \oplus \Psi$ and that is made of
  $\NFDNF(\Sigma)$ formulae of the form $\bigvee G$ with $G \in \mathbb{G}$.
  Because of~\Cref{lemma:SDF-longest-chain}, we know that all the recursive
  calls done by \textsc{Apply'} return a formula $\Phi_1' - (\dots -
  \Phi_\ell')$ with $\ell \leq i \cdot j$. Because of this, we can bound the
  time required to perform the various calls to $(:)$ and \textsc{Normalise}
  throughout the procedure. For example, in line~\ref{line:cand-recurse},
  computing $(\Phi_1 \oland \Psi_1) :
  \textsc{Apply'}(\ccor,\Phi',\Psi',\mathbb{G})$ once the result of
  $\textsc{Apply'}(\ccor,\Phi',\Psi',\mathbb{G})$ is given takes time
  $\poly(i,j,n,\abs{\Sigma})$.

  Below, given $ \ell \in [1,i \cdot j]$ and $r \in [0,i \cdot j]$, let us write
  $\mathtt{R}(\ell,r)$ for the maximal running time of \textsc{Apply'} on an
  input $(\oplus,\Phi',\Psi',\mathbb{G})$ with $\oplus \in
  \{\cand,\cnot,\ccor\}$, $\mathbb{G}$ as above, $\Phi' = \Phi_1' - (\dots -
  \Phi_\ell')$, and $\Psi' = \Psi_1' - \Psi''$, with $\Psi''$ formula in
  $\SDF(\Sigma)$, such that there are $r$ elements $G \in \mathbb{G}$ satisfying
  $\sem{\bigvee G} \subsetneq \sem{\Psi_1'}$. The fact that $r$ can be
  restricted to be at most $i \cdot j$ follows
  from~\Cref{lemma:SDF-longest-chain}. A simple inspection of \textsc{Apply'}
  yield the following inequalities: {\allowdisplaybreaks
  \begin{align*}
    \mathtt{R}(1,r) & \leq O(1) &\text{case: $\sem{\Phi}= \emptyset$}\\ 
    \mathtt{R}(\ell,0) & \leq O(1) &\text{case: $\sem{\Psi} = \emptyset$}\\ 
    \mathtt{R}(\ell+1,r+1) & \leq \max
      \begin{cases}
        \mathtt{R}(\ell,r) + \poly(i,j,n,\abs{\Sigma})\\
        \mathtt{R}(\ell,r+1) + \poly(i,j,\abs{\Sigma})\\
        2 \cdot \mathtt{R}(\ell,r) + \mathtt{R}(i \cdot j, r)\\ 
        \qquad {}+{} 2^{\poly(i,j)} \cdot \poly(n,\abs{\Sigma})
      \end{cases}
      &\begin{aligned}
        \text{{case: line~\ref{line:cand-recurse}}}\\ 
        \text{{case: line~\ref{line:cnot-recurse} or line~\ref{line:cor-recurse-one}}}\\ 
          \\ 
        \text{{case: lines~\ref{line:cor-recurse-two}--\ref{line:cor-recurse-three}}}.
      \end{aligned}
  \end{align*}
  } Let us write $C(i,j,n,\abs{\Sigma})$ for a function in $2^{\poly(i,j)} \cdot
  \poly(n,\abs{\Sigma})$ that upper bounds all functions $O(1)$,
  $\poly(i,j,n,\abs{\Sigma})$ and $2^{\poly(i,j)} \cdot \poly(n,\abs{\Sigma})$
  appearing in the inequalities above. Note that, in giving an upper bound to
  $\mathbb{R}(\ell,r)$, we can treat $i,j,n$ and $\abs{\Sigma}$ as constants,
  hence below we simply write $C$ instead of $C(i,j,n,\abs{\Sigma})$. We have
  $\mathbb{R}(1,r) \leq C$, $\mathbb{R}(\ell,0) \leq C$ and otherwise, given
  $\ell \geq 1$ and $r \geq 0$, $\mathbb{R}(\ell+1,r+1) \leq C +
  \max(\mathbb{R}(\ell,r+1), 2 \cdot \mathbb{R}(\ell,r) + \mathbb{R}(i \cdot j,
  r))$. A simple check shows that $\mathbb{R}(\ell,r) \leq C \cdot 3^{r \cdot i
  \cdot j + \ell}$. This inequality is clearly true for the base cases
  $\mathbb{R}(1,r)$ and $\mathbb{R}(\ell,0)$, and otherwise
  \begin{align*}
    \mathbb{R}(\ell+1,r+1) 
      & \leq C + \max(\mathbb{R}(\ell,r+1), 2 \cdot \mathbb{R}(\ell,r) + \mathbb{R}(i \cdot j, r))\\ 
      & \leq C + \max(C \cdot 3^{(r+1) \cdot i \cdot j + \ell}, 2 \cdot C \cdot 3^{r \cdot i \cdot j + \ell} + C \cdot 3^{r \cdot i \cdot j + i\cdot j})\\ 
      & 
      \leq C \cdot 3^{(r+1) \cdot i \cdot j + \ell+1}. 
  \end{align*}
  We conclude that, on the formulae $\Phi$ and $\Psi$ of the statement of the
  lemma, $\textsc{Apply'}(\oplus,\Phi,\Psi,\mathbb{G})$ runs in time
  $2^{\poly(i,j)} \cdot \poly(n,\abs{\Sigma})$. Then, taking into account the
  computation of $\mathbb{G}$, $\textsc{Apply}(\oplus,\Phi,\Psi)$ runs in
  $2^{\poly(i,j)} \cdot \poly(n,\abs{\Sigma})$.   
\end{proof}

\subsection{Proof of Lemma~\ref{lemma:Step1-implies-BooleanAlgebra}}

We now rely on the strict difference normal form for propositional logic to 
establish~\Cref{lemma:Step1-implies-BooleanAlgebra}. 
Let us fix a structure $\mcA = (A,\sigma,I)$ and consider its
first-order theory $\FO(\mcA)$. Moreover, let $\rho$ be a representation of
$\domain:=\bigcup_{n\in\NN}\domain_n$, where, for all $n\in\NN$,
$\domain_n\subseteq\powerset(A^n)$ is such that $\sem{\Psi}_{\mcA} \in
\domain_n$ for every $\Psi \in \AC(\sigma)$ having maximum variable $x_n$.
In order to establish~\Cref{lemma:Step1-implies-BooleanAlgebra}, we assume that:
\begin{enumerate}
  \item\label{step:basic_framework_ptime:Item1-bis} The structure
  $(\domain,\land,\leq)$ has a $(\rho,\theta)$-UXP signature; 
  \item\label{step:basic_framework_ptime:Item2-bis} The structure
  $(\un(\domain),\leq)$ has a $(\un(\rho),\len(\theta))$-UXP signature,
\end{enumerate}  
and show that, then, ${(\dnf(\domain),\bot,\top,\lor,\land,-,\leq)}$ has a
$(\dnf(\rho),\depth(\theta))$-UXP signature. 

The functions $\bot$ and $\top$ are
constants, and thus can be implemented with computable functions running in
constant time. Consider a binary function $\oplus \in \{\lor,\land,-\}$. The
pseudocode of the $(\depth(\theta)^2,\depth(\theta))$-UXP reduction that is a
$(\dnf(\rho)^2,\dnf(\rho))$-implementation of $\oplus$ is given
in~\Cref{algorithm:FOBoolOperations}. 
\begin{figure}[t]
  \begin{algorithmic}[1]
    \Require $(u_1,\dots,u_i)$ and $(v_1,\dots,v_j)$ in $\dom(\dnf(\rho))$.
    \Ensure $(w_1,\dots,w_k)$ belonging to $\dom(\dnf(\rho))$, and such that $k \leq (i+1)
    \cdot (j+1)$ and ${\hspace{0.6cm}(\dnf(\rho)(u_1,\dots,u_i)) \oplus
    (\dnf(\rho)(v_1,\dots,v_j)) = \dnf(\rho)(w_1,\dots,w_k)}$.
    \medskip
    \State Let $p_1,\dots,p_i$ and $q_1,\dots,q_j$ be $i+j$ fresh propositional
    symbols \label{algo:FOBool:freshprop} \State Let $\oplus_{\mathbb{B}}$ be
    equal to~$\cand$, $\cnot$ or $\ccor$, depending on $\oplus$ being $\land$,
    $-$ or $\lor$, respectively \State $\Phi \gets (p_1 - (p_1 \land p_2 -
    (\dots - (\bigwedge_{\ell=1}^i p_\ell - \bot))))$ \label{algo:FOBool:Phidef}
    \State $\Psi \gets (q_1 - (q_1 \land q_2 - (\dots - (\bigwedge_{r=1}^j q_r -
    \bot))))$ \label{algo:FOBool:Psidef} \State $\Phi_1' - (- \dots \Phi_k')
    \gets \textsc{Apply}(\oplus_{\mathbb{B}}, \Phi, \Psi)$
    \label{algo:FOBool:Apply} \Comment{\Cref{algorithm:computeBooleanSDF}}
    \State $(w_1,\dots,w_k) \gets ((),\dots,())$ \Comment{to be populated}
    \For{$t$ in $[1,k]$}\label{algo:FOBool:For} \State Let $\Phi_t' =
    \bigvee_{s=1}^{d} \Psi_{s}'$ \For{$s$ in $[1,d]$ such that $\bot$ does not
    occur in $\Psi_{s}'$} \label{algo:FOBool:SecondFor} \State $m$ $\gets$
    element of $\dom(\un(\rho))$ corresponding to $\top$ \For{$\ell \in [1,i]$
    such that $p_\ell$ occurs in $\Psi_s'$} $m \gets u_\ell \land m$
          \label{algo:FOBool:lineoland}
        \EndFor \For{$r \in [1,j]$ such that $q_r$ occurs in $\Psi_s'$} $m \gets
        v_r \land m$ \EndFor \State $w_t \gets m \lor w_t$
        \label{algo:FOBool:linelor} \EndFor \EndFor  
        \State \textbf{return} $(w_1,\dots,w_k)$ \label{algo:FOBool:EndFor}
  \end{algorithmic}
  \caption{Implementation of $\oplus \in \{\lor,\land,-\}$ as a $(\depth(\theta)^2,\depth(\theta))$-UXP reduction.}
  \label{algorithm:FOBoolOperations}
\end{figure}
In a nutshell, given as input two sequences $(u_1,\dots,u_i)$ and
$(v_1,\dots,v_j)$ in $\dom(\dnf(\rho))$, the procedure
in~\Cref{algorithm:FOBoolOperations} computes a sequence $(w_1,\dots,w_k)$
representing $(\dnf(\rho)(u_1,\dots,u_i)) \oplus (\dnf(\rho)(v_1,\dots,v_j))$ by
translating the input sequences into propositional formulae, computing $\oplus$
with respect to propositional calculus, and translating back into an element of
$\dom(\dnf(\rho))$.

\begin{lem}\label{lemma:AlgoFOBoolCorrect}
  The function in~\Cref{algorithm:FOBoolOperations} respects its specification.
\end{lem}

\begin{proof}
  Let $p_1,\dots,p_i$ and $q_1,\dots,q_j$ be the fresh propositional symbols
  introduced in line~\ref{algo:FOBool:freshprop}, and let $\Sigma =
  \{p_1,\dots,p_i,q_1,\dots,q_j\}$. Given $\ell \in [1,i]$, the procedure
  associates to $u_\ell$ the propositional symbol $p_\ell$. Similarly, given $r
  \in [1,j]$, it associates to $v_r$ the propositional symbol $q_r$. Recall that
  elements in $\dom(\dnf(\rho))$ are tuples representing formulae in difference
  normal form that might not be \emph{strict} (which is instead required in the
  notion of~$\SDF(\Sigma)$). For example, the sequence $(u_1,\dots,u_i)$ might
  be such that, for some $k \in [1,i-1]$, $\un(\rho)(u_{k+1}) \not\leq
  \un(\rho)(u_{k})$. However, the following equivalence holds directly from
  set-theoretical validities:
  \[ 
    \dnf(\rho)(u_1,\dots,u_i) = \dnf(\rho)(u_1, \, u_1 \land u_2, \, \dots, \, u_1 \land u_2 \land \dots \land u_i, \bot).
  \]
  At the propositional level, it suffices then to consider the formulae $\Phi$
  and $\Psi$ defined in lines~\ref{algo:FOBool:Phidef}
  and~\ref{algo:FOBool:Psidef}, which are in~$\SDF(\Sigma)$ and correspond to
  $(u_1,\dots,u_i)$ and $(v_1,\dots,v_j)$, respectively.
  By~\Cref{lemma:propcalc:BoolAlg}, the formula $\Phi_1' - (- \dots \Phi_k')$
  computed in line~\ref{algo:FOBool:Apply} corresponds to $\Phi
  \oplus_{\mathbb{B}} \Psi$. Moreover, as described at the beginning of the
  proof of~\Cref{lemma:propcalc:BoolAlg} (and according
  to~\Cref{lemma:operation-computation}), for every $t \in [1,k]$, the formula
  $\Phi_t'$ is equal to $\bigvee G$ for some $G \subseteq \{ \phi,\, \psi,\,
  \phi \oland \psi : \phi \in \{\Phi_1,\dots,\Phi_i, \bot\}, \psi \in
  \{\Psi_1,\dots,\Psi_j, \bot\} \}$. Note that, given $\phi \in
  \{\Phi_1,\dots,\Phi_i\}$ and $\psi \in \{\Psi_1,\dots,\Psi_j\}$, $\phi \oland
  \psi$ is a disjunction of formulae of the form $p_1 \land \dots \land p_\ell
  \land q_1 \land \dots \land q_r$, for some $\ell \in [1,i]$ and $r \in [1,j]$.
  Then, lines~\ref{algo:FOBool:For}--\ref{algo:FOBool:EndFor} simply translate
  back the formulae $p_\ell$ and $q_r$ into $u_\ell$ and $v_r$, respectively,
  and compute the necessary conjunction and disjunctions of these elements of
  $\un(\rho)$. In particular, the operations $\land$ and $\lor$ appearing in
  lines~\ref{algo:FOBool:lineoland}--\ref{algo:FOBool:linelor} are the binary
  functions implementing conjunction and disjunction on elements of $\un(\rho)$.
  Then, the fact that translating back from propositional calculus yield an
  element of representing $(\dnf(\rho)(u_1,\dots,u_i)) \oplus
  (\dnf(\rho)(v_1,\dots,v_j))$ stems directly from the definitions of $\cand$,
  $\cnot$ and $\ccor$, that only rely on set-theoretical tautologies.
\end{proof}

\begin{lem}\label{lemma:AlgoFOBoolRuntime}
  The function in~\Cref{algorithm:FOBoolOperations} is a
  $(\depth(\theta)^2,\depth(\theta))$-UXP reduction.
\end{lem}

\begin{proof}
  Directly from~\Cref{lemma:propcalc:BoolAlg}, the formula~$\Phi_1 - ( - \dots -
  \Phi_k')$ of line~\ref{algo:FOBool:Apply} can be computed in polynomial time
  when the lengths $i$ and $j$ of $(u_1,\dots,u_i)$ and $(v_1,\dots,v_j)$ are
  considered fixed (as it is the case for the parameter $\depth(\theta)$). Note
  moreover that $k \leq i \cdot j$, by~\Cref{lemma:SDF-longest-chain}, and that
  $d$ in line~\ref{algo:FOBool:SecondFor} is bounded in $\poly(i,j)$, as every
  $\Phi_t'$ is equal to $\bigvee G$ for some $G \subseteq \{ \phi,\, \psi,\,
  \phi \oland \psi : \phi \in \{\Phi_1,\dots,\Phi_i, \bot\}, \psi \in
  \{\Psi_1,\dots,\Psi_j, \bot\} \}$ such that for every $k \in [1,i]$ $\Phi_k
  \coloneqq p_1 \land \dots \land p_k$ occurs at most once in a formula of $G$,
  and for every $k \in [1,j]$, $\Psi_k \coloneqq q_1 \land \dots \land q_k$
  occurs at most once in a formula of $G$. Therefore, the \textbf{for} loops of
  lines~\ref{algo:FOBool:For}--\ref{algo:FOBool:linelor} perform overall a
  constant number of iteration, when taking into account the parameter
  $\depth(\theta)$. This implies a constant number of calls to the functions
  computing $\land$ and $\lor$ on elements of $\dom(\un(\rho))$. To conclude the
  proof it suffices then to show that these functions are
  $(\len(\theta)^2,\len(\theta))$-UXP reductions. For $\lor$, the analysis is
  simple: given $u = (a_1,\dots,a_n)$ and $v = (b_1,\dots,b_m)$ in
  $\dom(\un(\rho))$, $u \lor v = (a_1,\dots,a_n,b_1,\dots,b_m)$, and therefore
  $\lor$ is a constant time operation when the lengths $n$ and $m$ are fixed.
  For $\land$, $u \land v$ correspond to the sequence of all $a_\ell \land b_r$
  with $\ell \in [1,n]$ and $r \in [1,m]$. Since $n$ and $m$ are considered
  fixed, this is a constant number of calls to the function computing $\land$ on
  elements of $\dom(\rho)$, which from the properties
  of~\Cref{step:basic_framework_ptime} runs in polynomial time when accounting
  for the parameter~$\theta$.
\end{proof}

To conclude the proof of~\Cref{lemma:Step1-implies-BooleanAlgebra}, it now
suffices to provide a $(\depth(\theta)^2,\parone)$-UXP reduction that
implements~$\leq$. Its pseudocode is given in~\Cref{algorithm:FOLWEQ}. Briefly,
given $(u_1,\dots,u_i)$ and $(v_1,\dots,v_j)$ in $\dom(\dnf(\rho))$, this
function establishes whether $(\dnf(\rho)(u_1,\dots,u_i)) \leq
(\dnf(\rho)(v_1,\dots,v_j))$ by checking if $(\dnf(\rho)(u_1,\dots,u_i)) -
(\dnf(\rho)(v_1,\dots,v_j)) = \emptyset$.

\begin{figure}[t]
  \begin{algorithmic}[1]
    \Require $(u_1,\dots,u_i)$ and $(v_1,\dots,v_j)$ in $\dom(\dnf(\rho))$.
    \Ensure \textbf{true} iff $(\dnf(\rho)(u_1,\dots,u_i)) \leq
    (\dnf(\rho)(v_1,\dots,v_j))$.
    \medskip
    \State $(w_1,\dots,w_k) \gets (u_1,\dots,u_i) - (v_1,\dots,v_j)$ 
    \label{algoleq:lineW}
    \Comment{\Cref{algorithm:FOBoolOperations}} \State $\ell = 1$ \While{$\ell
    \leq k$} 
    \label{algoleq:linewhile}
      \If{$w_\ell \leq ()$} \textbf{return true}\label{algoleq:lineemp}
        \Comment{i.e., $\un(\rho)(w_\ell) = \emptyset$} \ElsIf{$\ell < k$ and
        $w_{\ell} \leq w_{\ell+1}$} $\ell \gets \ell+2$ \label{algoleq:lineseq}
        \Else   
        \ \textbf{return false} \EndIf \EndWhile \State \textbf{return
        true}\label{algoleq:linereturn}
  \end{algorithmic}
  \caption{Implementation of $\leq$ as a $(\depth(\theta)^2,\parone)$-UXP reduction.}
  \label{algorithm:FOLWEQ}
\end{figure}

\begin{lem} 
  The function in~\Cref{algorithm:FOLWEQ} respects its specification.
\end{lem}

\begin{proof}
  By~\Cref{lemma:AlgoFOBoolCorrect}, the sequence $(w_1,\dots,w_k)$ computed in
  line~\ref{algoleq:lineW} belongs to $\dom(\dnf(\rho))$ and represents
  $(\dnf(\rho)(u_1,\dots,u_i)) - (\dnf(\rho)(v_1,\dots,v_j))$. By definition,
  $\dnf(\rho)(w_1,\dots,w_k)$ is equivalent to $\un(\rho)(w_1) - (\un(\rho)(w_2) - (\dots -
  \un(\rho)(w_k)))$. Hence, $\dnf(\rho)(w_1,\dots,w_k) = \emptyset$ if and
  only if one of the following holds
  \begin{itemize}
    \setlength{\itemsep}{2pt}
    \item $k = 0$, i.e.,~$(w_1,\dots,w_k)$ is an empty sequence of elements in
    $\un(\rho)$, 
    \item $\un(\rho)(w_1) = \emptyset$, or
    \item $k \geq 2$, $\un(\rho)(w_1) \leq \un(\rho)(w_2)$ and
    $\dnf(\rho)(w_3,\dots,w_k) = \emptyset$. 
  \end{itemize}
  Lines~\ref{algoleq:linewhile}--\ref{algoleq:linereturn}
  implement this check, leading to the correctness of the procedure.
\end{proof}

\begin{lem}
  The function in~\Cref{algorithm:FOLWEQ} is a $(\depth(\theta)^2,\parone)$-UXP
  reduction.
\end{lem}

\begin{proof}
  By~\Cref{lemma:AlgoFOBoolRuntime}, there is a
  $(\depth(\theta)^2,\depth(\theta))$-UXP reduction computing the sequence
  $(w_1,\dots,w_k)$ in line~\ref{algoleq:lineW}.
  By~\Cref{lemma:AlgoFOBoolCorrect}, ${k \leq (i+1) \cdot (j+1)}$, and thus the
  \textbf{while} loop of line~\ref{algoleq:linewhile} iterates $O(i \cdot j)$
  times. Following lines~\ref{algoleq:lineemp} and~\ref{algoleq:lineseq}, each
  iteration might require a comparison between two elements of
  $\dom(\un(\rho))$. Because of the properties required
  by~\Cref{step:basic_framework_ptime}, this comparison can be implemented by a
  $(\len(\theta)^2,\parone)$-UXP reduction.   
  Putting all together, we conclude that the function in~\Cref{algorithm:FOLWEQ}
  is a $(\depth(\theta)^2,\parone)$-UXP reduction.
\end{proof}

\allowdisplaybreaks[1]
\section{Proofs Lemmas~\ref{lem:proj_univ_proj_X_setminus_Y} 
  and~\ref{lemma:step-one-two-imply-FO-tract} 
  from~Section~\ref{section:fo-framework}}
\label{appendix:sec4-new}

We present the proofs of~\Cref{lem:proj_univ_proj_X_setminus_Y,lemma:step-one-two-imply-FO-tract} 
deferred from \Cref{section:fo-framework}. Given a set $A$, $Z \subseteq A^j$ and $k \in \NN$, we define $Z^c \coloneqq A^j \setminus Z$ and write
\[ 
    \upscale{Z}^k \coloneqq Z \times A^{\max(0,k-j)} = 
    \{(v_1,\dots,v_{\max(j,k)}) \in \seq(A) :
    (v_1,\dots,v_j) \in Z\}
\]
for the set of tuples obtained from $Z$ by ``appending''
$\max(0,k-j)$ dimensions.

\LemProjUnivDual* 

\begin{proof}
    Let $M \coloneqq \max(m,n)$ and $\vec i = (i_1,\dots,_k)$.
    For the first equivalence:
    \begin{align*}
        \proj(\vec i, X - Y)
            &=\{ \gamma \in A^M : \exists \vec a \in A^k \text{ s.t. } \gamma[\vec i \gets \vec a] \in (X - Y) \}\\
            &= \compl{\{ \gamma \in A^M : \forall \vec a \in A^k \text{ s.t. } \gamma[\vec i \gets \vec a] \not\in (X - Y) \}}\\
            &= \compl{\{ \gamma \in A^M : \forall \vec a \in A^k \text{ s.t. } \gamma[\vec i \gets \vec a] \in (\compl{X} \lor Y) \}}\\
            &= \compl{\Big(\{ \gamma \not\in \upscale{\proj(\vec i, X)}^m : \forall \vec a \in A^k \text{ s.t. } \gamma[\vec i \gets \vec a] \in (\compl{X} \lor Y) \} \, \lor\\ 
            &\qquad \{ \gamma \in \upscale{\proj(\vec i, X)}^m : \forall \vec a \in A^k \text{ s.t. } \gamma[\vec i \gets \vec a] \in (\compl{X} \lor Y) \}\Big)}\\
            &= \compl{\big(\compl{\proj(\vec i, X)} \lor \unprojrel{X}(\vec i, Y) \big)}\\
            &=\proj(\vec i, X) - \unprojrel{X}(\vec i, Y).
    \end{align*}
    where in the second-last equality we have used the fact that if $\gamma \notin \upscale{\proj(\vec i, X)}^m$ then for every $\vec a \in A^k$ we have $\gamma[\vec i \gets \vec a] \not\in \upscale{(\compl{X})}^m$.
    
    The second equivalence is proven in a similar way. 
    Let $N \coloneqq \max(M,r)$.
    {\allowdisplaybreaks%
    \begin{align*}
          & \unprojrel{Z}(\vec i, X - Y)\\
      ={} & \{ \gamma \in \proj(\vec i, \upscale{Z}^N) : 
          \forall \vec a \in A^k\!\!,\,
          \gamma' \coloneqq \gamma[\vec i \gets \vec a] \in \upscale{Z}^N \text{ implies }
          \gamma' \in \upscale{(X - Y)}^N \}\\
      ={} & \{ \gamma \in \proj(\vec i, \upscale{Z}^N) : 
          \forall \vec a \in A^k \text{ and }
          \gamma' \coloneqq \gamma[\vec i \gets \vec a], (\gamma' \in \upscale{Z}^N \text{ implies } \gamma' \in \upscale{X}^N)\\
      & \hphantom{\{ \gamma \in \proj(\vec i, \upscale{Z}^N) :}
          \text{ and }
          (\gamma' \in \upscale{Z}^N \text{ implies } \gamma' \in \upscale{(\compl{Y})}^N)\}\\
      ={} & \{ \gamma \in \proj(\vec i, \upscale{Z}^N) : 
          \forall \vec a \in A^k \text{ and }
          \gamma' \coloneqq \gamma[\vec i \gets \vec a], \gamma' \in \upscale{Z}^N \text{ implies } \gamma' \in \upscale{X}^N\} \\
      & \land  
          \{ \gamma \in \proj(\vec i, \upscale{Z}^N) : \forall \vec a \in A^k \text{ and }
          \gamma' \coloneqq \gamma[\vec i \gets \vec a], \gamma' \in \upscale{Z}^N \text{ implies } \gamma' \in \upscale{\compl{Y}}^N\}
          \\
      ={} &\upscale{\unprojrel{Z}(\vec i, X)}^N \land
            \compl{\{ \gamma \in A^N : \gamma \not\in \proj(\vec i, \upscale{Z}^N), \text{ or } \gamma[\vec i \gets \vec a] \in (\upscale{Z}^N \land \upscale{Y}^N) \text{ for some } \vec a \in A^k \}}\\
      ={} & \upscale{\unprojrel{Z}(\vec i, X)}^N \land \upscale{\Big(\compl{\big(\compl{\proj(\vec i, Z)} \lor \proj(\vec i, Z \land Y) \big)}\Big)}^N\\
      ={} & \unprojrel{Z}(\vec i, X) \land \compl{\big(\compl{\proj(\vec i, Z)} \lor \proj(\vec i, Z \land Y) \big)} &\hspace{-20cm}\text{by 
            definition~of } \land \text{ and } N\\
      ={} & \unprojrel{Z}(\vec i, X) \land \big(\proj(\vec i, Z) \land \compl{\proj(\vec i, Z \land Y)} \big)\\
      ={} &\unprojrel{Z}(\vec i, X) - \proj(\vec i, Z \land Y) 
            &\hspace{-20cm}
            \text{as } \proj(\vec i, Z) \subseteq \unprojrel{Z}(\vec i, X)
            \text{ by def.~of } \unprojrel{\_}\\
      ={} &\unprojrel{Z}(\vec i, X) - \proj(\vec i, Y) 
            &\hspace{-20cm}
            \text{again from } \proj(\vec i, Z) \subseteq \unprojrel{Z}(\vec i, X) 
          \hspace{0.85cm}\qedhere
    \end{align*}
    }
\end{proof}

Before moving to the proof of~\Cref{lemma:step-one-two-imply-FO-tract} 
we need the following intermediate result.

\begin{lem}\label{lem:simplify_univ_proj_rel}
    Let $X \subseteq A^n$, 
    $Z_1,\ldots,Z_r$ with $Z_j \subseteq A^{m_j}$, and 
    $\vec i = (i_1,\dots,i_k) \in \vec I$.
    Let $Z \coloneqq {\bigvee}_{j=1}^r Z_j$. Then,
    $\unprojrel{Z}(\vec i, X)= {\bigwedge}_{j=1}^r\Big(\unprojrel{Z_j}(\vec i, X) \lor \big(\proj(\vec i, Z) - \proj(\vec i, Z_j)\big)\Big)$.
\end{lem}
\begin{proof}
    Let $M \coloneqq \max_{i=1}^r(n,m_i)$.
    The lemma follows from a simple calculation:
    \begin{align*}
      & \unprojrel{Z}(\vec i, X)\\
      = {} & 
        \{ \gamma \in \proj(\vec i, \upscale{Z}^n) : \forall \vec a \in A^k \text{ and }
        \gamma' \coloneqq \gamma[\vec i\gets \vec a],
        \gamma' \in \upscale{Z}^n \text{ implies }
        \gamma' \in \upscale{X}^M \}\\
      = {} &  
        \proj(\vec i, \upscale{Z}^n) \land \{ \gamma \in A^M : \text{for every } j \in [1,r], \text{ every } \vec a \in A^k \text{ and }
        \gamma' \coloneqq \gamma[\vec i\gets \vec a],\\
      & \hphantom{\proj(\vec i, \upscale{Z}^n) \land \{ \gamma \in A^M : } 
        \ \ \gamma' \in \upscale{Z_j}^M \text{ implies }
        \gamma' \in \upscale{X}^M \}\\
      = {} & \proj(\vec i, \upscale{Z}^n) \land \bigwedge_{j=1}^r\{ \gamma \in A^M : \text{for every } \vec a \in A^k \text{ and }
            \gamma' \coloneqq \gamma[\vec i\gets \vec a],\\[-10pt]
      & \hphantom{\proj(\vec i, \upscale{Z}^n) \land \bigwedge_{j=1}^r\{ \gamma \in A^M}
            \ \ \gamma' \in \upscale{Z_j}^M \text{ implies }
            \gamma' \in \upscale{X}^M \}\\
      = {} & 
        \bigwedge_{j=1}^r \Big(\proj(\vec i, \upscale{Z}^n) \land \{ \gamma \in A^M : \forall \vec a \in A^k \text{ and }
        \gamma' \coloneqq \gamma[\vec i\gets \vec a],\\[-10pt] 
      & \hphantom{\bigwedge_{j=1}^r \Big(\proj(\vec i, \upscale{Z}^n) \land \{ \gamma \in A^M : }
        \ \ \gamma' \in \upscale{Z_j}^M \text{ implies }
        \gamma' \in \upscale{X}^M \} \Big)\\
            &\hspace{1cm}\text{\textit{(Note: if $\gamma \notin\proj(\vec i, \upscale{Z_j}^M)$ then $\gamma' \notin \upscale{Z_j}^M$ for all $\vec a \in A^k$ and 
            $\gamma' \coloneqq \gamma[\vec i\gets \vec a]$.)}}\\
      = {} & 
        \bigwedge_{j=1}^r \Big(\proj(\vec i, \upscale{Z}^n) \land \big(\compl{\proj({\vec i, Z_j})} \lor \{ \gamma \in \proj(\vec i, \upscale{Z_j}^M) : \forall \vec a \in A^k \text{ and }
            \gamma' \coloneqq \gamma[\vec i\gets \vec a]\\[-10pt]
            & \hphantom{\bigwedge_{j = 1}^r \Big(\proj(\vec i, \upscale{Z}^n) \land \big(\compl{\proj({\vec i, Z_j})} \lor \{ \gamma \in \proj(\vec i, \upscale{Z_j})} \gamma' \in \upscale{Z_j}^M \text{ implies }
            \gamma' \in \upscale{X}^M \} \big) \Big)\\
        = {} &  
          \bigwedge_{j = 1}^r \Big(\proj(\vec i, \upscale{Z}^n) \land \big(\compl{\proj({\vec i, Z_j})} \lor \unprojrel{Z_j}(\vec i, X)\big) \Big)\\
        = {} & 
          \bigwedge_{j = 1}^r \Big((\proj(\vec i, \upscale{Z}^n) \land \unprojrel{Z_j}(\vec i, X)) \lor \big(\proj(\vec i, \upscale{Z}^n) \land \compl{\proj({\vec i, Z_j})}\big) \Big)\\
        = {} & 
           \bigwedge_{j = 1}^r \Big(\unprojrel{Z_j}(\vec i, X) \lor \big(\proj(\vec i, Z) \land \compl{\proj({\vec i, Z_j})}\big) \Big)\\
        = {} & 
          \bigwedge_{j = 1}^r \Big(\unprojrel{Z_j}(\vec i, X) \lor \big(\proj(\vec i, Z) - \proj({\vec i, Z_j})\big) \Big).
          \hspace{5.5cm}\qedhere
    \end{align*}
\end{proof}

We recall~\Cref{step:beforeReducPiUnivToSimpleCasesUxp} of the framework.

\STEPTHREE*

\StepOneTwoImplyFOTract*

\begin{proof}
    By~\Cref{lemma:Step1-implies-BooleanAlgebra}, 
    the structure 
    $\mcD' \coloneqq (\dnf(\domain),\bot,\top,\lor,\land,-,\leq)$ 
    has a $(\dnf(\rho), \depth(\theta))$-UXP  
    signature.
    It suffices to show that there are $(\parone \cdot \depth(\theta),\depth(\theta))$-UXP reductions implementing $\proj$ and $\unproj$ with respect to the representation $\dnf(\rho)$.

    Consider the projection~$\proj$.
    We describe a $(\parone \cdot \depth(\theta),\depth(\theta))$-UXP reduction 
    that given $\vec x \in \dom(\dnf(\rho))$
    computes $\vec y \in \dom(\dnf(\rho))$ satisfying $\dnf(\rho)(\vec y) = \proj(\vec i, \dnf(\rho)(\vec x))$. 
    Let $\vec i \in \vec I$
    and~$\vec x = (\vec x_1,\dots, \vec x_\ell) \in \dom(\dnf(\rho))$, where each $\vec x_r$ is in $\dom(\un(\rho))$, and it is thus of the form $(x_{r,1},\dots,x_{r,j_r})$ with every $x_{r,j} \in \dom(\rho)$. The reduction computes~$\vec y$ as follows:
    
    \begin{algorithmic}[1]
      \For{$r \in [1,\ell]$}
        \If{$r$ is odd}
          \ $\vec x_r' \gets \bigvee_{k=1}^{j_r} \dot\pi(\vec i, x_{r,k})$ 
        \Else 
          \ $\vec x_r' \gets \bigwedge_{k=1}^{j_{r-1}}\Big(\dotunproj_{x_{r-1,k}}(\vec i, \vec x_r) \lor \big(\vec x_{r-1}' - \dot\proj(\vec i, x_{r-1,k})\big)\Big)$
        \EndIf
      \EndFor
      \State \textbf{return} $\vec x_1' - (\vec x_2' - \dots ( \ldots - \vec x_{\ell}'))$
    \end{algorithmic}
    Above, note that all calls to the projections $\dot\pi$ and $\dotunproj$ are with respect to elements in~$\dom(\rho)$ and $\dom(\un(\rho))$, respectively, and thus return elements of $\dom(\dnf(\rho))$, accordingly to the first two properties in~\Cref{step:beforeReducPiUnivToSimpleCasesUxp}.
    The correctness of the procedure stems directly from~\Cref{lem:proj_univ_proj_X_setminus_Y}
    and~\Cref{lem:simplify_univ_proj_rel}.
    For its complexity, note that because of the parameter $\depth(\theta)$, the lengths $\ell,j_1,\dots,j_\ell$ are to be considered constant. 
    Hence, the procedure above boils down to a constant number of calls to suitable UXP reductions with respect to the parameter $\depth(\theta)$, and it is therefore a $(\parone \cdot \depth(\theta),\depth(\theta))$-UXP reduction.

    The arguments are analogous for~$\unproj$. In particular, given $\vec i \in \vec I$ and $\vec x = (\vec x_1,\dots, \vec x_\ell) \in \dom(\dnf(\rho))$, the procedure for $\unproj$ simply computes $\dotunproj_{\top}(\vec i, \vec x_1) - \proj(\vec i, (\vec x_2,\dots, \vec x_\ell))$.
\end{proof}

\allowdisplaybreaks[0]

\section{Instantiation of the framework to weak linear real arithmetic}\label{sec:affine_subspaces}

In this section, we provide a first non-trivial instantiation of our framework. 
We consider weak linear real arithmetic
(weak~LRA), the first-order theory of the structure $(\RR,0,1,+,=)$, 
and show that its $k$ negations satisfiability problem lies in PTIME.
In fact, since the first-order theories of 
weak~LRA over the reals and rationals are known 
be elementary equivalent, for technical convenience this section considers 
the structure ${\mathcal{Q} = (\QQ,0,1,+,=)}$ instead.

\subsection{Setup}
According to~\Cref{theorem:PuttingAllTogether}, instantiating the framework
requires first to define the domain $\domain$, its representation $\rho$ and the
change of representation $F \colon \AC(\sigma) \to \dom(\rho)$. In weak LRA,
conjunctions of atomic formulae are systems of affine equations, which over
$\QQ$ are known to define \emph{affine subspaces} (AS). We define $\domain_n$ be
the set of all affine subspaces of~$\QQ^n$. Then, $\domain = \bigcup_{n \in \NN}
\domain_n$ is the set of all affine subspaces over $\QQ^n$, for some $n$. 

Following our notation, let $\stdrepr{\QQ}$ be the (canonical) representation of
the rational numbers as pairs $(n,d)$ where $n \in \dom(\stdrepr{\ZZ})$ and $d
\in \dom(\stdrepr{\NN})$, and $\stdrepr{\QQ}(n,d) =
\frac{\stdrepr{\ZZ}(n)}{\stdrepr{\NN}(d)}$. Note that the structure
$(\QQ,+,\times,-,/,=,\leqslant)$ has a $(\stdrepr{\QQ},\parone)$-UXP signature
since arithmetic operations on $\QQ$ are in PTIME. Vectors are simply represented as
tuples: $\stdrepr{\QQ^n}:=\stdrepr{\QQ}^n$ for all $n$. To ease the
presentation, we do not make a distinction between $\QQ^n$ and
$\dom(\stdrepr{\QQ^n})$. We represent the affine subspaces by an offset and a
basis (i.e., a set of linearly independent vectors) of the linear part.
Formally, for every $n\in\NN$, if $v_0$ represents a vector in $\QQ^n$, and
$v_1,\ldots,v_k$ represent linearly independent vectors in $\QQ^n$, then we let
\[
    \rho_{\AS}(n,v_0,\cdot\cdot\cdot,v_k)\coloneqq\stdrepr{\QQ^n}(v_0)+\Span_{\QQ}\set{\stdrepr{\QQ^n}(v_1),\cdot\cdot\cdot,\stdrepr{\QQ^n}(v_k)}.
\]
Here, $+$ stands for the Minkowski sum, 
and given ${\vec v_1,\dots,\vec v_k \in \QQ^n}$, $\Span_{\QQ}\set{\vec v_1,
\dots, \vec v_k} \coloneqq \{\lambda_1 \cdot \vec v_1 + \dots + \lambda_k \cdot
\vec v_k : \lambda_1,\dots,\lambda_k \in \QQ\}$ is a \emph{linear subspace}.
We call the vectors
$v_0,\dots,v_k$ above a $V$-representation of the affine subspace
$\rho_{\AS}(n,v_0,\cdot\cdot\cdot,v_k)$. 
The presence of the dimension $n$ as first argument of
$\rho_{\AS}$ is required to map the vectors to a shifted lattice
of the right dimension. 
The definition
of $\rho_{\AS}$ is surjective but not injective, as every
affine subspace admits an infinite number of
\mbox{$V$-representations}. 
We also consider a constant symbol $\varnothing$ to represent the
empty set, i.e., $\rho_{\AS}(\varnothing) \coloneqq \emptyset$, and assume
$\emptyset \in \domain_n$ for every $n \in \NN$. 
Note that $(0,())$ stands instead for the only vector space of dimension $0$;
hence in particular $\rho_{\AS}(0,()) = \{()\} \neq \rho_{\AS}(\varnothing)$.
The map $\rho_{\AS}$ is the map $\rho$ required by~\Cref{theorem:PuttingAllTogether}.  

Finally, as a preliminary step for instantiating the framework, 
we must provide a change of representation
$F$ from conjunctions of atomic formulae to elements in
$\dom(\rho_{\AS})$. This function is essentially given by
the PTIME algorithm to compute the echelon form of a matrix
with rational entries:

\begin{propC}[\cite{Edmonds1967SystemsOD}]\label{prop:echelon_form}
    There is a PTIME algorithm computing the echelon form,
    along with the transformation matrices, of a given
    matrix with rational entries.
\end{propC}

Since $F$ runs in PTIME, the parameter $\xi$
from~\Cref{theorem:PuttingAllTogether} equals $\parone$. To instantiate the
framework, it now suffices to show that the structure~$\mcD$ from~\Cref{theorem:PuttingAllTogether} has a
$(\dnf(\rho_{\SL}),\depth(\parone))$-UXP signature, by establishing
\Cref{step:basic_framework_ptime}
and~\ref{step:beforeReducPiUnivToSimpleCasesUxp}.  






\subsection{Requirement~\ref{step:basic_framework_ptime}: Boolean connectives}
The next
lemma establishes~\Cref{step:basic_framework_ptime:Item1}
of~\Cref{step:basic_framework_ptime}. Therein, $+$ stands for the Minkowski sum, which we later need to implement subsequent parts of the framework.

\begin{restatable}{lem}{BasicFrameworkAssumptionAS}\label{lem:basic_framework_ptime_as}
    The structure~$(\domain,\land,+,\leq)$ has a $(\rho_{\AS},\parone)$-UXP
    signature.
\end{restatable}
\begin{proof}
    Observe that it suffices to show how to compute $\leq$, $\land$ and $+$ on
    affine subspaces of the same dimension: given $(n,\vec v_0,\dots, \vec v_k)$
    and $(m, \vec w_0,\dots, \vec w_j)$ from $\dom(\rho_{\AS})$ with $n > m$, we
    append $n-m$ many zeros to each vector~$\vec w_i$ in order to obtain an affine
    subspace of dimension~$n$. 
    (In the case of subspaces having same dimension, $\leq$ and $\land$ correspond 
    to $\subseteq$ and $\cap$, so we often use these symbols interchangeably.)
    Moreover, 
    all operations are straightforward when at least one of the two arguments 
    is~$\varnothing$.
    Below, let 
    $X \coloneqq (n,\vec v_0,\dots,\vec v_k)$ 
    and $Y \coloneqq (n,\vec w_0,\dots, \vec w_j)$.
    For brevity, we write $V$ and $W$ for the matrices having 
    as columns $\vec v_1,\dots,\vec v_k$ and $\vec w_1, \dots, \vec w_j$, respectively.

    The algorithm for testing $\rho_{\AS}(X) \subseteq \rho_{\AS}(Y)$ is simple: 
    the inclusion holds if and only if $\vec v_0 \in \vec w_0 +W \cdot \QQ^j$ 
    and, for every $i \in [1,k]$, $\vec v_i \in W \cdot \QQ^j$.
    These membership queries boil down to solving systems of equations over $\QQ$, 
    which can be done in polynomial time by, e.g., Gaussian elimination 
    or by appealing to~\Cref{prop:echelon_form}.

    Computing the Minkowski sum is also simple: we have $\rho_{\AS}(X) +
    \rho_{\AS}(Y) = (\vec v_0 + \vec w_0) + U$ with 
    $U \coloneqq \Span_\QQ(\vec v_1,\dots,\vec v_k, \vec w_1,\dots, \vec w_j)$. 
    Hence, the algorithm
    consists in computing a basis (of independent vectors) 
    $\vec t_1,\dots, \vec t_r$ 
    for the vector space $U$, and return 
    $(n,(\vec v_0 + \vec w_0),\vec t_1,\dots, \vec t_r)$. 
    The basis can be directly extracted from the echelon form 
    of the matrix $\begin{bmatrix}V & W\end{bmatrix}$, 
    so the algorithm runs in polynomial time (\Cref{prop:echelon_form}).

    For the intersection $\rho_{\AS}(X) \cap \rho_{\AS}(Y)$, we have: 
    \begin{align*}
        \rho_{\AS}(X)  \cap \rho_{\AS}(Y) &
            = \set{\vec v_0 + V \cdot \vec y : \vec y \in \QQ^k \text{ and } \vec v_0 + V \cdot \vec y = \vec w_0 + W \cdot \vec z \text{ for some } \vec z \in \QQ^j}\\
            & = \vec v_0 + \set{V \cdot \vec y :  \vec y \in \QQ^k \text{ and } V \cdot \vec y - W \cdot \vec z = \vec w_0 - \vec v_0 \text{ for some } \vec z \in \QQ^j}\\
            &= \vec v_0 + V \cdot \pi(U),
    \end{align*}
    where $\pi(\vec y, \vec z) \coloneqq \vec y$ is the projection from $\QQ^{k+j}$ to $\QQ^k$, 
    and $U \coloneqq \set{(\vec y, \vec x): V \cdot \vec y- W \cdot \vec z = \vec w_0 - \vec v_0}$ is an affine subspace
    whose representation $(k+j,\vec u_0,\dots,\vec u_\ell) \in \dom(\rho_{\AS})$ can be computed in polynomial time by putting the matrix $\begin{bmatrix}V^T&-W^T\end{bmatrix}^T$ in echelon form with \Cref{prop:echelon_form}. 
    Let $\vec u_i'$ be the vector obtained from $\vec u_i$ by projecting away the last $j$ entries. We have $\rho_{\AS}(X)  \cap \rho_{\AS}(Y) = (\vec v_0 + V \cdot \vec u_0') + 
    U$, with $U \coloneqq \Span_\QQ((V \cdot \vec u_1'),\dots,(V \cdot \vec u_\ell'))$.
    The algorithm computes a basis~$\vec t_1,\dots,\vec t_r$ for the vector space $U$ 
    (as done for the Minkowski sum), 
    and returns $(n,(\vec v_0 + V \cdot \vec u_0'),\vec t_1,\dots,\vec t_r)$.
\end{proof}

\Cref{step:basic_framework_ptime:Item2} requires an algorithm for inclusion testing of unions of affine subspaces represented as~$\un(\rho_{\AS})$.
The algorithm relies on the following well-known result.

\begin{restatable}{lem}{InclusionASinUofAS}\label{lem:inclusion_AS_in_UofAS}
    Let $X$ be an affine subspace and $Y=\bigcup_{i\in I} Y_i$ be a
    (finite) union of affine subspaces. Then, $X\subseteq Y$ if and only if
    $X\subseteq Y_i$ for some $i\in I$.
\end{restatable}
\begin{proof}
    This can be shown by induction on the size of $I$. If $|I|=1$ then the result is trivial. 
    By induction, suppose the result to be true up to a certain size $k$ and let $I$ of size $k+1$.
    Let $i_0\in I$ and $J=I\setminus\set{i_0}$.
    Assume that $X\subseteq \bigcup_{i\in I} Y_i$. If $X\subseteq Y_{i_0}$ then the result is proven.
    Otherwise, $X\nsubseteq Y_{i_0}$ and we will show that $X\subseteq\bigcup_{j\in J}Y_j$
    which will conclude by induction.

    Pick $x\in X\setminus Y_{i_0}$ and let $y\in X$ be arbitrary.
    The case $y\in \bigcup_{j\in J}Y_j$ is trivial. Otherwise, it must be the case that
    $y\in Y_{i_0}$ since $X\subseteq \bigcup_{i\in I} Y_i$. In particular, $x\neq y$ because $x\notin Y_{i_0}$.
    Since $x,y\in X$ and $x\neq y$, the affine line $L$ that passes through $x$ and $y$ is contained
    in $X$ ($X$ being an affine subspace) and contains infinitely many points. As a result, there exists $i_1\in I$ such
    that $L\cap Y_{i_1}$ contains at least two points. But $Y_{i_1}$ is an affine subspace so $Y_{i_1}$
    contains $L$. This implies that $i_0\neq i_1$ because otherwise we would have $x\in L\subseteq Y_{i_1}=Y_{i_0}$
    which is not possible. Therefore $i_1\in J$ and $y\in L\subseteq Y_{i_1}\subseteq \bigcup_{j\in J}Y_j$.
\end{proof}

\Cref{lem:inclusion_AS_in_UofAS} allows reducing 
inclusion testing of union of affine subspaces 
to inclusion testing on affine subspaces, 
which is in polynomial time by~\Cref{lem:basic_framework_ptime_as}.
\Cref{step:basic_framework_ptime:Item2}~follows.

\begin{lem}\label{lem:framework_assumption_UofAS}
    The structure
    $(\un(\domain),\leq)$ has
    a $(\un(\rho_{\AS}),\parone)$-UXP signature.
\end{lem}

\subsection{Requirement~\ref{step:beforeReducPiUnivToSimpleCasesUxp}: Projections and universal projections}
We now move to the second Requirement of the framework, which adds support for
projections and universal projections. Establishing the
Items~\ref{step:beforeReducPiUnivToSimpleCasesUxp:Item1}
and~\ref{step:beforeReducPiUnivToSimpleCasesUxp:Item3} of
\Cref{step:beforeReducPiUnivToSimpleCasesUxp} is trivial; in $V$-representation, (existential) projection is a simple operation.
Indeed, given $X \in
\dom(\rho_{\AS})$ and $\vec i \in \vec I$, $\dotproj(\vec i, X)$ can be computed
by simply crossing out the entries of all vectors of $X$ corresponding to the
indices in $\vec i$. The resulting vectors $\vec v_0,\dots,\vec v_\ell$ are such that
$\dotproj(\vec i, \rho_{\AS}(X)) =
\stdrepr{\QQ^n}(\vec v_0)+ V$ with $V \coloneqq \Span_{\QQ}\set{\stdrepr{\QQ^n}(\vec v_1),\cdots,\stdrepr{\QQ^n}(\vec v_\ell)}$.
It then suffices to compute, starting from $\vec v_1,\dots,\vec v_\ell$, a basis for
$V$,
which can be done in polynomial time following~\Cref{prop:echelon_form}. 
The following result is thus immediate. 

\begin{restatable}{lem}{ProjAS}\label{lem:proj_AS}
    The structure $(\domain,(\dotproj,\vec I))$ has a $(\rho_{\AS},\parone)$-UXP signature.
\end{restatable}


On the contrary, universal projections are not \emph{a
priori} easy to compute. 
Our main observation is that affine subspaces and their unions have very special properties that allow us to express
the universal projection in terms of projections. 
For simplicity, below we index entries in vectors starting from one, and instead of considering projections over arbitrary vectors of
indices $\vec i$, we assume $\vec i = [1,k]$ for some $k \in \NN$ so
that~$\unprojrel{Z}(\vec i, X)$ projects over the first $k$ dimensions. This is
w.l.o.g.,~ as we can reorder components appropriately. So, let~$\unprojrel{Z}(k.X) \coloneqq \unprojrel{Z}([1,k], X)$, 
and $\proj(k,X) \coloneqq \proj([1,k],X)$.  
Given a set $X\subseteq\QQ^n$ and $k\leqslant n$, it will be useful
to introduce the following set, for any $\vec p \in\QQ^{n-k}$:
\[
    X_{\vec p} := \set{\vec t\in\QQ^k:(\vec p,\vec t)\in X}.
\]
Note that the value of $k$ is implicit in the notation, but it will always be clear from the context.
We start with an auxiliary lemma.

\begin{lem}\label{lem:slice_of_AS}
    Let $X \subseteq \QQ^n$ be an affine subspace, and $k \leq n$. There is a linear subspace $V\subseteq\QQ^k$
    such that for every~$\vec p \in\pi(k,X)$, 
    there is $\vec b \in \QQ^k$ s.t.~$X_{\vec p} = \vec b + V$. 
    Furthermore, there is a polynomial time algorithm 
    that given in input an element of $\dom(\rho_{\AS})$ representing~$X$, 
    and $k$ (encoded in unary), computes an element of $\dom(\rho_{\AS})$ representing $V$.
\end{lem}
\begin{proof}
    If $\pi(k,X) = \emptyset$, then the statement allows to take an arbitrary linear space~$V \subseteq \QQ^k$.
    Otherwise, consider $\vec p_0 \in \pi(k,X)$. Then $X_{\vec p_0} \neq \emptyset$ is an affine subspace, hence
    $X_{\vec p_0}= \vec t_0 + V$ for some $t_0 \in \QQ^k$ and $V \subseteq \QQ^k$ a linear subspace.
    We show that $V$ is the linear subspace in the statement of the lemma.
    Let $\vec p \in \pi(k,X)$, and $\vec t' \in \QQ^k$ such that $(\vec p ,\vec t')\in X$.
    Denote by $\Lin(X)$ the \emph{linear part} of $X$, i.e., $X = \vec v + \Lin(X)$ for some vector $\vec v$. 
    The following chain of equivalences shows that, for every $\vec x \in \QQ^k$, 
    $\vec x \in X_{\vec p}$ if and only if $\vec x \in \vec t' + V$.
    \begin{align*}
        \vec x\in X_{\vec p}
            &\Leftrightarrow (\vec p, \vec x)\in X\\
            &\Leftrightarrow (\vec p, \vec x)-(\vec p,\vec t')\in \Lin(X)   &&\text{since $X$ is an affine subspace and $(\vec p,\vec t')\in X$}\\
            &\Leftrightarrow (\vec 0,\vec x- \vec t')\in \Lin(X)\\
            &\Leftrightarrow (\vec p_0, \vec t_{0})+(\vec 0,\vec x- \vec t')\in X&&\text{since $X$ is an affine subspace and $(\vec p_0, \vec t_{0})\in X$}\\
            &\Leftrightarrow (\vec p_0, \vec t_{0} + \vec x- \vec t')\in X\\
            &\Leftrightarrow \vec t_{0} + \vec x- \vec t' \in X_{\vec p_0}\\
            &\Leftrightarrow \vec x- \vec t' \in V&&\text{since $X_{\vec p_0} = \vec t_0 + V$}\\
            &\Leftrightarrow \vec x \in \vec t' + V
    \end{align*}
    Consider a representation $(n,\vec v_0,\dots,\vec v_\ell)$ of $X$.
    In order to compute $V$, we first compute (an arbitrary) $\vec p_0 \in \pi(k,X)$, and then consider the quantified system of equalities 
    \[ 
        \exists y_1,\dots,y_\ell : \ 
        \begin{bmatrix}
        \vec x\\ 
        \vec p_0
        \end{bmatrix}
        = \vec v_0 + \vec v_1 \cdot y_1 + \dots \vec v_n \cdot y_\ell.
    \] 
    Note that this is a formula fo weak LRA. 
    Following~\Cref{prop:echelon_form} and~\Cref{lem:proj_AS}, 
    we already know how to compute an affine subspace 
    for the quantifier-free part of this formula (which has variables $\vec x,y_1,\dots,y_\ell$), to then project away the coordinates 
    corresponding to $y_1,\dots,y_\ell$. 
    The result is a family of vectors $\vec w_0,\dots,\vec w_r$ 
    such that $X_{\vec p_0} = \vec w_0 + V$ with $V \coloneqq \Span_{\QQ}\set{\vec w_1,\dots,\vec w_r}$. Therefore, it suffices to return $(k,\vec 0,\vec w_1,\dots,\vec w_r)$ as a representation of $V$.
\end{proof}

\Cref{lem:slice_of_AS} allows us to tackle relative universal projections~$\pi_Z^\forall(k,X)$ in the cases where both $Z$ and $X$ are affine subspaces.

\begin{restatable}{lem}{RelUnivProjASAS}\label{lem:rel_univ_proj_AS_AS}
    Let $X$ and $Z$ be affine subspaces, and let $k \in \NN$. 
    Then either $\unprojrel{Z}(k,X)=\emptyset$ or
    $\unprojrel{Z}(k,X)=\proj(k,X\land
    Z)$. Moreover, there is an algorithm that given in input $X,Z \in \dom(\rho_{\AS})$
    and $k \in \NN$ in unary, returns $Y \in \dom(\rho_{\AS})$ 
    such that $\rho_{\AS}(Y) = \pi_{\rho_{\AS}(Z)}^\forall(k,\rho_{\AS}(X))$. 
    The algorithm runs in polynomial time.
\end{restatable}

\begin{proof}
    Following the discussion at the beginning of the proof of~\Cref{lem:basic_framework_ptime_as}, it is easy to see that the only 
    interesting case is when $X$ and $Z$ are affine subspaces having the same dimension.
    First note that by definition, we must have $\univ{Z}(k,X) \subseteq \pi(k,X \cap Z)$.
    Indeed, if $\vec x \in \univ{Z}(k,X)$ then $\vec x \in\pi(k,Z)$ so there is $\vec t$ such that $(\vec x, \vec t) \in Z$
    and then $(\vec x, \vec t)\in X$ so $(\vec x, \vec t)\in X \cap Z$.
    Consider then the reverse inclusion.
    By \Cref{lem:slice_of_AS}, there are linear subspaces
    $V$ and $W$ such that
    \[
        \forall \vec p \in \pi(k,Z),\exists \vec b_{\vec p} \in \QQ^k, X_{\vec p} = \vec b_{\vec p} + V
        \qquad\text{and}\qquad
        \forall \vec p \in\pi(k,X),\exists \vec g_{\vec p} \in \QQ^k, Z_{\vec p} = \vec g_{\vec p} + W.
    \]
    It follows that
    \begin{align*}
        \univ{Z}(k,X)
            &=\set{\vec x \in \pi(k,Z) : \forall \vec t \in \QQ^k, (\vec x,\vec t) \in Z\Rightarrow (\vec x,\vec t)\in X}\\ 
            &=\set{ \vec x \in\pi(k,Z): Z_{\vec x} \subseteq X_{\vec x}}\\
            &=\set{\vec x \in \pi(k,Z): \vec g_{\vec x}- {\vec b}_{\vec x} + W \subseteq V}\numberthis\label{eq:rel_univ_proj_AS:eq_univ}.
    \end{align*}
    We now distinguish two cases: 
    \begin{enumerate}
        \item If $W \not\subseteq V$ then it cannot be the case that 
        $\vec g_{\vec x} - \vec b_{\vec x} + W \subseteq V$. This is true for every
        $\vec x \in \pi(k,Z)$, and therefore $\univ{Z}(k,X)=\emptyset$.
        \item If $W \subseteq V$, then consider $\vec x \in \pi(k, X \cap Z)$.
        There is $\vec t$ such that $(\vec x, \vec t) \in X \cap Z$, that is,
        $\vec t \in X_{\vec x} \cap Z_{\vec x}$. In particular, 
        $\vec t = \vec b_{\vec x}+ \vec v = \vec g_{\vec x} + \vec w$ for some 
        $\vec v \in V$
        and $\vec w \in W$. But then 
        $\vec g_{\vec x} - \vec b_{\vec x} = \vec v-\vec w \in V$ since 
        $W \subseteq V$ and $W$ is a linear subspace.
        Therefore, 
        $\vec g_{\vec x} - \vec b_{\vec x} + W\subseteq \vec g_{\vec x} 
        - \vec b_{\vec x} +V = V$. This shows that $\vec x \in \univ{Z}(k,X)$
        by Equation~\eqref{eq:rel_univ_proj_AS:eq_univ}.
    \end{enumerate}
    Algorithmically, given $X,Z \in \dom(\rho_{\AS})$, one computes an element
    of $\dom(\rho_{\AS})$ representing
    $\pi_{\rho_{\AS}(Z)}^\forall(k,\rho_{\AS}(X))$ as follows. First, if one
    among $X$ or $Z$ is~$\varnothing$, the algorithm simply
    returns~$\varnothing$. Otherwise, it computes representations of the linear
    subspaces $V$ and $W$ above, by using the algorithm
    from~\Cref{lem:slice_of_AS}. 
    By~\Cref{lem:basic_framework_ptime_as}, testing $W \subseteq V$ 
    can be performed in polynomial time. 
    If the inclusion is not true, the algorithm returns $\emptyset$. 
    Otherwise, it returns $\pi(k, X \land Z)$, 
    which can again be computed in polynomial time by~\Cref{lem:proj_AS}.
\end{proof}

We extend~\Cref{lem:rel_univ_proj_AS_AS} to union of affine subspaces,
completing~\Cref{step:beforeReducPiUnivToSimpleCasesUxp}:

\begin{restatable}{lem}{RelUnivProjASUoAS}\label{lem:rel_univ_proj_AS_UofAS}
    Let $Z$ be an affine subspace, and $X=\bigvee_{j=1}^m X_j$, with ${X_j}$ affine subspace. 
    Let $k \in \NN$. 
    Then,
    $\unprojrel{Z}(k,X)=\bigvee_{j=1}^m
    \unprojrel{Z}(k,X_j)$. Moreover, there is an 
    algorithm that given in input $Z \in \dom(\rho_{\AS})$, $X \in \dom(\un(\rho_{\AS}))$
    and $k \in \NN$ in unary, returns $Y \in \dom(\un(\rho_{\AS}))$ 
    such that $\un(\rho_{\AS})(Y) = \pi_{\rho_{\AS}(Z)}^\forall(k,\un(\rho_{\AS})(X))$. 
    The algorithm runs in polynomial time.
\end{restatable}

\begin{proof}
    Again, without loss of generality we assume $Z$ and each $X_j$ to have the same dimension.
    Note that
    \begin{align*}
        \univ{Z}(k,X)
            &=\set{\vec x \in \pi(k,Z) : Z_{\vec x} \subseteq X_{\vec x} } =
            \set{\vec x \in \pi(k,Z): X_{\vec x} \subseteq \bigcup\nolimits_{j=1}^{m}\,(X_j)_{\vec x}}.
    \end{align*}
    Recall that $Z_{\vec x}$ and $(X_j)_{\vec x}$ are affine subspaces (\Cref{lem:slice_of_AS}). By~\Cref{lem:inclusion_AS_in_UofAS}, $Z_{\vec x} \subseteq \bigcup_{j=1}^m (X_j)_{\vec x}$
    if and only if $Z_{\vec x} \subseteq (X_j)_{\vec x}$ for some $j \in [1,m]$. Hence,
    \begin{align*}
        \univ{Z}(k,X)
            & =\set{\vec x\in\pi(k,Z): X_{\vec x} \subseteq (X_j)_{\vec x} \text{ for some } j \in [1,m]}\\
            &=\bigcup\nolimits_{j=1}^m\set{x \in \pi(k,Z): Z_{\vec x} \subseteq (X_j)_{\vec x}} = \bigcup\nolimits_{j=1}^m \univ{Z}(k,X_j).
    \end{align*}
    Algorithmically, given $Z \in \dom(\rho_{\AS})$ and $(X_1,\dots,X_m) \in \dom(\un(\rho_{\AS}))$,
    by~\Cref{lem:rel_univ_proj_AS_AS} we can compute in polynomial time a representation $Y_j \in \dom(\rho_{\AS})$ of $\unprojrel{\rho_{\AS}(Z)}(k,\rho_{\AS}(X_j))$, for every $j \in [1,m]$. 
    Then, it suffices to return $Y \coloneqq (Y_1,\dots,Y_m)$.
\end{proof}

The main result of the section follows:

\begin{thm}
    The $k$ negations satisfiability problem
    for weak~LRA is in PTIME.
\end{thm}

\begin{proof}
    \Cref{lem:basic_framework_ptime_as,lem:framework_assumption_UofAS} 
    imply \Cref{step:basic_framework_ptime}, 
    and
    \Cref{lem:proj_AS,lem:rel_univ_proj_AS_UofAS} imply~\Cref{step:beforeReducPiUnivToSimpleCasesUxp}. 
    By~\Cref{lemma:step-one-two-imply-FO-tract}, 
    ~$\mcD$ has a
    $(\dnf(\rho),\depth(\parone))$-UXP signature.
    Then, the theorem follow from~\Cref{theorem:PuttingAllTogether}.
\end{proof}
\section{Instantiation of the framework to weak linear integer arithmetic}
\label{sec:shifted_lattices} 
In this section, we apply our framework  
to show that the $k$ negations satisfiability problem for weak Presburger
arithmetic (weak PA), i.e.~the FO theory of the structure $\mathcal{Z} =
(\ZZ,0,1,+,=)$, is in PTIME. Our presentation follows quite closely of~\Cref{sec:affine_subspaces}, although the details are now more intricate, 
particularly those involving universal projection.

\subsection{Setup} We define the domain $\domain$, its representation $\rho$ and
the change of representation $F \colon \AC(\sigma) \to \dom(\rho)$ required
by~\Cref{theorem:PuttingAllTogether}. In weak PA, conjunctions of atomic
formulae are systems of affine equations, which over $\ZZ$ define \emph{shifted
(integer) lattices} (SL), which are not necessarily fully dimensional. We
let~$\domain_n$ be the set of all shifted lattices of $\ZZ^n$, so that $\domain$
is the set of all shifted lattices of~$\ZZ^n$ for some $n$. We represent
elements in $\domain$ with the standard representation of shifted lattice as a
\emph{base} point together with a collection of linearly independent vectors
(the \emph{periods} of the lattice). Recall that we write $\nu_{\ZZ^n}$ for the
canonical representation of $\ZZ^n$ (see~\Cref{sec:represent-and-complex}).
Formally, we define the representation~$\rho_{\SL}$ for shifted lattices as
follows. For every $n\in\NN$, if $v_0$ represents a vector in $\ZZ^n$, and
$v_1,\ldots,v_k$ represent linearly independent vectors in $\ZZ^n$, then
$(n,v_0,\dots,v_k) \in \dom(\rho_{\SL})$ and
\[
  \rho_{\SL}(n,v_0,\cdot\cdot\cdot,v_k) 
  \coloneqq 
  \stdrepr{\ZZ^n}(v_0)+\Span_{\ZZ}\set{\stdrepr{\ZZ^n}(v_1),\cdot\cdot\cdot,\stdrepr{\ZZ^n}(v_k)},
\]
which is a shifted lattice in $\domain_n$. As in~\Cref{sec:affine_subspaces}, $+$ stands for the Minkowski
sum, and given ${\vec v_1,\dots,\vec v_k \in \ZZ^n}$, $\Span_{\ZZ}\set{\vec v_1,
\dots, \vec v_k} \coloneqq \{\lambda_1 \cdot \vec v_1 + \dots + \lambda_k \cdot
\vec v_k : \lambda_1,\dots,\lambda_k \in \ZZ\}$ is a (non-shifted)
\emph{lattice}. 
We use a constant symbol $\varnothing$ to represent the
empty lattice, i.e., $\rho_{\SL}(\varnothing) \coloneqq \emptyset$, and assume
$\emptyset \in \domain_n$ for every $n \in \NN$.

Below, to ease the presentation, we do not make a distinction between elements
of $\ZZ^n$ and elements of $\dom(\stdrepr{\ZZ^n})$, that is we assume $\ZZ$ to
be the set of integers encoded in binary (hence, all algorithm we give assume
integers represented via~$\stdrepr{\ZZ^n}$). Because of this (usual) assumption,
occurrences of the map $\stdrepr{\ZZ^n}$ are often omitted, and $\rho_{\SL}$ is
seen as a function taking as input $n \in \ZZ$ together with linearly
independent vectors of $\ZZ^n$ (or $\varnothing$).

A polynomial time function $F$ allowing to change representation from
conjunctions of atomic formulae of~$\FO(\mathcal{Z})$ to elements in
$\dom(\rho_{\SL})$ can be obtained due to the Hermite normal form of an integer
matrix being computable in polynomial time. See~\cite[Chapter~4]{Schrijver99}
for an introduction to the Hermite normal form. Briefly, recall that the Hermite
normal form $H \in \ZZ^{n \times d}$ of a matrix $A \in \ZZ^{n \times d}$ is
unique and has (among others) the following properties:
\begin{itemize}
  \item $H$ is lower triangular (so, its non-zero columns are linearly
  independent) and every pivot of a non-zero column is positive,
  \item $H = A \cdot U$ for some unimodular matrix $U \in \ZZ^{d \times d}$ (i.e.~a matrix with determinant $\pm 1$), 
  \item $H$ generates the same (non-shifted) lattice as $A$, i.e., $H \cdot
  \ZZ^d = A \cdot \ZZ^d$. 
\end{itemize}

\begin{propC}[\cite{KannanB79}]\label{prop:hnf_snf}
  There is a polynomial time algorithm to compute the Hermite normal
  form~$H$, along with the transformation matrix~$U$, of a given
  matrix~${A \in \ZZ^{n \times d}}$.
\end{propC}

\begin{lem}\label{lemma:change-of-representation}
  There is a polynomial time function $F$ that given in input a system of
  equations $A \cdot \vec x = \vec b$ in $d$ variables, returns $\varnothing$ if
  the system is unsatisfiable, and otherwise it returns a tuple 
  $(d,\vec v_0,\dots,\vec v_k)$ where $\vec v_0 \in \ZZ^d$, $\vec v_1,\dots,\vec v_k \in \ZZ^d$ are $k$
  linearly independent vectors, and 
  $\{ \vec x \in \ZZ^d : A \cdot \vec x = \vec b \} 
  = \vec v_0 + \Span_{\ZZ}\set{\vec v_1,\dots, \vec v_k}$.
\end{lem}

\begin{proof}
  The function $F$ relies on the algorithm of~\Cref{prop:hnf_snf} 
  to compute $H$ in Hermite normal form and the unimodular matrix $U$ such that ${H = A \cdot U}$. 
  We have ${\{ \vec x \in \ZZ^d : A \cdot \vec x = \vec b \}} = \{ U \cdot \vec y : \vec y \in \ZZ^d \text{ and } H \cdot \vec y = \vec b\}$, as setting $\vec x = U \cdot \vec y$ yields $\vec y = U^{-1} \cdot \vec x$ and $H \cdot U^{-1} = A$.
  Let $\vec y = (y_1,\dots,y_d)$.
  Since $H$ is triangular, we have $H = [\vec b_1 \mid \dots \mid \vec b_j \mid \vec 0 \mid \dots \mid \vec 0]$ where the columns $\vec b_1,\dots,\vec b_j$ are non-zero. 
  Moreover, if $A \cdot \vec x = \vec b$ has a solution, then there is a unique way to generate $\vec b$ as a linear combination of $\vec b_1,\dots,\vec b_j$. Finding the values $y_1^*,\dots,y_j^*$ for $y_1,\dots,y_j$ that generate $\vec b$ is trivial (briefly, $y_1^*$ is the only integer making the first non-zero entry of $\vec b$ equal to the entry of $\vec b_1 \cdot y_1^*$ in the same position; update then $\vec b$ to $\vec b - \vec b_1 \cdot y_1^*$ and recursively find values for $y_2,\dots,y_j$).
  If $y_1^*,\dots,y_j^*$ do not exist, then $F$ outputs $\varnothing$.
  Otherwise, let $U = [\vec u_1 \mid \dots \mid \vec u_d]$.
  We set $\vec v_0 \coloneqq \vec u_1 \cdot y_1^* + \dots + \vec u_j \cdot y_j^*$.
  Note that, because of the shape of $H$, the values assigned to the variables $y_{j+1},\dots,y_d$ do not influence the satisfaction of $H \cdot \vec y = \vec b$.
  Then, $F$ returns 
  $(d, \vec v_0, \vec u_{j+1},\dots,\vec u_{d})$.
\end{proof}

Since the function $F$ runs in polynomial time, the parameter $\xi$
in~\Cref{theorem:PuttingAllTogether} equals~$\parone$, and to instantiate the framework
we need to show that~$\mcD$ has a
$(\dnf(\rho_{\SL}),\depth(\parone))$-UXP signature, by establishing
\Cref{step:basic_framework_ptime}
and~\ref{step:beforeReducPiUnivToSimpleCasesUxp}.

\subsection{Requirement~\ref{step:basic_framework_ptime}, Item~\ref{step:basic_framework_ptime:Item1}: 
  the structure~\texorpdfstring{$(\domain,\land,\leq)$}{(D,conj,leq)} 
  has a~\texorpdfstring{$(\rho_{\SL},\parone)$}{(rhoSL,1)}-UXP signature} 
Briefly, both the
problems of computing intersections and testing inclusions
for two shifted lattices represented as in $\rho_{\SL}$
reduces to solving systems of linear equations over $\ZZ$,
which can be done in polynomial time again thanks
to~\Cref{prop:hnf_snf} (alternatively,~\Cref{lemma:change-of-representation}). 
We now formalise these reductions.

\begin{restatable}{lem}{BasicFrameworkAssumptionSL}\label{lem:basic_framework_ptime_sl}
  $(\domain,\land,\leq)$ has a $(\rho_{\SL},\parone)$-UXP signature.
\end{restatable}
\begin{proof}
  First, note that in case one of the inputs given to $\land$ or $\leq$ is $\varnothing$, computing the output is trivial. Hence, we only consider the case of non-empty shifted lattices.
  Moreover, note that we can restrict ourselves to 
  shifted lattices having the same dimension. 
  (As in the previous section, this restriction means that $\land$ and $\leq$ 
  are equivalent to $\cap$ and $\subseteq$.)
  Indeed, consider elements $X = (n,\vec v_0,\dots,\vec v_k)$ and $Y = (m,\vec w_0,\dots, \vec w_j)$ from $\dom(\rho_{\SL})$, 
  with $n < m$. 
  Recall that, from the definition of FO theory given in~\Cref{sec:fo_structures}, 
  the shifted lattice $\rho_{\SL}(X)$ can be extended into a shifted lattice in $\ZZ^m$ as $\rho_{\SL}(X) \times \ZZ^{m-n}$.
  At the level of representation, this corresponds to adding $m-n$ vectors $\vec v_{k+1},\dots,\vec v_{k+m-n}$ such that, for $i \in [1,m-n]$, $\vec v_i$ has a $1$ in position $n+i$ and zeros in all other position (i.e., essentially adding an identity matrix for the last $m-n$ dimensions).
  One can then consider $X' = (m,\vec v_0,\dots, \vec v_k,\dots, \vec v_{k+m-n})$ instead of $X$ to compute $\land$ and $\leq$. Adding the new vectors can be done in polynomial time.
  
  Consider $X = (n,\vec v_0,\dots, \vec v_k)$ and $Y = (n, \vec w_0,\dots, \vec w_j)$ from $\dom(\rho_{\SL})$. We show that computing an element of $\dom(\rho_{\SL})$ representing $\rho_{\SL}(X) \cap \rho_{\SL}(Y)$ can be done in polynomial time.
  Let $A \coloneqq [ \vec v_1 \mid \dots \mid \vec v_k ]$ 
  and $B \coloneqq [ \vec w_1 \mid \dots \mid \vec w_j]$ be the matrices whose columns 
  correspond to the periods of the shifted lattices represented by $X$ and $Y$, respectively. Then,
  \begin{align*}
    \rho_{\SL}(X) \cap \rho_{\SL}(Y)
    &= 
    \left\{
    \begin{aligned}
    \vec v \in \ZZ^m : \ 
    &\vec v = \vec v_0 + A \cdot \vec y 
    = \vec w_0 + B \cdot \vec z\\
    &\text{for some }
    \vec y \in \ZZ^k \text{ and } 
    \vec z \in \ZZ^j 
    \end{aligned}
    \right\}\\ 
    &=
    \vec v_0
    + \left\{
    \begin{aligned}  
      A \cdot \vec y : \ 
      &
      \vec y \in \ZZ^k 
      \text{ and }
      A \vec y - B \vec z = \vec w_0  - \vec v_0\\
      &\text{for some } \vec z \in \ZZ^j
    \end{aligned}
    \right\}\\
    &= 
    \vec v_0 + A \cdot p(Z),
  \end{align*}
  where $p \colon \ZZ^{k} \times \ZZ^{j} \to \ZZ^k$ is the projection $p(\vec y,
  \vec z) = \vec y$ and $Z \coloneqq \{(\vec y, \vec z) \in \ZZ^{k} \times
  \ZZ^{j} : A \vec y - B \vec z = \vec w_0  - \vec v_0 \}$. In particular, $Z$
  is a shifted lattice given by the set of solutions of the weak PA formula $A
  \vec y - B \vec z = \vec w_0  - \vec v_0$. By appealing
  to~\Cref{lemma:change-of-representation}, we can compute a representation
  $(k+j,\vec z_0,\dots, \vec z_\ell)$ in $\dom(\rho_{\SL})$ of $Z$. The
  projection $p$ simply removes the last $j$ components of 
  $\vec z_0,\dots, \vec z_\ell$, resulting in vectors 
  $\vec z_0',\dots, \vec z_\ell' \in \ZZ^k$. 
  We have $\rho_{\SL}(X) \cap \rho_{\SL}(Y) = \vec v_0 + A \cdot \vec z_0' +
  {\Span_{\ZZ}\set{A \cdot \vec z_1',\ \dots,\ A \cdot \vec z_r'}}$. A
  representation for ${\rho_{\SL}(X) \cap \rho_{\SL}(Y)}$ is then given by
  $(n,\vec v_0 + A \cdot \vec z_0',\vec u_1, \dots, \vec u_r)$, where 
  $\vec u_1,\dots,\vec u_r$ form a basis for ${\Span_{\ZZ}\set{A \cdot \vec z_1',\
  \dots,\ A \cdot \vec z_r'}}$. To compute such a basis (in polynomial time) it
  suffices to put the matrix 
  $[A \cdot \vec z_1' \mid \dots \mid A \cdot \vec z_r']$ in Hermite normal 
  form and take as $\vec u_1,\dots,\vec u_r$ all its
  non-zero columns.

  Consider now the case of establishing whether $\rho_{\SL}(X) \subseteq \rho_{\SL}(Y)$. 
  It is easy to see that this inclusion holds if and only if $\vec v_0 \in \rho_{\SL}(Y)$ and for every $i \in [1,k]$, $\vec v_i \in \Span_{\ZZ}\set{\vec w_1,\dots,\vec w_j}$. 
  The right to left direction is trivial: 
  as $\Span_{\ZZ}\set{\vec w_1,\dots,\vec w_j}$ is closed under linear combinations, we conclude $\Span_{\ZZ}\set{\vec v_1,\dots,\vec v_k} \subseteq \Span_{\ZZ}\set{\vec w_1,\dots,\vec w_j}$ which, together with 
  $\vec v_0 \in \rho_{\SL}(Y)$, implies 
  $\rho_{\SL}(X) \subseteq \rho_{\SL}(Y)$.
  For the left to right direction, $\rho_{\SL}(X) \subseteq \rho_{\SL}(Y)$ directly implies 
  $\vec v_0 \in \rho_{\SL}(Y)$ 
  and 
  $\vec v_0 + \vec v_i \in \rho_{\SL}(Y)$, for every $i \in [1,k]$.
  Then, 
  \begin{align*}
  \vec v_i &= (\vec v_0 + \vec v_i) - \vec v_0\\ 
  &= \vec w_0 + B \cdot \vec z_1 - (\vec w_0 + B \cdot \vec z_2) 
  \text{ for some } \vec z_1,\vec z_2 \in \ZZ^j\\
  & = B \cdot (\vec z_1 - \vec z_2) \text{ for some } \vec z_1,\vec z_2 \in \ZZ^j\\ 
  & \in \Span_{\ZZ}\set{\vec w_1,\dots, \vec w_j}.
  \end{align*}
  The various membership queries $\vec v_0 \in \rho_{\SL}(Y)$ and, for all $i \in [1,k]$, $\vec v_i \in \Span_{\ZZ}\set{\vec w_1,\dots, \vec w_j}$
  ask for the feasibility of some systems of linear equations. This problem can be solved in polynomial time by Gaussian elimination (or by checking if $F$ in~\Cref{lemma:change-of-representation} returns $\varnothing$). 
\end{proof}

\subsection{Auxiliary lemmas on lattices}

Before moving to~\Cref{step:basic_framework_ptime:Item2}  
of~\Cref{step:basic_framework_ptime}, a few more notions and lemmas on lattices are required. 
Consider linearly independent vectors $\vec v_1,\dots,\vec v_k \in \ZZ^d$ and the lattice  $L \coloneqq \Span_{\ZZ}\set{\vec v_1,\dots, \vec v_k}$.
The determinant of $L$ is defined by $\det(L):=\sqrt{\det(B^\intercal B)}$ where $B \coloneqq [\vec v_1 \mid \dots \mid \vec v_k]$.
When $L$ is fully-dimensional (i.e., $k = d$), 
we have $\det(L)=|\det(B)|$.

\begin{propC}[{\cite[Lecture 1, Theorem~16]{LectureLattice}}]\label{lem:lattice_contained_detZ}
  Let $L \subseteq \ZZ^d$ be a fully-dimensional lattice.
  Then, $\det(L) \cdot \ZZ^d \subseteq L$.
\end{propC}

We let $L^\perp$ denote the \emph{orthogonal lattice} to $L$,
given by $L^\perp \coloneqq \set{\vec y\in\ZZ^d:{\forall \vec x\in L,} {\Angle{\vec y,\vec x}=0}}$, where $\Angle{\cdot,\cdot}$ stands for the dot product. See e.g.~\cite{10.1007/BFb0052236}.

\begin{lem}\label{lem:orto-lattice-properties}
  Let $L \subseteq \ZZ^d$ be a lattice of dimension $k$. Then, $L^\perp$ is a lattice of dimension $d-k$, and $L + L^\perp$ is a fully-dimensional lattice.
  Moreover, there is a polynomial time algorithm that on input $X = (d, \vec 0, \vec v_1,\dots, \vec v_k) \in \dom(\rho_{\SL})$ returns a basis $\vec v_{k+1},\dots,\vec v_d \in \ZZ^d$ of $L^{\perp}$, where $L \coloneqq \rho_{\SL}(X)$; and the vectors $\vec v_1,\dots, \vec v_d$ form a basis of $L + L^{\perp}$. 
\end{lem}

\begin{proof}
    For $k = 0$, i.e.,~when $L$ has dimension $0$, 
    $L^{\perp} = \ZZ^d$ and the lemma is trivial.
    Below, we show the result for $k=1$, and then generalise it to arbitrary~$k \geq 2$. 

    Let $L= \ZZ \cdot \vec v_1$ where $ \vec v_1 = (u_1,\dots,u_d) \in\ZZ^d \setminus \set{\vec 0}$.
    Without loss of generality, since $\vec v_1 \neq 0$, we can permute the coordinates so that $u_1\neq 0$.
    For $j\in [1,d]$, let
    $g_j \coloneqq \gcd(u_1,\ldots,u_j)$, which is non-zero as $u_1\neq0$.
    Note that $g_{j+1} = \gcd(g_j,u_{j+1})$ (for $j < d$). 
    We rely on the extended Euclidean algorithm for GCD to compute in polynomial time each $g_j$ and the Bézout's coefficients $a_1,\ldots,a_j\in\ZZ$
    such that $g_j=a_{j,1}u_1+\cdots+a_{j,j}u_j$.
    For $j \in [1,d-1]$, let
    \[
        \vec v_{j+1} \coloneqq (-\beta_j a_{j,1},\,\ldots,-\beta_j a_{j,j},\,\tfrac{g_j}{g_{j+1}},\,0,\,\ldots,\,0), \text{ where }
        \beta_j \coloneqq \tfrac{u_{j+1}}{g_{j+1}}.
    \]
    Note that both $\beta_j$ and $\tfrac{g_j}{g_{j+1}}$ belong to $\ZZ$, by definition of $g_j$ and $g_{j+1}$. 

    Observe that $\vec v_2,\dots,\vec v_d$ is a family of $d-1$ linearly independent vectors, as they form an echelon family of vectors.
    Let $L'\coloneqq \Span_\ZZ\set{\vec v_2,\ldots, \vec v_{d}}$, which is a lattice of dimension $d-1$. We claim that $L'=L^\perp$. 

    First, we show that $L'\subseteq L^\perp$. Let $j\in\set{1,\ldots,d-1}$, we need to show that $\Angle{\vec v_1, \vec v_{j+1}}=0$. This is simple to check:
    \begin{align*}
        \Angle{\vec v_1,\vec v_{j+1}}
            &=u_1\cdot(-\beta_ja_{j,1})+\cdots+u_j\cdot(-\beta_ja_{j,j})+u_{j+1}\cdot\tfrac{g_j}{g_{j+1}}\\
            &=\beta_j\left(g_j-u_1a_{1,j}-\cdots-u_ja_{j,j}\right)
             && \text{since }u_{j+1}\cdot\tfrac{g_j}{g_{j+1}}=\beta_jg_j\\
            &=0 && \hspace*{52pt}\text{by def.~of }g_j.
    \end{align*}

    Let us now show that~$L^\perp\subseteq L'$.
    For any $\vec x = (x_1,\dots,x_d)\in\ZZ^d$, define $k(\vec x) \coloneqq \max(\set{0}\cup\set{i:x_i \neq 0})$. We show that if $\vec x\in L^\perp$ then $\vec x\in L'$, for every $\vec x \in \ZZ^d$,
    by induction on $k(\vec x)$. If $k(\vec x)=0$ then $\vec x=0$ so the result is trivial. Otherwise, let $\ell:=k(\vec x)$
    and observe that
    $\Angle{\vec v_1,\vec x}=0$ since $\vec x\in L^\perp$, that is
    \[
        u_1x_1+\cdots+u_{\ell-1}x_{\ell-1}=-u_\ell x_\ell.
    \]
    By Bézout's identity, $\gcd(u_1,\ldots,u_{\ell-1})=g_{\ell-1}$ divides $u_\ell x_\ell$. Furthermore,
    note that $g_\ell=\gcd(g_{\ell-1},u_\ell)$ by definition. 
    Note that, for all integers $a,b,c$, if $a$ divides $b \cdot c$ then $a$ divides $\gcd(a,b) \cdot c$.
    Hence,
    $g_{\ell-1}$ divides $g_\ell x_\ell$, and $N \coloneqq \tfrac{g_\ell x_\ell}{g_{\ell-1}}$ is an integer.    
    Now consider the vector $\vec x' = (x_1',\dots,x_k')=\vec x-N \vec v_{\ell}$.
    Recall that $\vec v_{\ell}\in L'\subseteq L^\perp$ by the previous inclusion, so $\vec x'\in L^\perp$
    since $L^\perp$ is a lattice. Furthermore, we claim that $k(\vec x')<k( \vec x)$. Indeed,
    the coordinates of $\vec x$ above $\ell=k(\vec x)$ are all zero, and the same is true for $\vec v_{\ell}$ by definition.
    The $\ell$-th coordinate is
    \[
        x'_{\ell}
            \ = \ x_\ell-N \cdot \tfrac{g_{\ell-1}}{g_\ell}
            \ = \ x_\ell-\tfrac{g_\ell x_\ell}{g_{\ell-1}}\cdot\tfrac{g_{\ell-1}}{g_\ell}
            \ = \ 0.
    \]
    This shows that $k(\vec x')<\ell=k(\vec x)$ so by induction $\vec x'\in L'$. 
    Since $L'$ is a lattice and $\vec v^{\ell-1}\in L'$, we conclude that~$\vec x= \vec x'+N \vec v_{\ell}\in L'$.

    In summary, we have shown that $L^\perp=L'$ when $L=\ZZ \vec v_1$, and that $\dim(L^\perp)=d-1$. It follows
    immediately by the orthogonality of the vectors that $\dim(L+L^\perp)=\dim(L)+\dim(L^\perp)=d$.
    Finally, computing $\vec v_2,\dots,\vec v_k$ can be done in polynomial time. 
    Hence, the lemma is proven when $k = 1$.

    For the case $k \geq 2$, observe that 
    given $L=\Span_{\ZZ}\set{\vec v_1,\ldots,\vec v_k}$, we have~$L^\perp=\bigcap_{i=1}^k(\ZZ \vec v_i)^\perp$.
    It is clear then that $\dim(L^\perp)\geqslant d-k$ since every $(\ZZ \vec v_i)^\perp$ has dimension $d-1$.
    On the other hand, it is easy to see that the orthogonality yields 
    $\dim(L+L^\perp)=\dim(L)+\dim(L^\perp)$ and therefore
    $\dim(L)+\dim(L^\perp)\leqslant d$, so we conclude that $\dim(L^\perp)=d-k$ and $\dim(L+L^\perp)=\dim(L)+\dim(L^\perp)=d$.
    To compute a basis for $L^\perp$, we first 
    compute a basis for each $(\ZZ \vec v_i)^\perp$ in polynomial time by relying on the argument used above for the case $k = 1$. 
    Afterwards, the intersection of the $d$ resulting lattices can be computed in polynomial time by slightly extending the arguments used in~\Cref{lem:basic_framework_ptime_sl} to compute $\rho_{\SL}(X) \land \rho_{\SL}(Y)$. 
    Given, for every $i \in [1,k]$, the matrix $A_i \in \ZZ^{d \times (d-1)}$ whose columns form a basis of $(\ZZ \vec v_i)^\perp$, we have 
    \begin{align*}
      L^{\perp} & = \left\{  
      \begin{aligned}
        A_1 \cdot \vec y_1 : {}&
        \vec y_1 \in \ZZ^{d-1} \text{ and } A_1 \cdot \vec y_1 = A_2 \cdot \vec y_2 = \dots = A_d \cdot \vec y_d\\ 
        &\text{for some } \vec y_2,\dots,\vec y_d \in \ZZ^{d-1}
      \end{aligned}
      \right\}\\
      & = A \cdot \pi(Z),
    \end{align*} 
    where $\pi \colon (\ZZ^{(d-1)})^d \to \ZZ^{d-1}$ 
    is here the projection into the first vector of dimension $d-1$,
    and $Z \coloneqq \{(\vec y_1, \dots, y_d) \in (\ZZ^{(d-1)})^d : A_1 \cdot \vec y_1 = A_2 \cdot \vec y_2 = \dots = A_d \cdot \vec y_d \}$.
    Then, the computation of a basis for $L^\perp$ proceeds as in~\Cref{lem:basic_framework_ptime_sl}, by appealing to~\Cref{prop:hnf_snf}. 
\end{proof}

Below, given an assertion $\Phi$,
we write $\indicator{}{\Phi}$ for
the indicator function defined as $\indicator{}{\Phi} = 1$ if~$\Phi$ is true and $\indicator{}{\Phi} = 0$ otherwise.

\begin{lem}\label{lemma:lattice-volume}
  Let $L = \Span_{\ZZ}\set{\vec v_1,\dots, \vec v_k} \subseteq \ZZ^d$ and $s \geq 1$ such that $s \cdot \ZZ^d \subseteq L$. 
  For every $\vec v_0 \in \ZZ^d$, $|(\vec v_0 + L) \cap [0,s)^d | = \frac{s^d}{\det(L)}$.
\end{lem}

\begin{proof}
  Note that the hypotheses of the lemma imply $s \in \NN$ and that $L$ is fully-dimensional. Considering $\vec v_0 = \vec 0$ suffices to show the lemma. Indeed, given $\vec w \in \ZZ^d$, we have 
  \begin{align*}
    |(\vec w + L) \cap [0,s)^d| & = 
    \sum_{\vec y \in [0,s)^d} \indicator{}{\vec y \in (\vec w + L)}\\
    & = \sum_{\vec y \in [0,s)^d} \indicator{}{\vec y - \vec w \in L}\\
    & = \sum_{\vec y \in [0,s)^d} \indicator{}{(\vec y - \vec w \bmod s) \in L} 
    &\text{since } s \cdot \ZZ^d \subseteq L\\ 
    & = \sum_{\vec y \in [0,s)^d} \indicator{}{\vec y \in L}
    &
    \hspace{-30pt}
    \begin{aligned}
      \text{since } f \colon [0,s)^d \to [0,s)^d\\ 
      \text{defined as }f(\vec y) \coloneqq \vec y - \vec x \bmod s\\
      \text{is a bijection}
    \end{aligned}\\ 
    & = |L \cap [0,s)^d|.
  \end{align*}
  Let us show that $|L \cap [0,s)^d| = \frac{s^d}{\det(L)}$. 
  Let $H$ be the Hermite normal form of the matrix $A \coloneqq [\vec v_1 \mid \dots \mid \vec v_k]$, and $U$ be the unimodular matrix 
  such that 
  $H = A \cdot U$. Since $L$ is fully-dimensional, $H$ is invertible and (from the fact that $H$ is triangular and has positive pivots), we conclude that 
  \[ 
    H=\begin{bmatrix}
      p_1 & 0 & 0 & \cdots & 0 \\
      a_{2,1}   & p_2     & 0 &  \cdots & 0 \\
      a_{3,1}   & a_{3,2}       & p_3 &  \cdots & 0 \\
      \vdots   & \vdots  & \vdots  & \ddots & \vdots \\
      a_{d,1}   & a_{d,2} & a_{d,3}   & \cdots  & p_d
     \end{bmatrix}.
  \]
  Here, $p_i$ is strictly positive, for every $i \in [1,d]$. 
  We claim that $p_i$ divides~$s$.
  Indeed, since $s \cdot \ZZ^d \subseteq L$, we have $s \cdot \vec e_i \in L$, where $\vec e_i$ is the $i$th unit vector of the canonical basis of $\ZZ^d$.
  Let $\vec x = (x_1,\dots,x_d) \in \ZZ^d$ such that $H \cdot \vec x = s \cdot e_i$. Because of the shape of $H$, it must be the case that $x_j = 0$ for every $j \in [1,i-1]$. This implies $p_i \cdot x_i = s$, i.e., $p_i$ divides $s$.
  Since $U$ is unimodular, 
  the function $f(\vec u) \coloneqq U \cdot \vec u$ is a bijection on $\ZZ^d$, 
  that is, $U \cdot \ZZ^d = \ZZ^d$. 
  We get:
  \begin{align*}
    |L \cap [0,s)^d| 
      &= | \{ \vec y \in [0,s)^d : \exists \vec x \in \ZZ^d, \vec y = A \cdot \vec x \}|\\
      &= |\{ \vec y \in [0,s)^d : \exists \vec x \in \ZZ^d, \vec y = H \cdot \vec x \}|.
  \end{align*}
  Let us denote with $(H \cdot \vec x)_i$ the $i$th entry of $H \cdot \vec x$. We have,
  \begin{align*}
    0 \leq (H \cdot \vec x)_i < s
    \ \Leftrightarrow\ 
    0 \leq {\underbrace{\sum_{j=1}^{i-1} a_{i,j} \cdot x_j}_{\alpha_i}} + p_i \cdot x_i < s 
    \ \Leftrightarrow\ 
    - \frac{\alpha_i}{p_i} \leq x_i < \frac{s}{p_i} - \frac{\alpha_i}{p_i}.
  \end{align*}
  Note that $\frac{s}{p_i}$ is a positive integer, since $p_i,s \geq 1$ and $p_i$ divides $s$. Therefore, the rightmost inequality above is of the form $\beta \leq x_i < N + \beta$ for some positive integer $N$ and rational $\beta$. This system of inequalities has always $N$ integer solutions for $x_i$, independently of the value of $\beta$. 
  It therefore follows that the cardinality of the set $\{ \vec y \in [0,s)^d : \exists \vec x \in \ZZ^d, \vec y = H \cdot \vec x \}$ 
  is 
  \[ 
    \frac{s^d}{p_1 \cdot \ldots \cdot p_k} 
    = 
    \frac{s^d}{|\det(H)|}
    = 
    \frac{s^d}{\det(L)}.
    \qedhere
  \]
\end{proof}

The following two lemmas are related to the problem of testing inclusion of union of lattices, but will also play a role when considering the universal projection in Requirement~\ref{step:beforeReducPiUnivToSimpleCasesUxp}.

\begin{lem}\label{lem:lattice_equal_period}
  Consider $\ell$ lattices $L_0,\ldots,L_\ell\subseteq\ZZ^d$ and 
  vectors $\vec v_0,\ldots, \vec v_\ell\in\ZZ^d$. 
  Suppose $s \cdot \ZZ^d\subseteq L_0,\ldots,L_\ell$ for some $s \geq 1$.
  Then,
  $\vec v_0+L_0 = \bigcup_{i=1}^\ell(\vec v_i+L_i)$
  if and only if $(\vec v_0+L_0)\cap[0,s)^d = \bigcup_{i=1}^\ell(\vec v_i+L_i)\cap[0,s)^d$.
\end{lem}
\begin{proof} 
  The left to right direction is immediate. 
  For the other direction, suppose
  $(\vec v_0+L_0)\cap[0,s)^d = \bigcup_{i=1}^\ell( \vec v_i+L_i)\cap[0,s)^d$. We prove the two inclusions. 

  ($\subseteq$): We consider $\vec v \in L_0$ and show $\vec v_0 + \vec v \in \bigcup_{i=1}^\ell(\vec v_i+L_i)$.
  Since $s \cdot \ZZ^d \subseteq L_0$, there is $\vec u \in \ZZ^d$ such that $\vec v_0 + \vec v + s \cdot \vec u \in (\vec v_0 + L_0) \cap [0,s)^d$. Then, by hypothesis, $\vec v_0 + \vec v + s \cdot \vec u \in \bigcup_{i=1}^\ell(\vec v_i+L_i)\cap[0,s)^d$.
  Since $s \cdot \ZZ^d \subseteq L_i$ for every $i \in [1,\ell]$, 
  we conclude that $\vec v_0 + \vec v \in \bigcup_{i=1}^\ell(\vec v_i+L_i)\cap[0,s)^d$.

  ($\supseteq$): This inclusion is similar to the previous one. 
  Consider $i \in [1,\ell]$ and $\vec v \in L_i$. 
  We show that $\vec v_i + \vec v \in \vec v_0 + L_0$.
  Since $s \cdot \ZZ^d \subseteq L_i$, there is $\vec u \in \ZZ^d$ 
  such that $\vec v_i + \vec v + s \cdot \vec u \in \vec (v_i + L_i) \cap [0,s)^d$. Then, $\vec v_i + \vec v + s \cdot \vec u \in \vec (v_0 + L_0) \cap [0,s)^d$, 
  and since $s \cdot \ZZ^d \subseteq L_0$ 
  we conclude that $\vec v_i + \vec v \in \vec v_0 + L_0$.
\end{proof}

\begin{lem}\label{lem:inclusion_uofls_fulldim}
  Consider $\ell$ lattices $L_0,\ldots,L_\ell\subseteq\ZZ^d$ and 
  vectors $\vec v_0,\ldots, \vec v_\ell\in\ZZ^d$,
  with $L_0$ fully-dimensional. Then, 
  \[
      \vec v_0 + L_0 \subseteq \bigcup\nolimits_{i = 1}^\ell \vec v_i + L_i
      \ \Leftrightarrow\
      \vec v_0 + L_0 \subseteq \bigcup\nolimits_{i \in I} \vec v_i + L_i,
  \]
  where $I \coloneqq \{ i \in [1,\ell] : L_i \text{ is fully-dimensional}\}$.
\end{lem}
\begin{proof}
  The lemma is trivial for $d = 0$, hence assume $d \geq 1$.
  The right to left direction is straightforward.
  For the left to right direction, \emph{ad absurdum}
  suppose $\vec v_0 + L_0 \subseteq \bigcup\nolimits_{i = 1}^\ell \vec v_i + L_i$ 
  but $\vec v_0 + L_0 \not\subseteq \bigcup\nolimits_{i \in I} \vec v_i + L_i$.
  Since, for every $i \in I$, $L_i$ is fully-dimensional, 
  by~\Cref{lem:lattice_contained_detZ} 
  we can find $s_i \in \NN$ such that $s_i \cdot \ZZ^d \subseteq L_i$.
  Similarly, we can find $s_0 \in \NN$ such that $s_0 \cdot \ZZ^d \subseteq L_0$. 
  Let $s \coloneqq \lcm\{ s_i : i \in {I \cup \{0\}} \}$, 
  so that $s \cdot \ZZ^d \subseteq L_i$ for every $i \in {I \cup \{0\}}$.
  By properties of lattices, this means that for every $i \in {I \cup \{0\}}$ there is a (finite) set $U_i \subseteq [0,s)^d$ such that 
  $\vec v_i + L_i = U_i + s \cdot \ZZ^d$.
  Now, on the one hand,
  \[
    \vec v_0 + L_0 \subseteq \bigcup\nolimits_{i=1}^\ell \vec v_i + L_i 
    = (U + s \cdot \ZZ^d) \cup \bigcup\nolimits_{j \in [1,\ell] \setminus I} 
    (\vec v_j + L_j),
  \]
  where $U \coloneqq \bigcup_{i \in I} U_i$. 
  On the other hand, we have 
  \[ 
    \vec v_0 + L_0 \not\subseteq \bigcup\nolimits_{i \in I} \vec v_i + L_i  = U + s \cdot \ZZ^d.
  \]
  Recall that $\vec v_0 + L_0 = U_0 + s \cdot \ZZ^d$, so we must have $U_0 \not\subseteq U$. Since $L_0$ is fully-dimensional, $U_0 \neq \emptyset$. 
  Then, take $\vec z \in U_0 \setminus U$. 
  Note $U_0,U \subseteq [0,s)^d$ and $\vec z \in [0,s)^d$,
  and therefore
  $\vec z + s \cdot \ZZ^d \subseteq (U_0 + s \cdot \ZZ^d) \setminus (U + s \cdot \ZZ^d)$.
  This means $\vec z + s \cdot \ZZ^d \subseteq \bigcup\nolimits_{j \in [1,\ell] \setminus I} 
  (\vec v_j + L_j)$. 
  We claim that this is not possible by a dimension argument. 
  To see that, fix an integer $M$.
  The last inclusion implies that
  \begin{equation*}
      |( \vec z+ s \cdot \ZZ^d)\cap[0,M \cdot s)^d|\leqslant \sum\nolimits_{j \in [1,\ell] \setminus I}|(\vec v_j + L_j)\cap[0,M \cdot s)^d|.
  \end{equation*}
  However, observe that
  $|(\vec v+ s \cdot \ZZ^d)\cap[0,M \cdot s)^d| = (M-1)^d \in \Omega(M^d)$
  whereas, since $\dim(L_i)\leqslant d-1$ for every $i\in [1,\ell] \setminus I$,
  $|(\vec v_j + L_j)\cap[0,M \cdot s)^d|=O_{M\to\infty}(M^{k-1})$.
  As $M\to\infty$, we can see that the left hand-side of the above equation
  grows much faster than the right-hand side, and so we have reached a contradiction.
\end{proof}

\subsection{
Requirement~\ref{step:basic_framework_ptime}, Item~\ref{step:basic_framework_ptime:Item2}: \texorpdfstring{$(\un(\domain),\leq)$}{(un(D),leq)} has a 
\texorpdfstring{$(\un(\rho_{\SL}),\len(\parone))$}{(un(rhoSL),len(1))}-UXP signature
}

We are ready to present an algorithm
to solve inclusion between union of shifted
lattices that runs in polynomial time when the length of the union is considered fixed (as it is the case when taking into account the parameter $\len(\parone)$).

\begin{lem}\label{lemma:wPA-un-UXP}
  The structure~$(\un(\domain),\leq)$ has a
  $(\un(\rho_{\SL}),\len(\parone))$-UXP signature.
\end{lem}

\begin{proof}
  As in~\Cref{lem:basic_framework_ptime_sl}, without loss of generality we assume all shifted lattices to be non-empty and have the same dimension. 
  We describe an algorithm that given
  $X \coloneqq (X_1,\dots,X_n)$ and $Y = (Y_1,\dots,Y_m)$ in $\dom(\un(\rho_{\SL}))$,
  with $X_i = (d,\vec v_{i,0},\dots,\vec v_{i,k_i})$ 
  and $Y_j = (d, \vec w_{j,0},\dots,\vec w_{j,\ell_j})$, 
  checks if $\un(\rho)(X) \leq \un(\rho)(Y)$.
  The algorithm runs in time $2^{m} \cdot \poly(|X|,|Y|)$.
  The inclusion is checked by verifying that, 
  for every $i \in [1,n]$, $\rho_{\SL}(X_i) \leq \un(\rho)(Y)$.

  For every $i \in [1,n]$, the algorithm first computes the following objects:
  \begin{enumerate}
    \item Using~\Cref{lem:basic_framework_ptime_sl}, for all $j \in [1,m]$, 
    we compute $\vec v_j' \in \ZZ^d$ and a basis for the lattice $L_j' \subseteq \ZZ^d$ such that $\rho_{\SL}(Y_j \land X_i) = \vec v_j' + L_j'$.
    This step only requires polynomial time.
    \item A basis for 
    the lattice $L^{\perp}$ orthogonal to $L \coloneqq \Span_{\ZZ}(\vec v_{i,1},\dots,\vec v_{i,k_i})$, using~\Cref{lem:orto-lattice-properties}. 
    This step only requires polynomial time.
    \item The set $I \coloneqq \{ j \in [1,m] : L^{\perp} + \rho_{\SL}(Y_j \land X_i) \text{ is fully-dimensional} \}$. 
    To check whether ${L^{\perp} + \rho_{\SL}(Y_j \land X_i)}$ is fully-dimensional, it suffices to check that $Y_j \land X_i$ and $X_i$ have the same dimension, i.e., the same number of periods. This can be done in polynomial time.
    \item The positive integer $s \coloneqq \det(L + L^\perp) \cdot \prod_{j \in I} \det(L_j' + L^\perp)$.
    Since we have bases for all these lattices,
    computing the determinant only requires polynomial time (e.g., bring the matrix in Hermite normal form and multiply all the elements in the diagonal).
  \end{enumerate}
  By~\Cref{lem:lattice_contained_detZ}, we have $s \cdot \ZZ^{d} \subseteq L + L^\perp$ and $s \cdot \ZZ^d \subseteq L_j' + L^\perp$ for every $j \in I$.
  At this stage,
  observe the following equivalences that stem from the previous lemmas on lattices:
  {\allowdisplaybreaks\begin{align*}
    & \rho_{\SL}(X_i) \leq \un(\rho)(Y)\\
    \Leftrightarrow{} & \rho_{\SL}(X_i) = \bigcup_{j=1}^m (\rho_{\SL}(Y_j \land X_j))\\
    \Leftrightarrow{} & Z = \bigcup_{j=1}^m W_j
    &
    \begin{aligned}
    Z &\coloneqq \vec v_{i,0} + L + L^\perp
    \text{ and }\\
    W_j &\coloneqq \vec v_j' + L_j' + L^\perp
    \end{aligned}\\
    \Leftrightarrow{} & Z = \bigcup_{j \in I} 
    W_j
    &
    \begin{aligned}
    \text{\Cref{lem:inclusion_uofls_fulldim} and the fact}\\
    \text{that $W_j \subseteq Z$ for all $j \in [1,m]$}
    \end{aligned}\\
    \Leftrightarrow{} & Z \cap [0,s)^d = \bigcup_{j \in I} W_j \cap [0,s)^d
    &\text{\Cref{lem:lattice_equal_period}}\\
    \Leftrightarrow{} & |Z \cap [0,s)^d| = |\bigcup_{j \in I} W_j \cap [0,s)^d|
    &
    \begin{aligned}
    \text{$\Leftarrow$ follows from the fact}\\ 
    \text{that $W_j \subseteq Z$ for all $j \in I$}
    \end{aligned}\\
    \Leftrightarrow{} &
    |Z \cap [0,s)^d| = \sum_{J \subseteq I} 
    (-1)^{|J|+1}|\bigcap_{j \in J} W_j \cap [0,s)^d| 
    &\text{inclusion-exclusion.}
  \end{align*}}
  The algorithm then computes, for every $J \subseteq I$ 
  a shifted lattice $V_J$ from $\dom(\rho_{\SL})$ representing $\bigcap_{j \in J} W_j$.
  If $V_J \neq \varnothing$, below let $V_J \coloneqq (d,\vec u_{J,0},\dots,\vec u_{J,r_J})$.
  As described in the last part of the proof of~\Cref{lem:orto-lattice-properties},
  computing such a representation can be done in polynomial time. 
  Since, by~\Cref{lem:lattice_contained_detZ}, $s \cdot \ZZ^d \subseteq L_j' + L^\perp$ for every $j \in I$, 
  we have $s \cdot \ZZ^d \subseteq \Span_{\ZZ}\{\vec u_{J,1},\dots,\vec u_{J,r_J}\}$ for every $J \subseteq I$ with $V_J \neq \varnothing$.
  Note now the following:
  \begin{align*}
    &
    |Z \cap [0,s)^d| = \sum_{J \subseteq I} 
    (-1)^{|J|+1}|\bigcap_{j \in J} W_j \cap [0,s)^d|\\
    \Leftrightarrow{} &
    \frac{s^d}{\det(L+L^{\perp})} 
    = \sum_{J \subseteq I} 
    (-1)^{|J|+1} \frac{s^d \cdot \indicator{}{V_J \neq \varnothing}}{\det(\Span_{\ZZ}\{\vec u_{J,1},\dots,\vec u_{J,r_J}\})}
    &\text{\Cref{lemma:lattice-volume}}\\
    \Leftrightarrow{} &
    1
    = \sum_{J \subseteq I} 
    (-1)^{|J|+1} \frac{\det(L+L^{\perp}) \cdot \indicator{}{V_J \neq \varnothing} }{|\det([\vec u_{J,1} \mid \dots \mid \vec u_{J,r_J}])|}.
  \end{align*}
  Hence, the algorithm computes $\sum_{J \subseteq I} (-1)^{|J|+1} \frac{\det(L+L^{\perp}) \cdot \indicator{}{V_J \neq \varnothing}}{|\det([\vec u_{J,0} \mid \dots \mid \vec u_{J,r_J}])|}$ and 
  returns true if and only if it equals $1$.
  Each determinant computation requires polynomial time, but $O(2^{|I|})$ such computations are required. 
  Overall, we observe that the exponential blow-up in the algorithm is limited to the iterations of all subsets $J$ of $I$. All other operations are in polynomial time, 
  resulting in a $2^{m} \cdot \poly(|X|,|Y|)$ running time.
\end{proof}

\subsection{Requirement~\ref{step:beforeReducPiUnivToSimpleCasesUxp}: 
  both projections are in \texorpdfstring{$(\dnf(\rho),\depth(\theta))$}{(dnf(rhoSL),dep(theta))}-UXP}
Establishing the Items~\ref{step:beforeReducPiUnivToSimpleCasesUxp:Item1}
and~\ref{step:beforeReducPiUnivToSimpleCasesUxp:Item3}
of \Cref{step:beforeReducPiUnivToSimpleCasesUxp} 
is trivial: thanks to our choice of representation
based on~$\rho_{\SL}$, 
given $X \in \dom(\rho_{\SL})$ and $\vec i \in \vec I$, $\dotproj(\vec i, X)$
can be computed by simply
crossing out the entries of all vectors of $X$ corresponding
to the indices in $\vec i$, bringing the resulting matrix of periods in Hermite normal form and removing all its zero columns (to force the periods to be linearly independent).
The following result is thus immediate. 

\begin{restatable}{lem}{ProjSL}\label{lem:proj_SL_high}
  The structure~$(\domain,(\dotproj,\vec I))$ has a $(\rho_{\SL},\parone)$-UXP signature.
\end{restatable}

On the contrary, computing the universal projections~$\dotunprojrel{Z}(\vec i, X)$, 
as required by the Items~\ref{step:beforeReducPiUnivToSimpleCasesUxp:Item2}
and~\ref{step:beforeReducPiUnivToSimpleCasesUxp:Item4}, is computationally
expensive.
For simplicity, below we index entries in vectors starting from one, and instead of considering projections over arbitrary vectors of
indices $\vec i$, we assume $\vec i = [1,k]$ for some $k \in \NN$ so
that~$\unprojrel{Z}(\vec i, X)$ projects over the first $k$ dimensions. This is
w.l.o.g.,~ as we can reorder components appropriately. So, let~$\unprojrel{Z}(k.X) \coloneqq \unprojrel{Z}([1,k], X)$, 
and $\proj(k,X) \coloneqq \proj([1,k],X)$.  
For a set $S\subseteq\ZZ^d$ and $\vec x \in \RR^k$, with $k \leq d$ we define
the \emph{slice} of $S$ at $\vec x$, denoted by $\slice{S}{\vec x}$ as the set 
\[
  \slice{S}{\vec x}:=\set{\vec t\in\ZZ^{d-k}:(\vec x, \vec t)\in S}.
\] 
Before giving the algorithm for universal projection, we need the following
result.

\begin{lem}\label{lem:slice_of_LS}
  Let $L \subseteq \ZZ^d$ be a lattice and $\vec v_0 \in \ZZ^d$. There is a lattice $L' \subseteq \ZZ^d$ such that for all $\vec x \in \pi(k,\vec v_0 + L)$ there is $\vec t_{\vec x} \in \ZZ^{k}$ such that $\slice{(\vec v_0 + L)}{\vec x} = {\vec t_{\vec x} + L'}$.
  Moreover, there is an algorithm that given ${X = (d, \vec v_0, \vec v_1,\dots, \vec v_n) \in \dom(\rho_{\SL})}$ and $k \in \NN$ written in unary, returns a basis of $L'$, with respect to~$L \coloneqq \rho_{\SL}(X)$. 
  The algorithm runs in polynomial time.
\end{lem}
\begin{proof}
  Let $L \coloneqq \Span_{\ZZ}(\vec v_1,\dots, \vec v_n)$, 
  with $\vec v_1,\dots, \vec v_n$ being linearly independent.
  Note that $\pi(k,\vec v_0 + L)$ cannot be empty. 
  Let $\vec x_0$ be the vector in $\ZZ^{d-k}$ that is obtained from $\vec v_0$ by 
  removing the first $k$ components, so that $\vec x_0 \in \pi(k,\vec v_0 + L)$.
  The non-empty set $\slice{(\vec v_0 + L)}{\vec x_0}$
  corresponds to the set of solutions~$\vec x \in \ZZ^{d-k}$ to the following weak PA formula~$\Phi(\vec x)$
  \[ 
    \exists y_1,\dots,y_n : \ 
    \begin{bmatrix}
      \vec x_0\\ 
      \vec x
    \end{bmatrix}
    = \vec v_0 + \vec v_1 \cdot y_1 + \dots \vec v_n \cdot y_n.
  \]
  By~\Cref{lemma:change-of-representation} and~\Cref{prop:hnf_snf}, we can
  compute in polynomial time with respect to $\vec v_0,\dots, \vec v_k$ a family
  of vectors $\vec w_0,\dots,\vec w_j$ such that $\vec w_1,\dots, \vec w_j$ are
  linearly independent and 
  $\slice{(\vec v_0 + L)}{\vec x_0} = \sem{\Phi(\vec x)}_{\mathcal{Z}} = \vec w_0 + \Span_{\ZZ}(\vec w_1,\dots, \vec w_j)$. 
  Then, $L'$ in the statement of the lemma is given by $\Span_{\ZZ}(\vec w_1,\dots, \vec w_j)$. 

  To conclude the proof, we must check that for every $\vec x \in \pi(k,\vec v_0 + L)$ there is $\vec t_{\vec x} \in \ZZ^{k}$ such that $\slice{(\vec v_0 + L)}{\vec x} = \vec t_{\vec x} + L'$.
  To this end, consider $\vec x \in \pi(k,\vec v_0 + L)$ 
  and pick as $\vec t_{\vec x} \in \ZZ^{k}$ the only vector in $\slice{\{\vec v_0\}}{\vec x}$.
  Given $\vec s \in \ZZ^k$, we have
  {\allowdisplaybreaks
  \begin{align*}
      & \vec s \in \slice{(\vec v_0 + L)}{\vec x}\\
          \Leftrightarrow{}& (\vec x,\vec s) \in \vec v_0 + L\\
          \Leftrightarrow{}& (\vec x, \vec s)-(\vec x, \vec t_{\vec x})
          \in L 
          &\text{since $(\vec x, \vec t_{\vec x}) = \vec v_0$ and $L$ is a lattice}\\
          \Leftrightarrow{}& (\vec x_0, \vec w_0)+(\vec 0,\vec t_{\vec x}- \vec s)\in \vec v_0 + L
          &\text{as $(\vec x_0, \vec w_0)\in \vec v_0 + L$ and $L$ is a lattice}\\
          \Leftrightarrow{}& (\vec x_0, \vec t_{\vec x}- \vec s + \vec w_0)\in \vec v_0 + L\\
          \Leftrightarrow{}& \vec t_{\vec x}-\vec s+ \vec w_0 \in \slice{(\vec v_0 + L)}{\vec x_0}\\
          \Leftrightarrow{}& \vec t_{\vec x}- \vec s \in L'
          &\text{by def.~of $\vec w_0$}\\
          \Leftrightarrow{}& \vec s \in \vec t_{\vec x} + L'
          &\text{as $L'$ is a lattice.}
          &\quad\qedhere
  \end{align*}}
\end{proof}

Back to the problem of performing universal projection, 
intuitively, we need to count points in
unions of shifted lattices (similarly to inclusion testing) but in a parametric way. This means that given a union of shifted lattices 
$X= \bigcup_{i=1}^m (\vec v_i + L_i)$ every
intersection $\bigcap_{j\in J}(\vec v_i + L_i)$ with $J\subseteq L$ in the
inclusion-exclusion formula may or may not need to be
accounted for, depending on the value of a parameter $f \colon 2^{[1,m]} \to \{0,1\}$ belonging
of a certain set of parameters $\mathcal{F}$ (see the lemma below, 
the exact definition of $\mathcal{F}$ is technical and only given in the proof; $2^{[1,m]}$ stands for the powerset of $[1,m]$). 
The
algorithm therefore considers all possible ways in which
intersections may or may not be taken, which is roughly
$2^{2^{m}}$; i.e., the number of functions in $[2^{[1,m]} \to \{0,1\}]$.
Our algorithm allows us to conclude a rather
surprising fact: the relative universal projection can be
expressed as a complex combination of unions, intersections,
projections and the relative complementations that are
exclusively applied to the initial~sets in input. The number
of these operations only depends on $m$, resulting in an algorithm that runs in polynomial time when $m$ is fixed.

\begin{restatable}{lem}{UnivProjUofSLRel}\label{lem:univ_proj_UofSL_relative_high}
    Let $Z,X_1,\dots,X_m \subseteq \ZZ^d$ be shifted lattices, 
    $X = \bigcup_{i=1}^m X_i$ and $k \in \NN$.
    There are $I \subseteq [1,m]$ and $\mathcal{F}\subseteq [2^{I}\to\set{0,1}]$~such that
    \[
        \unprojrel{Z}(k,X)=\bigcup_{f\in\mathcal{F}}
            \Big(\Big(\bigcap_{J :f(J)=1}\bigcap_{j\in J}\proj(k,X_j\cap Z)\Big)
            - \Big(\bigcup_{J :f(J)=0}\bigcap_{j\in J}\proj(k,X_j\cap Z)\Big)\Big).
    \]
    Fix $m \in \NN$. There is an algorithm that given in input $
    X = (X_1,\dots,X_m)\in \dom(\un(\rho_{\SL}))$, $Z \in \dom(\rho_{\SL})$,
    and $k \in \NN$ in unary, 
    returns $Y \in \dom(\dnf(\rho_{\SL}))$ such that $\dnf(\rho_{\SL})(Y) = \unprojrel{\rho_{\SL}(Z)}(k,\un(\rho_{\SL})(X))$. 
    The algorithm runs in polynomial~time.
\end{restatable}

\begin{proof}
  Let us focus on the first part of the lemma. 
  We first show the result under the following additional hypothesis:
  \begin{equation}
    \tag{$\dagger$}\label{eq:univ-proj-additional-hyp}
    \text{$X \subseteq Z$ and for every $\vec x \in \pi(k,Z)$, $\slice{Z}{\vec x} \subseteq \ZZ^{k}$ is fully-dimensional.}
  \end{equation}
  Note that, in the statement of the lemma, the intersections $X_j \cap Z$ can then be replaced by $X_j$.
  We later show how to discharge this additional hypothesis.

  Starting from representations of $X$ and $Z$ in $\dom(\un(\rho_{\SL}))$ and $\dom(\rho_{\SL})$ respectively,
  we apply~\Cref{lem:slice_of_LS} to compute bases for the lattices $L_i' \subseteq \ZZ^k$ ($i \in [0,m]$) such that 
  for every $i \in [1,m]$ and $\vec x \in \pi(k,X_i)$ there is $\vec t_{\vec x} \in \ZZ^{k}$ such that $\slice{X_i}{\vec x} = \vec t_{\vec x} + L_i'$, 
  and for every $\vec y \in \pi(k,Z)$ there is $\vec t_{\vec y} \in \ZZ^{k}$ such that $\slice{Z}{\vec y} = \vec t_{\vec y} + L_0'$.
  Note that the assumption that $\slice{Z}{\vec x}$ is fully-dimensional implies that $L_0'$ is fully-dimensional. 
  Let $I \coloneqq \{ i \in [1,m] : L_i' \text{ is fully-dimensional} \}$.
  Then,
  \begin{align*}
    \univ{Z}(k,X)
        &=\set{ \vec x\in\pi(k,Z):\forall \vec t \in \ZZ^{k}, (\vec x, \vec t)\in Z \text{ implies } (\vec x, \vec t)\in \textstyle\bigcup_{i=1}^m X_i}\\
        &=\set{\vec x \in \pi(k,Z): \slice{Z}{\vec x} \subseteq \textstyle\bigcup\nolimits_{i=1}^m\slice{X_i}{\vec x}}
        \\
        &=\set{\vec x \in \pi(k,Z): \slice{Z}{\vec x} \subseteq \textstyle\bigcup\nolimits_{i \in I} \slice{X_i}{\vec x}}
        \hspace{3.3cm}\text{by \Cref{lem:inclusion_uofls_fulldim}}\\
        &=\set{\vec x \in \pi(k,Z): \slice{Z}{\vec x} = \textstyle\bigcup\nolimits_{i \in I} \slice{X_i}{\vec x}}
        \hspace{3.3cm}\text{since $X \subseteq Z$}
  \end{align*}
  Let $s \coloneqq \prod_{i \in I \cup \{0\}} 
  |\det(L_i')|$ (since we have a basis for each $L_i'$, computing $s$ requires polynomial time). By~\Cref{lem:lattice_contained_detZ}, $s \cdot \ZZ^{d-k} \subseteq L_i'$ for every $i \in I \cup \{0\}$.
  We now replay the proof of~\Cref{lemma:wPA-un-UXP} to derive that, for every $\vec x \in \ZZ^{d-k}$, 
  \[ 
    \slice{Z}{\vec x} = \textstyle\bigcup\nolimits_{i \in I} \slice{X_i}{\vec x} 
    \ \iff 
    \ 
    |\slice{Z}{\vec x} \cap [0,s)^{k}| = \sum_{J \subseteq I} (-1)^{|J|+1} |\bigcap_{j \in J} \slice{X}{\vec x} \cap [0,s)^k|.
  \]
  Consider $J \subseteq I$. If $\bigcap_{j \in J} \slice{X}{\vec x} \neq \emptyset$, i.e.~when $\vec x \in \bigcap_{j \in J} \pi(k,X_j)$, 
  by $s \cdot \ZZ^{d-k} \subseteq L_i'$ 
  we conclude that $|\bigcap_{j \in J} \slice{X}{\vec x} \cap [0,s)^k| = |\bigcap_{j \in J} L_i' \cap [0,s)^k|$.
  Similarly, $|\slice{Z}{\vec x} \cap [0,s)^{k}| = |L_0' \cap [0,s)^{k}|$. Define 
  \[
    N_0 \coloneqq |L_0' \cap [0,s)^{k}|,
    \ \
    N_J \coloneqq |\textstyle\bigcap_{i\in J}L_i'\cap[0,d)^k|,
    \ \ 
    T_J(\vec x) \coloneqq \indicator{}{ \vec x\in \textstyle\bigcap_{j\in J}\pi(k,X_j)}.
  \]
  Hence, $\slice{Z}{\vec x} = \textstyle\bigcup\nolimits_{i \in I} \slice{X_i}{\vec x}$ holds if and only if $N(\vec x) = 1$, where
  \begin{align*}
    N(\vec x) \coloneqq \sum_{J \subseteq I} \Big(T_J(\vec x) \cdot \frac{(-1)^{|J|+1} \cdot N_J}{N_0}\Big).
  \end{align*}
  Observe that the quantity $N(\vec x)$ ultimately depends on the values of $T_J(\vec x)$, which are in $\{0,1\}$.
  More precisely, given $f \colon 2^{I} \to \{0,1\}$, 
  define 
  \[
    N(f) \coloneqq \sum_{J \subseteq I} \Big(f(J) \cdot \frac{(-1)^{|J|+1} \cdot N_J}{N_0}\Big).
  \]
  Then, for all $\vec x \in \ZZ^{d-k}$ and $J \subseteq I$, given $f_{\vec x}(J) \coloneqq T_J(\vec x)$ we have $N(\vec x) = N(f_{\vec x})$.
  Let $\mathcal{F} \coloneqq \{f \colon 2^I \to \{0,1\} : N(f) = 1\}$.
  We have $N(\vec x) = 1$ if and only if $f_{\vec x} \in \mathcal{F}$;
  and so to summarise $\univ{Z}(k,X) = \bigcup_{f \in \mathcal{F}} \{\vec x \in \pi(k,Z) : f_{\vec x} = f\}$.
  
  Given~$f\in\mathcal{F}$ and $\vec x\in\pi(k,Z)$ we have
  \begin{align*}
    &f_{\vec x} = f\\
    \iff{}& \text{for all } J\subseteq I, T_J(\vec x) = f(J)
    &\text{by def.~of $f_{\vec x}$}\\
    \iff{}& \text{for all } J\subseteq I, \indicator{}{\vec x \in \textstyle\bigcap_{j\in J}\pi(k,X_j)}=f(J)
    &\text{by def.~of $T_J(\vec x)$}\\
    \iff{}& \text{for all } J \in \set{J \subseteq I : f(J)=1}, \vec x\in\bigcap\nolimits_{i\in J}\pi(k,X_i) \text{ and }\\
    &\text{for all } J\in\set{J \subseteq I:f(J)=0}, \vec x\notin\bigcap\nolimits_{i\in J}\pi(k,X_i)\\
    \iff{}& \vec x\in\Big(\bigcap_{J:f(J)=1}\bigcap_{j\in J}\pi(k,X_j)\Big)
        - \Big(\bigcup_{J:f(J)=0}\bigcap_{j\in J}\pi(k,X_j)\Big).
\end{align*}
This concludes the proof of the equivalence in the statement of the lemma, subject to the additional hypothesis~\eqref{eq:univ-proj-additional-hyp}. 

Before showing how to remove the hypothesis~\eqref{eq:univ-proj-additional-hyp}, we consider the second part of the lemma and study the complexity 
of computing the element $Y \in \dom(\dnf(\rho_{\SL}))$ that represents $\unprojrel{\rho_{\SL}(Z)}(k,\un(\rho_{\SL}(X)))$ (again assuming~\eqref{eq:univ-proj-additional-hyp}).
First of all not that w.l.o.g.~we can assume all elements of $\dom(\rho_{\SL})$ in $X$, and $Z$, to be different from $\varnothing$. Indeed, 
if $Z = \varnothing$ or $X = (\varnothing)$ then $Y$ can be set as $\varnothing$, and we can remove from $X$ every element equal to $\varnothing$ as this does not change the set~$\unprojrel{\rho_{\SL}(Z)}(k,\un(\rho_{\SL}(X)))$.
Referring to the objects in the first part of the proof, 
we see that computing bases for the lattices $L_i'$ ($i \in [0,m]$) can be done in polynomial time, by~\Cref{lem:slice_of_LS}. 
Given these bases it is trivial to compute the set $I$.
Then, by~\Cref{lemma:lattice-volume}, computing $N_0$ and $N_J$ for any $J \subseteq I$ also requires polynomial time (there are however $2^{|I|}$ many such $J$).
To compute the set $\mathcal{F}$ it suffices to list all $f \colon 2^I \to \{0,1\}$, compute $N(f)$ and check that this integer equals $1$. 
Hence, by definition of $N(f)$, we conclude that constructing $\mathcal{F}$ takes time $2^{2^{O(m)}} \poly(|X|,|Z|)$. 
Given $\mathcal{F}$, the element $Y$ is computed as 
\[
  \bigvee_{f\in\mathcal{F}}
            \Big(\Big(\bigwedge_{J :f(J)=1}\bigwedge_{j\in J}\proj(k,X_j)\Big)
            - \Big(\bigvee_{J :f(J)=0}\bigwedge_{j\in J}\proj(k,X_j)\Big)\Big).
\]
When $m$ is considered fixed, this corresponds to a constant number of applications of the operators $\lor$, $\land$, $-$ and $\pi$ to the sets $X_1,\dots,X_m$. Then, $Y$ can be computed in polynomial time 
thanks to~\Cref{lem:proj_SL_high} and~\Cref{lemma:Step1-implies-BooleanAlgebra}; where the latter holds for weak PA as we have already established~\Cref{step:basic_framework_ptime} of the framework.

What is missing is to get rid of the hypothesis~\eqref{eq:univ-proj-additional-hyp} without incurring an exponential blow-up with respect to $\max\{|X_1|,\dots,|X_m|,|Z|\}$. 
To this end, consider again shifted lattices~$Z,X_1,\dots,X_m$ (assumed for simplicity to be elements in~$\dom(\rho_{\SL})$) as in the first statement of the lemma (but now without assuming~\eqref{eq:univ-proj-additional-hyp}).
We will define shifted lattices $\widetilde{X}_1,\dots,\widetilde{X}_m,\widetilde{Z}$ such that
\begin{enumerate}[label=(\Alph*)]
  \item\label{it:A} $\widetilde{X} \coloneqq \bigcup_{j=1}^m \widetilde{X}_j \subseteq \widetilde{Z}$ and for every $\vec x \in \pi(k,\widetilde{Z})$, $\slice{\widetilde{Z}}{\vec x} \subseteq \ZZ^k$ is fully-dimensional (that is, the sets~$\widetilde{X}$ and $\widetilde{Z}$ satisfy~\eqref{eq:univ-proj-additional-hyp}),
  \item\label{it:B} 
  $\pi(k,\widetilde{X}_j) = \pi(k,X_j \cap Z)$ and 
  $\unprojrel{\widetilde{Z}}(k,\widetilde{X}) = \unprojrel{Z}(k,X)$.
\end{enumerate}
Thanks to~\Cref{it:A,it:B}, 
to compute a representation in $\dom(\dnf(\rho_{\SL}))$ of the set 
$\unprojrel{\rho_{\SL}(Z)}(k,\un(\rho_{\SL}(X)))$,
it suffices to apply the equivalence in the statement of the lemma on 
the fully-dimensional sets $\widetilde{X}_1,\dots,\widetilde{X}_m,\widetilde{Z}$.

Let us define $\widetilde{X}_1,\dots,\widetilde{X}_m$ and $\widetilde{Z}$.
First, apply~\Cref{lem:slice_of_LS} to obtain a basis for a lattice $L \subseteq \ZZ^{d-k}$ such that for all $\vec x \in \pi(k,Z)$ there is $\vec t_{\vec x} \in \ZZ^k$ such that $\slice{Z}{\vec x} = \vec t_{\vec x} + L$.
We compute a basis for the lattice $L^\perp$ orthogonal to $L$, according to~\Cref{lem:orto-lattice-properties}.
Recall that $L + L^\perp$ is fully-dimensional. We define:
\[ 
  \widetilde{Z} \coloneqq Z + \{\vec 0\} \times L^\perp,
  \qquad 
  \widetilde{X}_j \coloneqq (X_j  \cap Z) + \{\vec 0\} \times L^\perp,
  \qquad 
  \widetilde{X} \coloneqq \bigcup_{j=1}^m \widetilde{X}_j,
\]
where $\vec 0$ stands here for the zero vector of $\ZZ^{d-k}$.
By definition, $\widetilde{Z}$ and all $\widetilde{X}_j$ are shifted lattices, and moreover $\widetilde{X} \subseteq \widetilde{Z}$.

Observe that for every $ \vec x \in\ZZ^{n-k}$ and every $A\subseteq \ZZ^d$,
\begin{align*}
    \slice{(A+\set{\vec 0} \times L^\perp)}{\vec x}
    &=\set{\vec t \in \ZZ^k : (\vec x, \vec t)\in A + \set{\vec 0} \times L^\perp}\\
    & =\set{ \vec t \in \ZZ^k: \vec t \in \slice{A}{\vec x} + L^\perp}\\
    & =\slice{A}{\vec x}+L^\perp.
\end{align*}
Hence, for every $\vec x \in \pi(k,Z)$ we have $\slice{\widetilde{Z}}{\vec x} = \slice{Z}{\vec x} + L^\perp = \vec t_{\vec x} + L + L^\perp$, and therefore $\slice{\widetilde{Z}}{\vec x}$ is fully-dimensional. This establishes~Item~\ref{it:A}.

Moving towards~Item~\ref{it:B}, we observe that for every $A \subseteq \ZZ^d$, 
\begin{align*}
  \pi(A+\set{\vec 0} \times L^\perp)
  &=\set{ \vec x\in\ZZ^{d-k} : \slice{(A+\set{\vec 0} \times L^\perp)}{\vec x}\neq \emptyset}\\
  &=\set{\vec x\in\ZZ^{d-k}:\slice{A}{\vec x}+L^\perp\neq\emptyset}\\
  &=\set{\vec x\in\ZZ^{d-k}:\slice{A}{\vec x}\neq\emptyset}
      &\text{as }L^\perp\neq\emptyset\\
  &=\pi(k,A).
\end{align*}
Therefore, $\pi(k,\widetilde{X}_j) = \pi(k,X_j \cap Z)$ for every $j \in [1,m]$. Lastly, let us show the equivalence $\univ{\widetilde{Z}}(k,\widetilde{X}) = \univ{Z}(k,X)$. Recall that for $\vec x \in \pi(k,Z)=\pi(k,\widetilde{Z})$, we have $\slice{Z}{\vec x}= \vec t_{\vec x}+L$. 
Furthermore,
$X\cap Z\subseteq Z$, so $\slice{(X\cap Z)}{\vec x}\subseteq \slice{Z}{ \vec x}$ and we can write
$\slice{(X\cap Z)}{\vec x}= \vec t_{\vec x} + A_{\vec x}$ for some set $A_{\vec x}\subseteq L$. We have
\begin{align*}
  \univ{\widetilde{Z}}(k,\widetilde{X}) ={} & \set{\vec x\in\pi(k,\widetilde{Z}):\slice{\widetilde{Z}}{\vec x}\subseteq\slice{\widetilde{X}}{ \vec x}}\\
  ={} & \set{\vec x \in\pi(k,Z) : \slice{Z}{\vec x}+L^\perp\subseteq\slice{(X\cap Z)}{\vec x}+L^\perp}\\
  ={} & \set{\vec x\in\pi(k,Z): \vec t_{\vec x}+L+L^\perp\subseteq \vec t_{\vec x}+A_{\vec x}+L^\perp}\\
  ={} & \set{\vec x\in\pi(k,Z): \vec t_{\vec x} + L\subseteq \vec t_{\vec x}+A_{\vec x}}
          &&\text{see below}\\
  ={} & \set{\vec x\in\pi(k,Z):\slice{Z}{\vec x}\subseteq\slice{(X\cap Z)}{\vec x}}\\
  ={} & \set{\vec x\in\pi(k,Z):\slice{Z}{\vec x}\subseteq\slice{X}{\vec x}}\\
  ={} & \univ{Z}(k,X).
\end{align*}
Above, we have used the fact that $\vec t_{\vec x}+L+L^\perp\subseteq \vec t_{\vec x}+A_{\vec x}+L^\perp$ if and only if $\vec t_{\vec x}+L\subseteq \vec t_{\vec x}+A_{\vec x}$; which is trivially equivalent to 
\[ 
  L+L^\perp\subseteq A_{\vec x}+L^\perp \text{ if and only if } L\subseteq A_{\vec x}.
\]
The right to left direction of this double implication is straightforward. 
For the other direction, suppose $L+L^\perp\subseteq A_{\vec x}+L^\perp$
and pick $\vec u \in L$. Since $\vec 0 \in L^\perp$, we have $\vec u \in L+L^\perp \subseteq A_{\vec x}+L^\perp$, 
so there are $\vec v_1 \in A_{\vec x}$ and $\vec v_2 \in L^\perp$ such that $\vec u = \vec v_1 + \vec v_2$. 
By $A_{\vec x} \subseteq L$ we get $\vec u - \vec v_1 \in L$. 
From $\vec u - \vec v_1 = \vec v_2$ and the fact that $L$ and $L^\perp$ are orthogonal lattices we conclude that $\vec v_2 = \vec 0$, 
and thus $\vec u \in A_{\vec x}$.
This completes the proof of~Item~\ref{it:B}.

To conclude, let us discuss how to compute representations of $\widetilde{Z}$ and $\widetilde{X}_j$ in $\dom(\rho_{\SL})$ in polynomial time. Starting from the representations $(d,\vec v_1,\dots,\vec v_\ell)$ and $(d, \vec w_0,\dots, \vec w_{r_j})$ of $Z$ and $X_j$, 
respectively, 
we compute
bases $\vec u_0,\dots,\vec u_i$ and $\vec u_{i+1},\dots,\vec u_{d}$ of the lattices $L$ and $L^\perp$, respectively, in polynomial time by relying on~\Cref{lem:slice_of_LS} and~\Cref{lem:orto-lattice-properties}. 
So, $\widetilde{Z} = \vec v_0 + \Span_{\ZZ}\{\vec v_0,\dots,\vec v_\ell,\vec 0 \times \vec u_{i+1},\dots, \vec 0 \times \vec u_{d}\}$
and to find one of its representations in $\dom(\rho_{\SL})$ it suffices to put the matrix $[\vec v_1 \mid \dots \mid \vec v_\ell \mid \vec 0 \times \vec u_{i+1} \mid \dots \mid \vec 0 \times \vec u_{d}]$
in Hermite normal form using~\Cref{prop:hnf_snf}, to then take as period all its non-zero columns; and $\vec v_0$ as its base point.
The computation of $\widetilde{X}_j$ is analogous, but we must first compute a representation for $X_j \land Z$, which can be done in polynomial time by~\Cref{lem:basic_framework_ptime_sl}.
\end{proof}

As the parameter $\len(\parone)$ fixes 
the number of shifted lattices of an element in $\dom(\un(\rho_{\SL}))$, \Cref{lem:univ_proj_UofSL_relative_high} establishes Items~\ref{step:beforeReducPiUnivToSimpleCasesUxp:Item2}
and~\ref{step:beforeReducPiUnivToSimpleCasesUxp:Item4} 
of~\Cref{step:beforeReducPiUnivToSimpleCasesUxp}.
Then, by appealing to~\Cref{lemma:step-one-two-imply-FO-tract} and~\Cref{theorem:PuttingAllTogether}, 
the Lemmas~\ref{lem:basic_framework_ptime_sl}, \ref{lemma:wPA-un-UXP}, 
\ref{lem:proj_SL_high} and \ref{lem:univ_proj_UofSL_relative_high}
yield the main result of the section.

\begin{thm}
    Fix $k \in \NN$. The $k$ negations satisfiability
    problem for weak Presburger arithmetic is decidable in polynomial time.
\end{thm}

\noindent
By~\Cref{theorem:PuttingAllTogether}, we also conclude that there is a
polynomial-time procedure that given a weak PA formula $\Phi$ only having $k$
negations, returns an element of $\dom(\dnf(\rho_{\SL}))$ representing the set of
solutions~$\sem{\Phi}_{\mathcal{Z}}$.

\section{Conclusion}\label{sec:final-remarks}

We developed a framework to establish polynomial-time decidability of fixed
negation sentences of first-order theories whose signatures enjoy certain
fixed-parameter tractability properties. A key feature of the framework is that it treats
complementation in a general way, and considers universal projection as a
first-class citizen. Note that, a priori, the latter operation might be easier
than the former to decide, as shown for instance in~\cite{ChistikovH17}.

We instantiated our framework to show that the fixed negation satisfiability
problems for weak linear real arithmetic and weak Presburger arithmetic are
decidable in PTIME. This is in sharp contrast with standard Presburger
arithmetic, which is known to be NP-hard even when the Boolean structure and the
number of variables in the formula is fixed~\cite{NguyenP22}. We believe that
our framework also provides a sensible approach to study fixed negation
fragments of FO extensions of, e.g., certain abstract domains. 
An interesting extension of our running example in \Cref{section:fo-framework} 
to further test our framework
is \emph{octagon arithmetic}, where inequalities take the form~$\pm x \pm y \leq
c$, with $c \in \ZZ$~\cite{Mine06}. Over the integers, the full first-order
theory of octagon arithmetic is known to be PSPACE-complete~\cite{BenediktCM23}.
More generally, as the various requirements to instantiate the framework
relate to natural computational problems (deciding inclusion and computing projections), 
we are confident that our framework can also be appied to logical theories
outside the world of arithmetic.

\section*{Acknowledgement}
\begin{minipage}{0.88\linewidth}
  This work is part of a project that has received funding from the European
  Research Council (ERC) under the European Union's Horizon 2020 research and
  innovation programme (GA 852769, ARiAT). Alessio Mansutti is co-funded by
  the European Union (GA 101154447), MICIU/AEI (GA~CEX2024-001471-M and
  PID2022-138072OB-I00) and FEDER, UE. Views and opinions expressed are however
  those of the author(s) only and do not necessarily reflect those of the
  European Union or European Commission. Neither the European Union nor the
  granting authority can be held responsible for them.
\end{minipage}%
\begin{minipage}{0.12\linewidth}
  \flushright
  \includegraphics[scale=0.1]{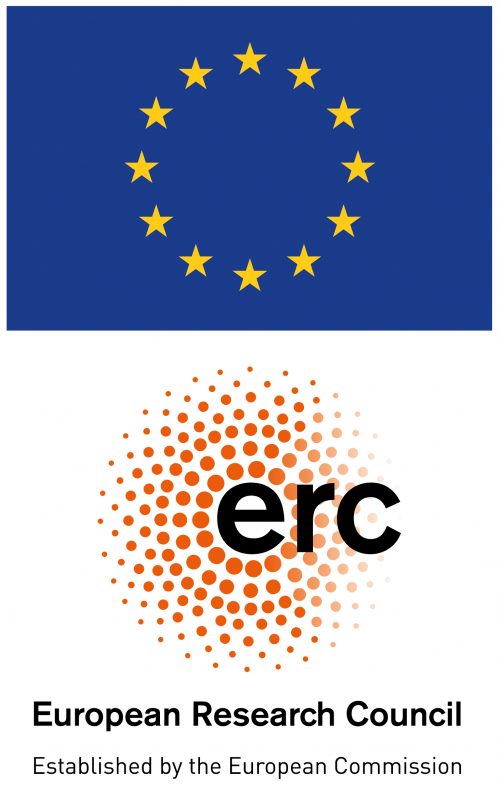}
  \\
  ~
\end{minipage}

\bibliographystyle{alphaurl}
\bibliography{biblio}




\end{document}